\documentclass[10pt,reqno]{amsart}
\usepackage{graphicx,multicol}
\usepackage{multirow} 
\usepackage{indentfirst,csquotes}
 \usepackage{array} 

\usepackage[caption=false]{subfig}
\usepackage{amssymb,amsthm,amsmath}
\usepackage{xcolor,paralist,hyperref,fancyhdr,etoolbox}
\newtheorem{theorem}{Theorem}[]
\newtheorem{definition}[theorem]{Definition}

\newtheorem{lemma}[theorem]{Lemma}

\newtheorem{corollary}[theorem]{Corollary}

\usepackage[round]{natbib} 
\newtheorem{rmk}{\bf Remark}
\usepackage[plain,noend]{algorithm2e}

\usepackage[utf8]{inputenc}
\numberwithin{equation}{section}
\usepackage{breakurl}

\hypersetup{ colorlinks=true, linkcolor=blue, filecolor=blue, urlcolor=blue,citecolor = blue }

\usepackage{lipsum}
\usepackage{multicol}

\usepackage{listings}
\begin{document}


\title[]{Distribution-Free test for Changepoint Detection in Angular Mean Direction:  Application in  Finance} 

\maketitle
\begin{center}
\author{Surojit Biswas} \\ 
\email{surojit23$@$iitkgp.ac.in}\\ and\\
\author{Buddhananda Banerjee\footnotetext{corresponding author\/: bbanerjee@maths.iitkgp.ac.in}}\\
\email{bbanerjee@maths.iitkgp.ac.in}\\
\vspace{0.3cm}
\address{Department of Mathematics\\  Indian Institute of Technology Kharagpur, India-$721302$}

\end{center}

\let\thefootnote\relax

\begin{abstract}
In this paper, we propose a distribution-free test for detecting changepoint in the mean direction of angular data. The uncertainty in angular measurements is quantified through the \textit{square of an angle}, derived from the intrinsic geometry of the torus. It is established that, under the null hypothesis, the test statistic distributionally converges to the Kolmogorov distribution, while under the alternative hypothesis, both the consistency of the test and the asymptotic properties of the changepoint estimator are established. Through extensive simulations, we compare the empirical performance of the proposed method with two existing approaches for angular data and further benchmark it against a test based on the circular arc length distance. Finally, we demonstrate the practical utility of our approach by analyzing the timestamps of extreme events in Bitcoin, Ethereum, and Gold price datasets, where the continuous, high-frequency nature of the data is modeled in the circular framework.\\

\keywords{Keywords: Angular Data;  Area element;
Changepoint ; CUSUM ; Finance}
\end{abstract} 

\bigskip
\setlength{\columnsep}{1.5cm} 

\section{Introduction}

Circular or angular data are often encountered in various fields, including finance, where they represent the time of occurrence of cryptocurrency stock maxima or minima, and bioinformatics, where dihedral angles in protein structures are measured. In meteorology, they include wind and sea wave directions, and in climatology, the path of a cyclone. These measurements exhibit circularity, in contrast to linear data, which possess a distinct starting and ending point. Circular data exhibit cyclical behavior, where they repeat their values after completing a full revolution. This cycle often starts and ends at the same point, which is commonly represented as 0 and $2\pi$. The cyclical form of this data presents specific difficulties for statistical analysis and visualization, as conventional methods may not sufficiently consider the periodic characteristics of the data. In order to precisely analyze and understand circular data, specialized approaches such as circular mean, circular variance, and directional statistics are employed. These techniques guarantee that the obtained results are both significant and representative of the underlying patterns. For a more comprehensive understanding of circular data, readers can refer to  \cite{kumar2025directional} and \cite{mardia2000directional}.

Changepoint analysis constitutes a fundamental statistical methodology for identifying unanticipated changes within sequentially ordered observations. Such structural changes may be attributable to shifts in distributional parameters within a given family or, in some cases, a complete change in distributional framework. The presence of changepoints can exert a pronounced influence on standard inferential procedures, potentially rendering conventional analyses invalid or misleading. Accordingly, the principal objective of changepoint analysis is to formally assess, via suitable statistical tests, whether one or more changepoints are present in a data sequence.

A wide range of research has been carried out on the study of changepoint detection for sequences of real-valued random variables. Many methodological advances in this area can be found in works such as \cite{fearnhead2019changepoint, haynes2017computationally, Hovarth_1999, Antoch_1997, Davis_1995}. Beyond the univariate setting, considerable research has also been directed toward changepoint problems involving multivariate or vector-valued observations, with important contributions reported in \cite{gupta2022real, Kirch_2014, Kokoszka_2000, shao2010testing, anastasiou2023generalized, pishchagina2023online, liang2021joint}.
In addition, recent developments have extended changepoint methodologies to functional data, further enriching the theoretical landscape; key references include \cite{horvath2012inference, banerjee2018more, Horman_2010, banerjee2020data}. For a discussion of algorithmic strategies and modern computational approaches to changepoint detection in time series, the survey presented in \cite{gupta2024comprehensive} offers a thorough and insightful overview.

In contrast, the literature addressing changepoint problems for angular (circular) data remains comparatively limited. In the context of angular observations, changepoints may manifest as shifts in the mean direction, changes in concentration parameters, or both, even the distribution. Early contributions include the rank-based methodology proposed by \cite{lombard1986change}, which targets changes in both mean direction and concentration. Subsequently, \cite{grabovsky2001change} investigated modifications in concentration using a refined CUSUM approach. Likelihood-based procedures have also been explored: \cite{ghosh1999change} developed methods specifically for von Mises distributions, while \cite{sengupta2008likelihood} introduced the likelihood-integrated approach. More generally, graph-based changepoint techniques for non-Euclidean data, such as the method proposed by \cite{chen2015graph}, are applicable to circular settings as well.
Collectively, these contributions represent important progress toward the development of rigorous statistical tools for detecting changepoints in angular data; nevertheless, considerable scope remains for comprehensive theoretical and methodological advancement in this domain.\\
\noindent

\paragraph{\textbf{Daily timestamp of extreme values in cryptocurrency and gold:}}

Cryptocurrency is a decentralized digital asset enabling peer-to-peer transactions without central authorities. Using blockchain and cryptographic techniques, it ensures secure, transparent exchanges. Prominent examples like Bitcoin and Ethereum highlight its potential, though high volatility and speculation present challenges for financial modeling, risk management, and regulatory frameworks. Alongside cryptocurrencies, gold serves as a globally recognized safe-haven asset whose price dynamics are influenced by macroeconomic uncertainty, monetary policy, and geopolitical stress. Extreme gold price movements often reflect shifts in global liquidity preferences and risk hedging behavior, complementing cryptocurrency market responses.

Extreme price behaviors in these assets often indicate imbalances in supply and demand, presenting arbitrage opportunities across exchanges or trading venues. The timing of these extremes can influence market sentiment. A low price early in the day may trigger caution, while a late-day low could lead traders to reassess positions before market close. Similarly, an early high may fuel optimism, whereas a late high might prompt profit-taking. Identifying when these fluctuations occur helps pinpoint key support and resistance levels. If extreme price events consistently appear at specific times, traders may adjust their strategies accordingly. Liquidity also plays a critical role in price extremes as low-liquidity periods can lead to sharp movements, while high liquidity conditions tend to yield more stable prices that reflect broader consensus. Understanding the timing of extreme price movements provides valuable insight into market dynamics. Thus, detecting changes in the distribution of timestamps associated with these extremes using changepoint analysis is crucial for enhancing trading strategies and market interpretation.

Changepoint research on the timing of extreme values in cryptocurrency stocks is important for multiple reasons. It aids in the detection of fluctuations in market dynamics, uncovering alterations in trade patterns and external effects. This study strengthens trading methods by aligning them with current market circumstances, improves predictive models for more accurate forecasting, and assists in risk management by emphasizing moments of heightened volatility. Additionally, it offers valuable information on traders' sentiments and enhances market efficiency by identifying and adjusting to patterns.

In summary, changepoint analysis is a powerful tool that enables traders and analysts to enhance decision-making and optimize strategic responses in a constantly evolving market environment. However, the continuous availability of data throughout the year naturally induces a cyclical pattern in the timestamps of extreme values in cryptocurrency and gold markets. Consequently, conventional changepoint methods designed for linear data are not suitable in this context. This limitation motivates the development of specialized approaches for detecting changepoints in circular (angular) domains. In this study, we focus on timestamps of extreme values collected from two major cryptocurrencies, Bitcoin and Ethereum, as well as from historical  (gold) pricing data.

\begin{itemize}
	\item \textbf{Bitcoin:} Bitcoin, launched in 2009 by the pseudonymous Satoshi Nakamoto, is the first and most prominent cryptocurrency. It operates on a decentralized peer-to-peer network using blockchain technology to ensure secure, transparent, and immutable transactions without intermediaries like banks. With a capped supply of 21 million coins, Bitcoin is often viewed as a deflationary asset and dubbed ``digital gold." Its Proof of Work consensus mechanism secures the network and validates transactions. Bitcoin's rise has not only positioned it as a store of value and medium of exchange but also inspired the development of numerous other cryptocurrencies and blockchain-based innovations.

	\item  \textbf{Ethereum:} Introduced by Vitalik Buterin in 2015, Ethereum extends blockchain use beyond digital currency by enabling smart contracts and decentralized applications (dApps). Unlike Bitcoin, Ethereum supports programmable, trustless transactions without intermediaries. Its native token, Ether (ETH), powers the network and incentivizes users. The Ethereum Virtual Machine (EVM) allows developers to build complex dApps across various sectors. Ongoing upgrades, including Ethereum 2.0 and the shift to Proof of Stake, aim to improve scalability, security, and sustainability, solidifying Ethereum’s role in blockchain innovation.

    \item \textbf{Gold}:
Gold is a globally recognized safe-haven asset actively traded in international commodity markets. Unlike cryptocurrencies, gold prices are influenced by macroeconomic uncertainty, inflation expectations, monetary policy decisions, geopolitical tension, and institutional hedging demand. Extreme price events in gold often arise during periods of systemic stress or major economic announcements, making the timing of such events particularly relevant. This complementary behavior relative to cryptocurrency markets provides an additional perspective for changepoint detection in circular timestamp domains.
	
\end{itemize}

\noindent
\textbf{Angular representation of timestamp of extreme values:} We have obtained per-minute historical datasets for Bitcoin, Ethereum, and gold. For the cryptocurrencies, each dataset has columns for the Unix timestamp, date, symbol, opening price, highest price, lowest price, closing price, volume (in crypto), and volume (in the currency used as a base). In contrast, the Gold price dataset contains the date, opening price, highest price, lowest price, closing price, and volume information across multiple timeframes.
We are just concerned with three specific columns: date, highest price (maximum), and lowest price (minimum).
The date column denotes the timestamp, the high column signifies the maximum price observed during the specified time, and the low column indicates the minimum price observed during the specified period. 

Let us consider the data linked to the highest price. Given that our data was collected at per-minute intervals, we obtained 60 data points every hour for each column.
To reduce this data to a single observation each hour, we will select the highest value observed within each hour. As a consequence, there are 24 recorded instances every day. Subsequently, we ascertain the highest value among these 24 hourly measurements to obtain a solitary observation for each day.
To obtain the angle (in radians) corresponding to the timestamp of the daily maximum, we apply the following formula:

\begin{equation}
     \theta= \left[\frac{\displaystyle \arg \max_{1\leq j \leq 24}  \left(\max_{0\leq i \leq 59} v_{i+60(j-1)}\right)}{24\times 60} \times 2\pi\right]\in [0,2\pi),
     \label{data_proc_formula}
\end{equation}
where $v$ is \textit{the highest price} of each minute. If we replace \textit{max }by \textit{min} in the above formula, we will get the angular data corresponding to the \textit{lowest price}. 

The timestamps corresponding to these occurrences are transformed into angular data using the formula in Equation-\ref{data_proc_formula} and visualized using rose plots. The rose plot of the timestamps of occurrence of the lowest price is shown for Figure-\ref{rose_plot_bit_eth} (A) the Bitcoin dataset, (B) the Ethereum dataset, and (C) the Gold dataset, while panel (D) presents the corresponding rose plot for the highest price in the Gold dataset.


\begin{figure}[h!]
	\centering

	\subfloat[ ]{%
		{\includegraphics[trim= 20 20 20 20, clip,width=0.25\textwidth, height=0.25\textwidth]{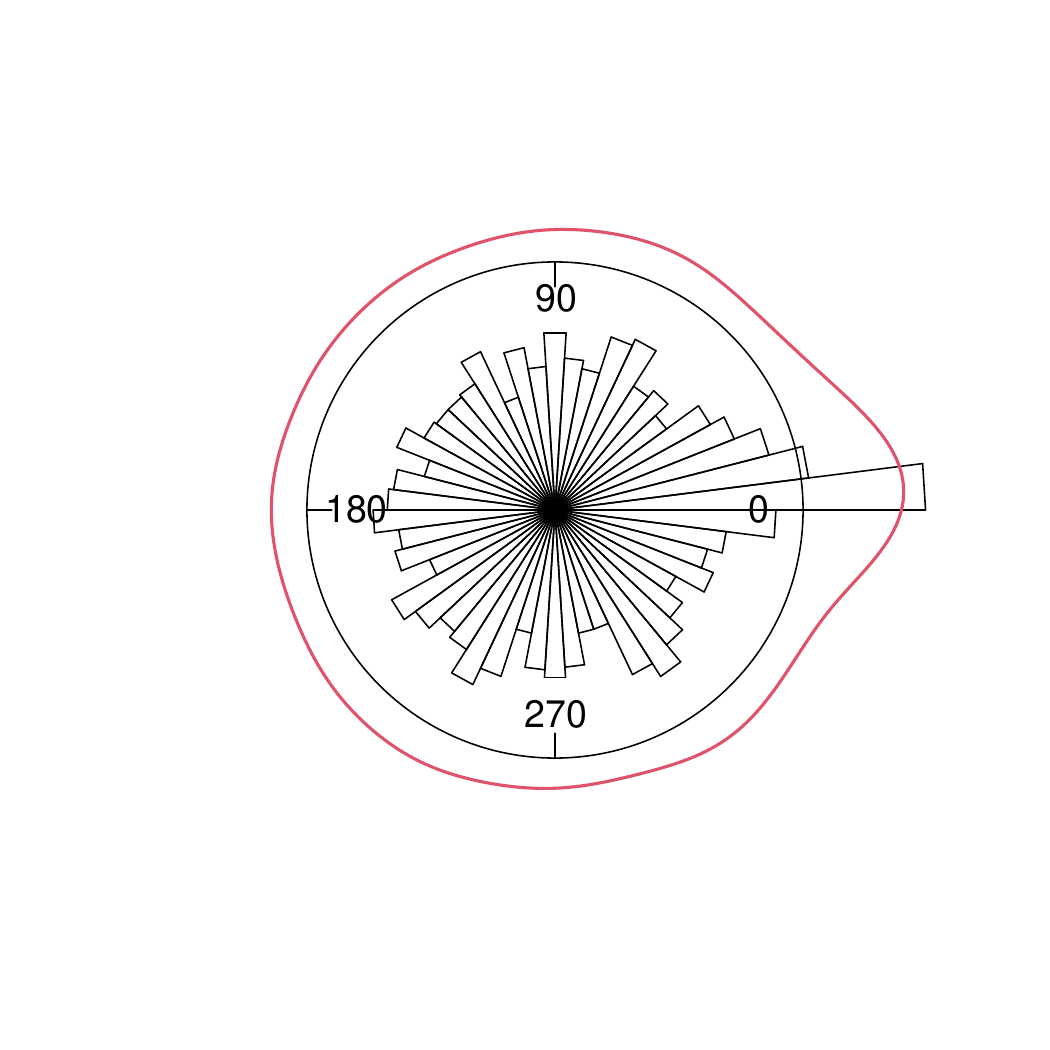}}}
	\subfloat[ ]{%
		{\includegraphics[trim= 20 20 20 20, clip,width=0.25\textwidth, height=0.25\textwidth]{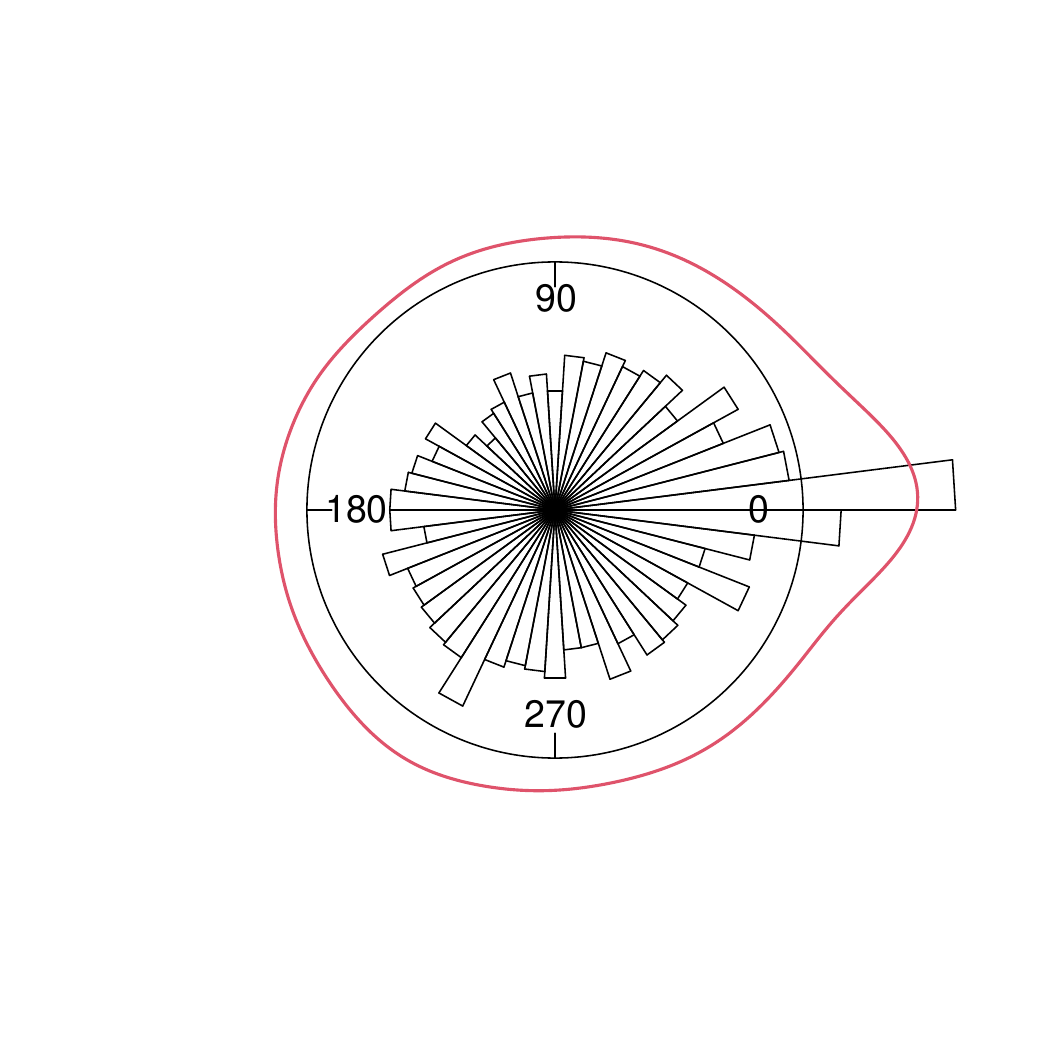}}}
        \subfloat[ ]{%
		{\includegraphics[trim= 20 20 20 20, clip,width=0.25\textwidth, height=0.25\textwidth]{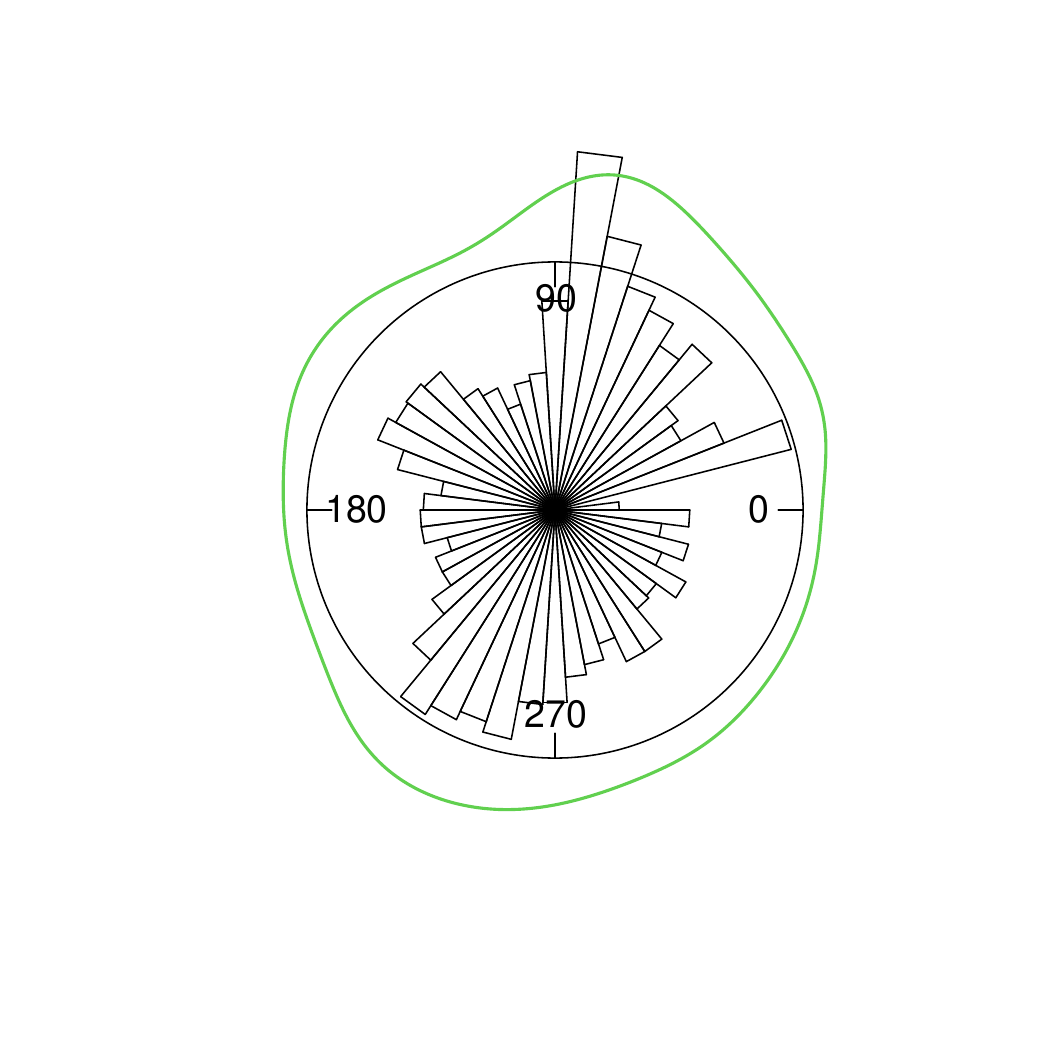}}}
        \subfloat[ ]{%
		{\includegraphics[trim= 20 20 20 20, clip,width=0.25\textwidth, height=0.25\textwidth]{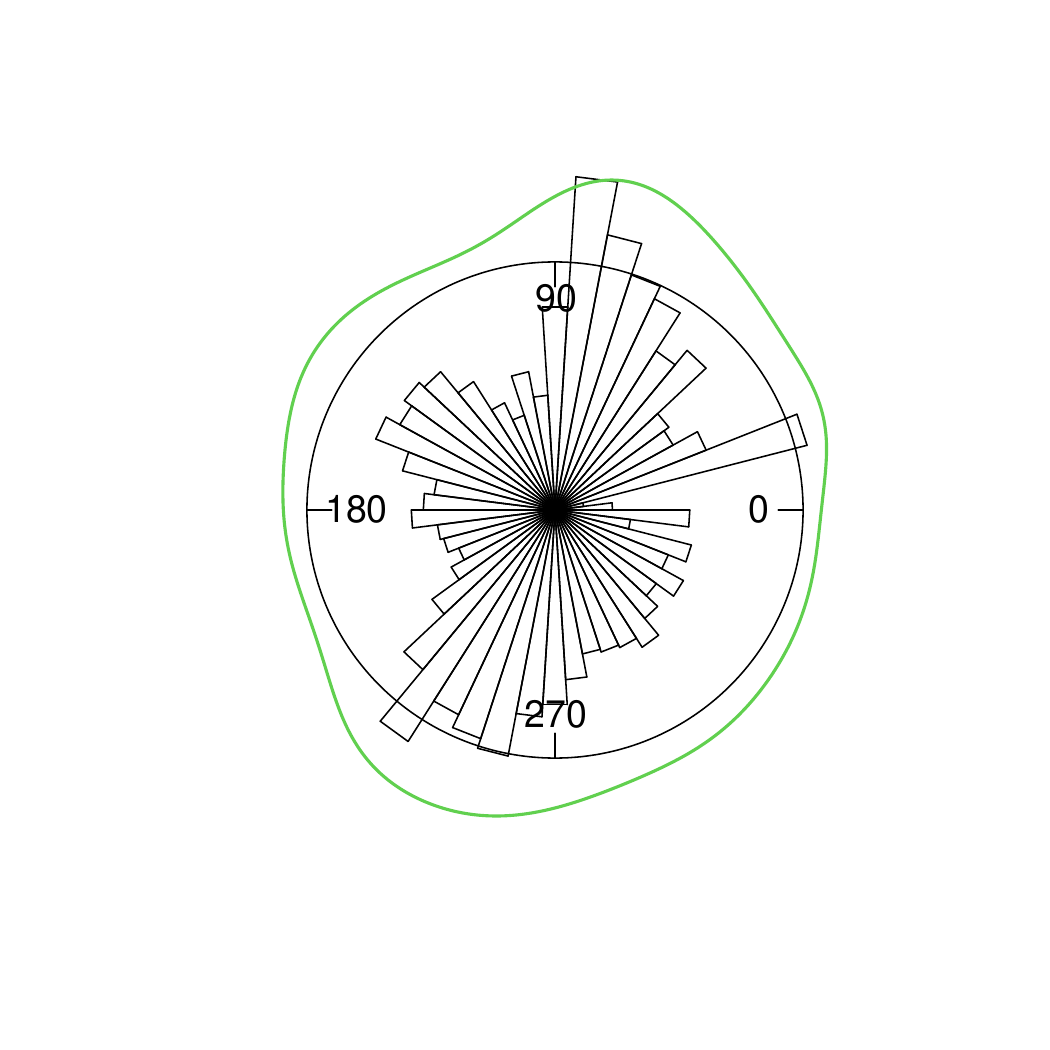}}}
	\caption{Rose plot of the timestamps of occurrence of the \textbf{lowest }price of (A) Bitcoin price dataset; (B) Ethereum price dataset, and (C) Gold price dataset, on the other hand, (D) is the same for the \textbf{highest} price of the Gold price dataset.}
	\label{rose_plot_bit_eth}
\end{figure}


\textit{The main contribution of this paper is the development of a geometry-driven, distribution-free  test for changepoint in  angular mean direction. 
Unlike general-purpose nonparametric approaches 
\cite[see][]{matteson2014nonparametric} or kernel-based methods for linear data
\citep[see][]{harchaoui2007retrospective,arlot2019kernel}, our procedure uses the intrinsic geometry of the torus.  We use a new concept of the ``square of an angle" proposed by \cite{biswas2025semi}, and introduce a method for detecting the changepoint in the mean direction of angular data. In addition, we derive the asymptotic pivotal distribution of the test statistic assuming the null hypothesis with no change, and it is found to be the Kolmogorov distribution. On the other hand, we have shown that the proposed test is consistent under the alternative hypothesis. In the subsequent sections, we used the term \textit{``Square Angle Mean Change (SAMC)"} to refer to this particular test. A comprehensive simulation analysis demonstrates that the SAMC test exhibits significant statistical power, and it shows the superiority over the baseline method constructed using the circular arc length. We also compare the proposed test with two existing tests by \cite{lombard1986change} and \cite{chen2015graph}. The applicability of the proposed SAMC test is demonstrated using the timestamps of two popular cryptocurrency datasets and the Gold price dataset. Although \cite{fernandez2025multivariate} discussed model fitting for cryptocurrency data, no changepoint analysis has been conducted on either cryptocurrency or Gold price datasets. To the best of our knowledge, this is the first study in the literature to address changepoint analysis for these datasets.}

The work is organized as follows.
In Section-\ref{square_angle}, we revisit the definition of the square of an angle and establish its key properties. Section-\ref{changepoint detection} demonstrates how this construction can be utilized to develop a distribution-free test for detecting changes in the mean direction of circular/angular data. The performance of the proposed method is evaluated in Section-\ref{simulation} through an extensive simulation study, using data generated from two widely used circular distributions: the von Mises and the Wrapped Cauchy. Section-\ref{comarison_sec} presents a comparative analysis between the proposed test with two existing tests and a benchmark test based on circular arc length.
In Section-\ref{data_analysis}, we apply the proposed method to real data, specifically the timestamps of extreme values from two major cryptocurrencies, to illustrate its practical relevance. Section-\ref{future work} discusses the limitations of the proposed framework and outlines possible directions for future research. Finally, Section-\ref{conclusion} concludes the paper, followed by the technical proofs of Lemma-\ref{lemma1} (see Appendix-\ref{lemma proof}), Lemma-\ref{bbridge_lemma} (see Appendix-\ref{null_proof}), and Corollary-\ref{consistency_corr} (see Appendix-\ref{alt_consistency_proof}).

\section{Square of an angle}
\label{square_angle} 

The rest of our work will be based on the curved torus defined by the parametric equation 
\begin{equation}
  X(\phi,\theta)=\{  (R+r\cos{\theta})\cos{\phi}, (R+r\cos{\theta})\sin{\phi}, r\sin{\theta} \}\subset \mathbb{R}^3, 
  \label{torus para equn}
\end{equation}
with the parameter space $\{ 
 (\phi,\theta):0<\phi,\theta<2\pi\}= \mathbb{S}_1 \times \mathbb{S}_1,$ known as  the flat torus. 
In Equation-\ref{torus para equn}, \(R\) denotes the major radius of the torus, that is,
the distance from the axis of revolution to the center of the generating
circle, while \(r\) denotes the minor radius, or the radius of the
generating circular cross-section. The coordinate \(\phi\) is the
longitudinal angle around the axis of revolution and \(\theta\) is the
meridional angle around the tube. For a regular embedded ring torus,
\(R>r>0\).

The induced area element of the curved torus (Equation-\ref{torus para equn})  is given by
\begin{equation}
   dA (\phi, \theta)=r(R+r\cos{\theta})~d\phi~d\theta.
    \label{torus area element}
\end{equation}
See \cite{biswas2025semi} for a detailed calculation of the area element using the tools from differential geometry.

Let $ \phi$ and $ \theta$ be the horizontal and vertical angles of a torus, respectively, where $ \phi$ and $ \theta$ are in the range of $ [0,2\pi)$. To start, we establish the area bounded by the coordinates $(0,0)$ and $(\phi,\theta)$ on the surface of a curved torus.  Consider two points on a flat torus, denoted as $(0,0)$ and $(\phi,\theta)$. The flat torus is defined by the interval $[0,2\pi)$ in both the horizontal and vertical directions. The area between these two points, when mapped onto the surface of the curved torus with horizontal radius $R$ and vertical radius $r$, can be calculated using Equation-\ref{torus area element}.  It is important to observe that when considering two locations $(0,0)$ and $(\phi,\theta)$ on a flat torus, the surface of the curved torus is divided into four subsets that are both mutually exclusive and exhaustive. The mapping above decomposition $\mathbb{T}_1:=[0, \phi]\times [0,\theta]$, $\mathbb{T}_2:=[\phi,2\pi]\times [0,\theta] $,  $\mathbb{T}_3:=[0,\phi]\times [\theta,2\pi]$ and $\mathbb{T}_4:=(\phi,2\pi] \times [\theta,2\pi]$ are represented on the surface of a curve torus (Equation-\ref{torus para equn}) in Figure-\ref{area_plot}. The corresponding areas on the surface of the curved torus are denoted as $A_1$, $A_2$, $A_3$, and $A_4$, respectively. \cite{biswas2025semi} have calculated of the areas $A_i=\displaystyle\iint_{\mathbb{T}_i}dA(s,t) ~ \text{for} ~i=1,2,3,4$ using $dA (\phi, \theta)=rR\left(1+\frac{r}{R}\cos{\theta}\right) d\phi ~d\theta$ from Equation-\ref{torus area element} .

Analogous to the notion of circular distance - which is the length of the smaller arc between two angles, and the notion of geodesic distance on a surface - which is the length of the shortest path joining two points on the surface, \cite{biswas2025semi} have defined the \textit{proportionate area} included between these points $(0,0)$, $(\phi,\theta)$ as given   Definition-\ref{torus area of a segemnt} and 
\textit{the square of an angle} $\theta$ as given in Definition-\ref{zero centered circle area of a segemnt} below. 
\begin{definition}
	The \textit{proportionate area included between the   $(0,0)$, and $(\phi,\theta)$ } is defined as 
	
	$$A_T\left[(0,0),(\phi,\theta)\right]=\frac{\min \{A_1,A_2,A_3,A_4\}}{4\pi^2rR}.
	\label{torus area of a segemnt}
	$$
\end{definition}

\begin{definition} \textit{The square of an angle} $\theta$ is defined as
	$$A_C^{(0)}(\theta)=A_T^{(0)}(\theta,\theta)=A_T\left[(0,0),(\theta,\theta)\right], \mbox{~where~} \frac{r}{R}=1.$$
	\label{zero centered circle area of a segemnt}
\end{definition}

\begin{lemma}
	
	 $A_C^{(0)}(\theta)=\begin{cases}
			(2\pi)^{-2}~ \theta\left[\theta +  \sin{\theta} \right] & \text{~if~} 0<\theta\leq \pi\\
			(2\pi)^{-2}~ (2\pi-\theta)\left[2\pi-(\theta+\sin{\theta})\right] & \text{~if~} \pi<\theta\leq 2\pi,
		\end{cases}$
	\label{lemma1}
\end{lemma}
\begin{proof}
	For proof see  Appendix-\ref{lemma proof}
\end{proof}

\begin{figure}[h!]
	\centering
	\subfloat[ ]{%
		{\includegraphics[trim= 0 0 0 0, clip, width=0.45\textwidth, height=0.35\textwidth]{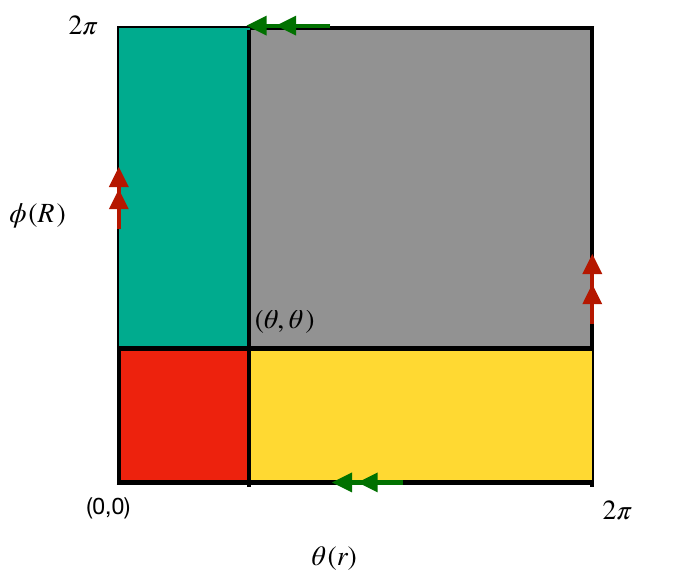}}}
	\subfloat[]{%
		{\includegraphics[trim= 70 70 70 70, clip, width=0.5\textwidth, height=0.4\textwidth]{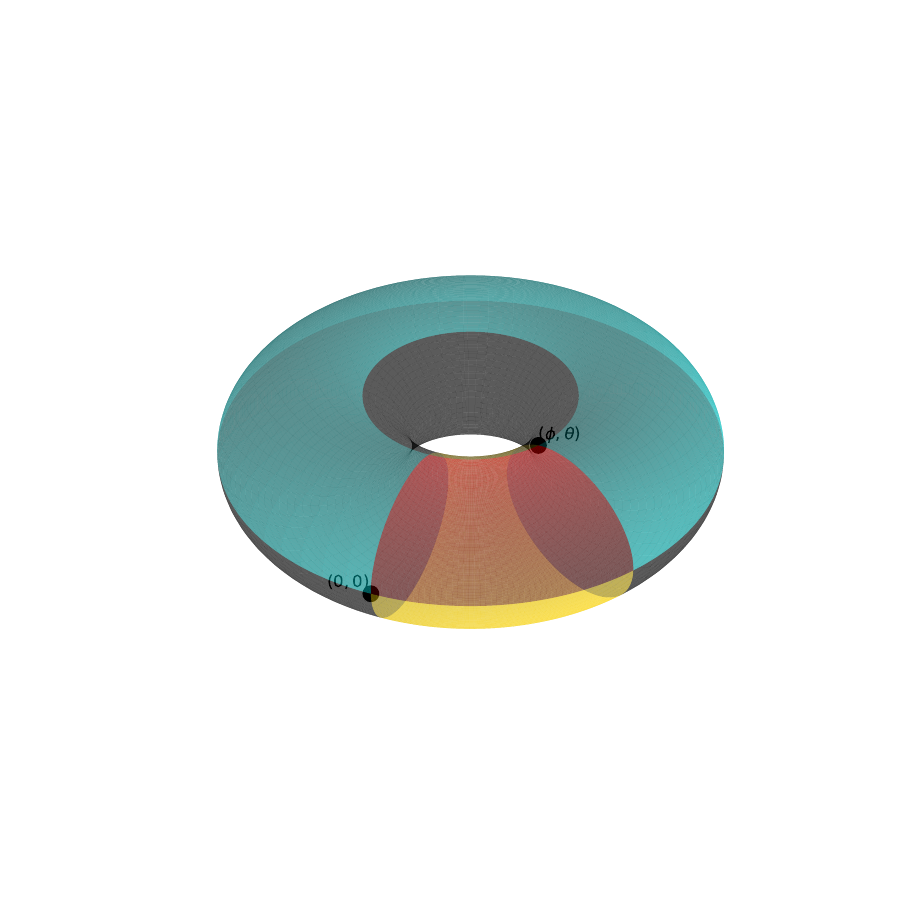}}}\hspace{5pt}
	
	\caption{Decomposition of the surface area between two points $(0, 0)$ and $(\phi, \theta)$ (a) on a flat torus and (b) on a curved torus.}
	\label{area_plot}
\end{figure}

The geometric and statistical motivation behind the ``square of an angle" is
to construct a test statistic for detecting a changepoint in the mean direction of angular data; it is useful to define a measure of separation that parallels the role of squared Euclidean distance in linear settings. While arc length is the most familiar notion of distance on the circle, it is linear in the angular difference and therefore does not provide the quadratic structure needed for variance-type statistics used in CUSUM. To address this, we introduce the concept of the \emph{square of an angle}, derived from the intrinsic geometry of the torus. This measure, independent of any underlying distribution, offers several advantages: it is geometrically consistent in the sense that it arises naturally from the manifold structure, just as squared Euclidean distance arises from the inner product in $\mathbb{R}^d$; it provides a quadratic form that enables the construction of variance-like CUSUM statistics, ensuring distribution-free validity; and it amplifies small shifts in mean direction, thereby improving the power of changepoint detection in the same way that squared deviations amplify signal in linear variance-based tests. In this sense, the square of an angle can be viewed as a principled, geometry-driven analogue of squared Euclidean distance that both respects the circular nature of angular data and enhances statistical inference, often yielding better power than traditional arc-length-based methods.

\textcolor{black}{
\begin{rmk}\label{rmk: info_preserv}
Let $g(\theta)
=
\operatorname{sgn}(\theta)A_C^{(0)}(\theta)$, where $\operatorname{sgn}(\theta)$ is defined in Equation-\ref{sign_function}. Now, the mapping \(g\) is strictly monotone on each of the intervals
\([0,\pi]\) and \((\pi,2\pi)\), and the corresponding image intervals are
disjoint. Hence, \(g\) is one-to-one almost everywhere on
\([0,2\pi)\). Since \(g\) is independent of the mean-direction parameter
\(\mu\), the complete transformed observation \(g(\Theta)\) preserves the
Fisher information contained in \(\Theta\). In particular, for a von
Mises distribution with known concentration parameter \(\kappa\),
\[
I_{g(\Theta)}(\mu)
=
I_\Theta(\mu)
=
\kappa\frac{I_1(\kappa)}{I_0(\kappa)}.
\]
Thus, the Fisher-information loss caused solely by the transformation is
zero. This statement should be distinguished from the efficiency of the SAMC
test statistic. SAMC uses only standardized partial sums of
\(g(\Theta_i)\), rather than the complete transformed likelihood.  The likelihood procedure exploits the assumed von Mises model and its
sufficient statistics, whereas SAMC uses a distribution-free approach based on the transformed real-valued sequence. Its practical
efficiency relative to the likelihood benchmark is examined in the revised
local-power simulation study. These results show the trade-off between
model-specific likelihood efficiency and distribution-free calibration,
without claiming uniform superiority of SAMC. For more detailed calculation see Appendix-\ref{info_preserv}.
\end{rmk}}

\section{Changepoint detection }
\label{changepoint detection}
\subsection{Changepoint detection in mean direction }
\label{section_cp_comcentration}

Let $\theta_1,\ldots, \theta_n \in [0,2\pi)$ be independent angular random variables.
We consider the following testing problem : 

\begin{eqnarray}
	H_{0} &:& \theta_{i}\overset{\mathrm{i.i.d.}}{\scalebox{1.5}{$\sim$}} F(\theta;\mu_1) 
	\mbox{~~for all~~} i =1,2,\ldots n, \nonumber\\
	H_{1}&:&   \begin{cases}
		\theta_{i} \overset{\mathrm{i.i.d.}}{\scalebox{1.5}{$\sim$}} F(\theta;\mu_1) & \text{, } 1 \leq i \leq k^*\\
		\theta_{i}\overset{\mathrm{i.i.d.}}{\scalebox{1.5}{$\sim$}} F(\theta;\mu_2)   & \text{, } (k^*+1) \leq i \leq n,
	\end{cases}
	\label{concentration test}
\end{eqnarray}
where $F$ is a circular distribution with a fixed concentration parameter $\kappa$, common for all $\theta_i$, and a mean direction $\mu_1 \neq \mu_2$ under $H_{1}$. Since the hypothesis focuses on the mean direction the concentration parameter is taken to be constant, which is one of the common practice in classical changepoint problems for linear uni-variate data \cite[see][]{gombay1990asymptotic,fotopoulos2010exact}, multivariate data \cite[see][]{Hovarth_1999}, functional data \cite[see][]{berkes2009detecting}, angular data \cite[see][]{ghosh1999change} as well.
Without loss of generality, the sign of the  circular random variable, $ \theta \in [0, 2\pi)$ can be defined as  
\begin{eqnarray}
	sgn(\theta)&=& \begin{cases}
			1& \text{~if~~} 0\leq \theta \leq \pi\\
			-1 & \text{~if~} \theta>\pi,
		\end{cases}
	\label{sign_function}
\end{eqnarray}
Using Definition-\ref{zero centered circle area of a segemnt} and the immediate above equation we can calculate
\begin{eqnarray}
   a_i=sgn(\theta)~A_C^{(0)}[\theta_i],\mbox{~~~~~for~} i=1,2,\ldots,n.
\label{square_area} 
\end{eqnarray}

Note that, $a_i$'s  are i.i.d. real-value random variables under the null hypothesis, $H_{0}$  with variance $$\sigma_{a}^2=Var(a_1)\widehat{=}\frac{1}{n-1} \sum_{i=1}^{n}\left(a_i-\bar{a}\right)^2=\widehat\sigma_a^2, ~~~ \mbox{where}~~~ n\bar{a}=\displaystyle \sum_{i=1}^{n}a_i.$$   

The CUSUM technique is well established in literature \cite[see,][]{fearnhead2019changepoint, cho2015multiple} to construct the test statistic identifying the changepoint. Using the concept of ``square of an angle'', we now propose a test statistic based on CUSUM technique to test the existence of a changepoint in the \textbf{mean direction}, $\left[\tan^{-1*} \left(\frac{E(\sin \Theta)}{E(\cos \Theta)}\right)\right],$ where $\tan^{-1*}$ is the quadrant-specific inverse  is defined as  $$\tan^{-1*}(y, x) =
	\begin{cases}
		\arctan\left(\frac y x\right) &\text{if } x > 0, y \geq 0 \\
		\arctan\left(\frac y x\right) + \pi &\text{if } x < 0 , \\
		\arctan\left(\frac y x\right) + 2\pi &\text{if } x \geq 0 \text{ and } y < 0, \\
		+\frac{\pi}{2} &\text{if } x = 0 \text{ and } y > 0, \\
		\text{undefined} &\text{if } x = 0 \text{ and } y = 0.
	\end{cases}$$
For more details, see p-13, Equation-1.3.5 \cite{jammalamadaka2001topics}. Here of $E(\sin \Theta)=\int_{0}^{2\pi} \sin\theta f(\theta) d\theta$ for a  suitable probability density function, $f(\theta)$ of a circular random variable $\Theta.$ Similarly, $E(\cos \Theta)$ can be defined.

 Now, we define a CUSUM process as
\begin{equation}
    T(k)=\frac{1}{\sqrt{n}~\widehat\sigma_a }\left[ \sum_{i=1}^{k} a_i-k\bar{a}   \right], \mbox{~~~~~for all~~} k=1,\ldots,n 
	\label{cusum process}
\end{equation}
to construct the test statistic
\begin{equation}
	\mathbb{M}_n=\displaystyle \max_{1 \leq k < n}  |T(k)|.
	\label{concentration_test_statistic}
\end{equation}

Hence, we reject the null hypothesis, $H_{0}$, if $\mathbb{M}_n>k_{\alpha}$, where, $k_{\alpha}$ is the upper $\alpha$ point of the exact (or asymptotic) distribution of $\mathbb{M}_n$ under the null hypothesis. The closed-form distribution of $\mathbb{M}_n$ is not available; hence, we need to take recourse to simulation to obtain the cut-off value $k_{\alpha}$. When $n$ is large, the limiting distribution of $\mathbb{M}_n$ can be derived as follows:
Let us consider $u \in (0,1)$, and denote $k= \lfloor{nu}\rfloor$. Hence, from Equation-\ref{cusum process} we can write
\begin{equation}
	T_{n}(u)=T(\lfloor{nu}\rfloor)=\frac{1}{\sqrt{n~} \hat\sigma_{a}} \left[ \sum_{i=1}^{\lfloor{nu}\rfloor} a_i-u \sum_{i=1}^n a_i   \right].
	\label{concentration_asym}
\end{equation}

\textcolor{black}{It is crucial to address the methodological transition from the geometric construction to the asymptotic framework. 
While the derivation of the asymptotic distribution of $\mathbb{M}_{n}$ ultimately relies on the standard functional central limit theorem (FCLT) under mild conditions, the fundamental statistical innovation lies in the geometric transformation that enables this application. 
By projecting the angular data onto the intrinsic geometry of a curved torus, the proposed metric $A_{C}^{(0)}(\theta)$ (and consequently the transformed sequence $a_{i}$) translates the topological structure of the manifold into a real-valued quadratic form. 
This non-trivial geometric mapping resolves the cyclical nature of the data, guarantees finite second moments ($E[a_{i}^{2}] < \infty$), and provides a stationary variance structure under $H_{0}$. 
Consequently, this transformation rigorously satisfies the necessary regularity conditions for the FCLT, enabling the empirical process to pivot around the Kolmogorov distribution irrespective of the underlying circular distribution.} 
Then, with the proper embedding of Skorohod topology in $D[0,1]$ \cite[see][Ch. 3]{billingsley2013convergence}, under the null hypothesis, $H_{0}$, and as $n \rightarrow \infty$, the process $T_n(u)$  converges weakly to $B_{0}(u),$ where $B_{0}(u)$ is the standard Brownian bridge on $[0,1].$ Hence, we get the following lemma.

\begin{lemma}
	Under the null hypothesis, $H_0$,
	\begin{equation}
		\mathbb{M}_n  \overset{d}{\to} \displaystyle \sup_{0<u<1} |B_{0}(u)|=K_\infty,
		\label{bbridge}
	\end{equation}
	where $K_\infty$ follows the \textit{Kolmogorov distribution}.
	\label{bbridge_lemma}
\end{lemma}
The outline of the proof is provided in Appendix-\ref{null_proof}.


\begin{rmk}
\textcolor{black}{Throughout, we assume that the circular observations
\(\theta_1,\ldots,\theta_n\) are independent and identically distributed
and that the transformed variables
$a_i=\operatorname{sgn}(\theta_i)A_C^{(0)}(\theta_i)$
have finite second moments, i.e.,
\(E(a_i^2)<\infty\). Under these conditions, the standardized partial-sum
process of the \(a_i\)'s converges weakly to a Wiener process, and the
corresponding CUSUM process converges to a standard Brownian bridge.
Consequently,
$\mathbb{M}_n
\xrightarrow{d}
\sup_{0<u<1}|B_0(u)|,$
which follows the Kolmogorov distribution.}

\textcolor{black}{For finite samples, the critical value is calibrated according to the
length of the particular sequence being tested. Specifically, for a
segment containing \(m\) observations, we use the finite-grid Kolmogorov
critical value
$K_{\infty}^{(m)}$
obtained from Algorithm-\ref{alg:algo_cutoff}. Thus, the complete sequence
uses \(K_{\infty}\)=\(K_{\infty}^{(n)}\), whereas a subsegment of length \(m\)
uses \(K_{\infty}^{(m)}\).
In the empirical post-detection analysis, a resulting subsegment is
examined further only when its length exceeds the prescribed minimum
segment length \(n/10\), where \(n\) denotes the length of the original
sequence. If the segment is shorter than or equal to this threshold, no
further test is performed on that segment. This minimum-length rule is used
only for the empirical post-detection stability checks.}
\end{rmk}

Now, we can compute the upper-$\alpha$ value, $k_{\alpha},$ from the above limiting random variable $K_\infty$ of the test statistic, $\mathbb{M}_n$. The corresponding large sample approximation is discussed in Section-\ref{simulation}. Under the alternative hypothesis, $H_1$, the proposed test is consistent, which is indicated by the following corollary.


\begin{theorem}
\textcolor{black}{Let $k^* = \lfloor n u^* \rfloor$ with $u^* \in (0,1)$ be the location of the true single changepoint and $\hat{k}^*$ be the estimated location obtained by 
\[
\hat{k}^* = \arg\max_{1 \leq k < n} \left| T(k) \right|.
\]
Then, $\frac{\hat{k}^*}{n} \xrightarrow{p} u^*$ under $H_1$.}
\label{loc_consistancy}
\end{theorem}
\begin{proof}
The proof is provided in  Appendix-\ref{loc_consistancy_proof}
\end{proof}

\begin{corollary}
     Under $H_0$, the test statistic, $ \mathbb{M}_n$ we have the following:
$$\lim_{n \rightarrow\infty} \mathbb{P}_{H_0}\left(\mathbb{M}_n\geq k_\alpha  \right)\rightarrow \alpha.$$
\label{level_corr_cir}
\end{corollary}
\begin{proof}
    As we obtained in Lemma-\ref{bbridge_lemma}, the asymptotic distribution of the
test statistic is Kolmogorov distribution in general under the null hypothesis. So, the
proof is straightforward.
\end{proof}

\begin{corollary}
	For a fixed value of $\alpha \in (0,1)$, the probability of Type-II error   goes to zero exponentially, as the sample size increases, i.e.
	
	$$\lim_{n \rightarrow\infty} \mathbb{P}_{H_1}\left(\mathbb{M}_n< k_\alpha  \right) \rightarrow 0.$$
	\label{consistency_corr}
\end{corollary}
\begin{proof}
The proof is provided in  Appendix-\ref{alt_consistency_proof}
\end{proof}

\section{Simulation Studies}
\label{simulation}

A comprehensive simulation study was conducted  to detect the changepoint in the mean direction, we have considered the von Mises distribution with the probability distribution function 
\begin{equation}
	f_{\text{vm}}(\theta;\mu,\kappa)=\frac{e^{\kappa\cos(\theta-\mu)}}{2\pi I_{0}(\kappa)},
	\label{von mises}
\end{equation}
where $0\leq \theta<2\pi$, $0\leq \mu<2\pi$, $\kappa>0$, and $ I_{0}(\kappa)$ is the modified Bessel function with order zero evaluated at $\kappa.$ \

The Figure-\ref{fig: null_density_plot} displays a density plot of the  test statistic, $\mathbb{M}_n$ under $H_{0}$
with the sample size of $n=1000$ from the von Mises distribution with the different mean directions $\mu=\frac{\pi}{4},\frac{2\pi}{3},\frac{4\pi}{3},\frac{8\pi}{5}$, and the fixed concentration parameter, $\kappa=3$. It is evident from Figure-\ref{fig: null_density_plot} that the densities of the test statistic for different mean directions are nearly identical, and these are close to the density of the limiting distribution of the random variable $K_{\infty}$ (Equation-\ref{bbridge}).  In Figure-\ref{sb_null_kappa_plot}(a)  $\&$ (b), the distribution of the estimated location of the changepoint and the histogram of the SAMC test statistic, $\mathbb{M}_n$ (Equation-\ref{concentration_test_statistic}) are presented under the null hypothesis, $H_{0}$, respectively. The number of iterations for each specification is conducted $5 \times 10^3$ number of times.

In Table-\ref{table: null cut-off table concentration change}, we present the cut-off values of the test statistic $\mathbb{M}_n$ under the null hypothesis, $H_{0}$. These values are provided for different sample sizes, namely 50, 100, 200, 500, and 1000, drawn from a von Mises distribution with a mean direction of $\mu=0$ and various concentration parameters, specifically $\kappa=0.5, 1, 1.5, 2, 4, $ and $10$.  This table shows that the cut-off is pivotal with respect to the concentration parameter, and hence, it is a distribution-free test.  The Table-\ref{table: null cut-off table mean change}  displays the cut-off values of the test statistic $\mathbb{M}_n$ under the null hypothesis, $H_{0}$. The samples used for this table have sizes of $50, 100, 200, 500,$ and $1000$, and are drawn from a von Mises distribution with a fixed concentration parameter of $\kappa=3,$ and different mean direction parameters $\mu=\frac{\pi}{4},\frac{2\pi}{3},\frac{4\pi}{3}$ and $\frac{8\pi}{5}.$ The results from this table are consistent with the Figure-\ref{fig: null_density_plot}. In both tables, we present the cut-off values obtained from the limiting distribution of $K_\infty$ (Equation-\ref{bbridge}) using grid sizes of $n=50, 100, 200, 500, $ or $ 1000$ accordingly. These values are generally referred to as $K_\infty^{(n)}$. 

\begin{figure*}[h!]
	\centering
	\includegraphics[width=0.6\textwidth,height=0.3\textwidth]{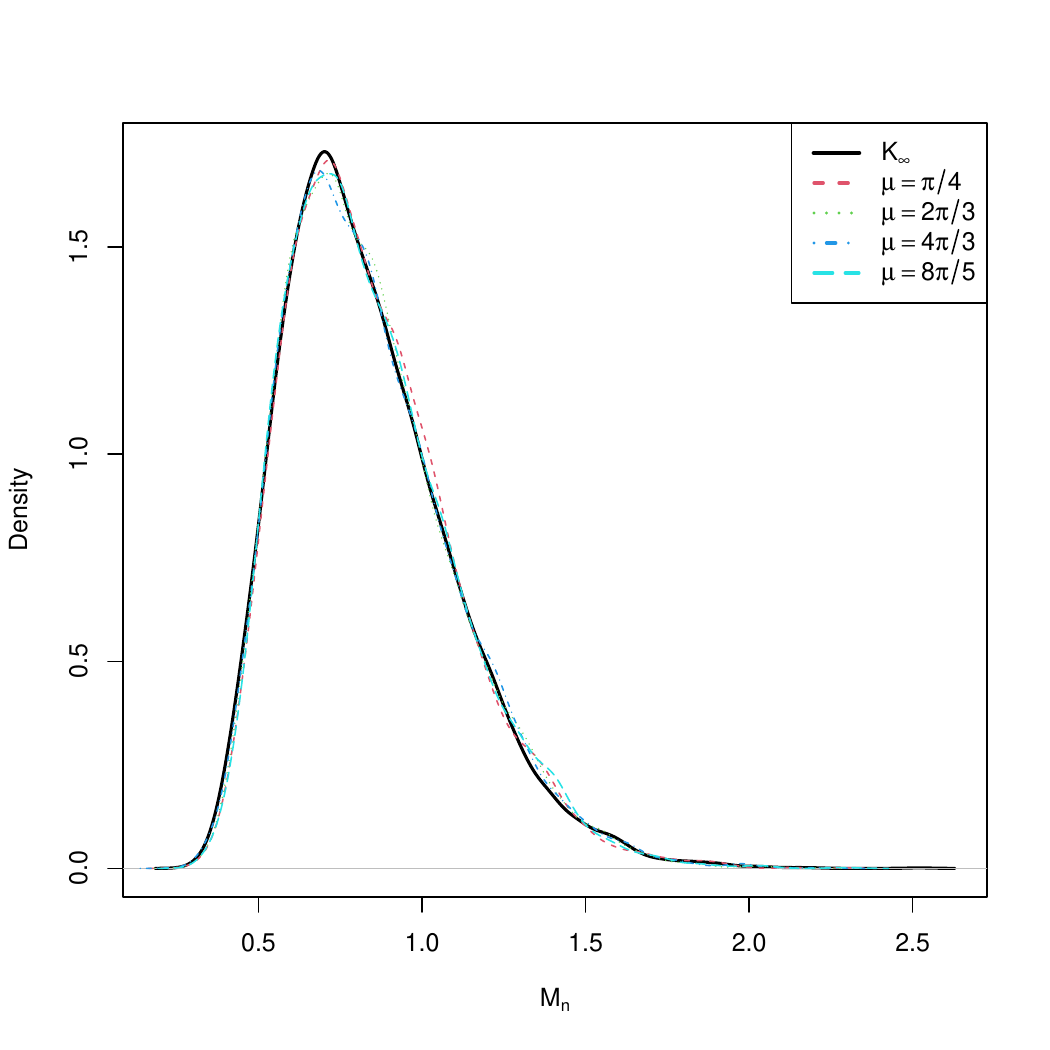}
	\caption{Density plot of the SAMC  test statistic, $\mathbb{M}_n$ under $H_{0}$ with sample of size $n=1000$ from von Mises distribution with the different mean direction $\mu=\frac{\pi}{4},\frac{2\pi}{3},\frac{4\pi}{3},\frac{8\pi}{5}$, and the fixed concentration parameter, $\kappa=3$ along with the density plot of the limiting random variable $K_{\infty}$.}
	\label{fig: null_density_plot}
\end{figure*}

\begin{figure}[h!]
	\centering
	\subfloat[]{%
		{\includegraphics[width=0.32\textwidth, height=0.35\textwidth]{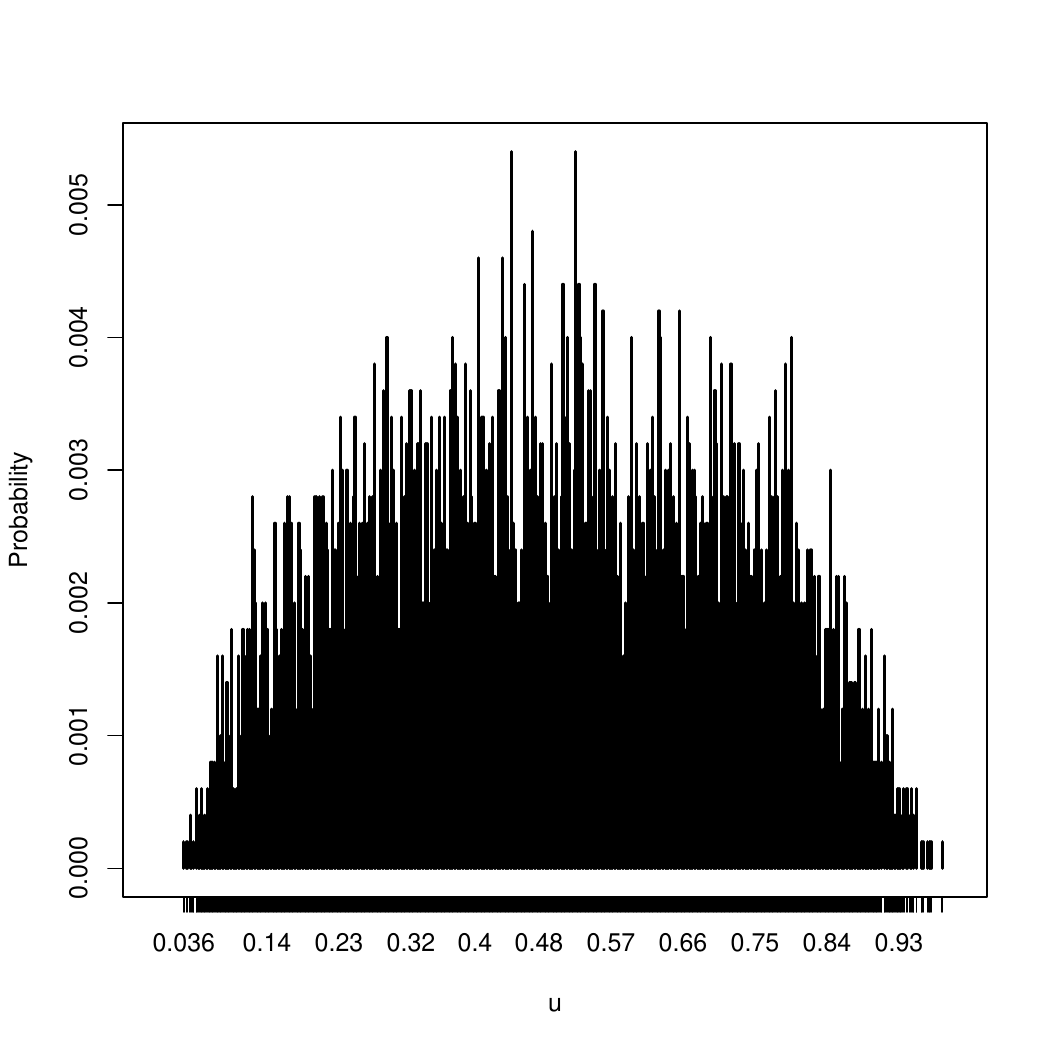}}}
	\subfloat[]{%
		{\includegraphics[width=0.32\textwidth, height=0.35\textwidth]{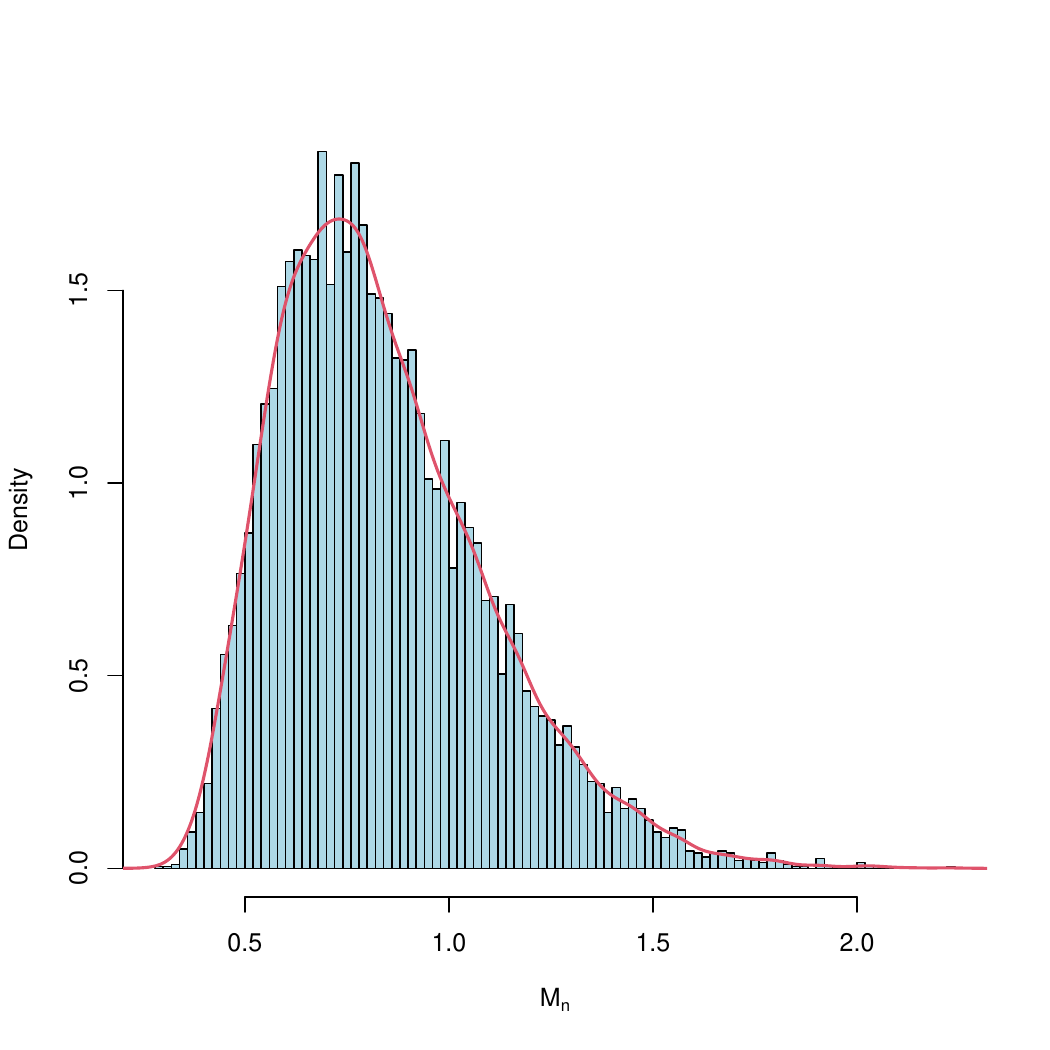}}}\hspace{5pt}
	
	\caption{ (a) The plot of the distribution of the estimated location of changepoint and (b) the histogram of the SAMC test statistic, $\mathbb{M}_n$ under the null hypothesis, $H_{0}$ when the data are from vthe vonMises distribution (sample size $n=500$) with  concentration parameter, $\kappa=2$ and the mean direction being $\mu=\frac{\pi}{6}$.}
	\label{sb_null_kappa_plot}
\end{figure}

The power computation of the SAMC test has been summarised in Algorithm-\ref{alg:algo_power_kappa}. Under the alternative hypothesis, $H_{1}$, the Figure-\ref{sb_alt_kappa_plot}(a) $\&$ (b) show the distribution of the estimated location of the changepoint and the histogram of the SAMC test statistic, $\mathbb{M}_n$ (Equation-\ref{concentration_test_statistic}), respectively. 
Here also, the random samples are drawn from the von Mises distribution with a sample size of $n=500$ and fixed concentration parameter $\kappa=2$. The mean direction is $\mu_1=\frac{\pi}{6}$ before the true changepoint  $k^{*}=\frac{n}{2}$, and $\mu_1=\frac{2\pi}{6}$ after the changepoint.

The Figure-\ref{power_100ss_vm}(a) \& (b) depict the power curves of the SAMC test statistic, $\mathbb{M}_n$, at the level of $1\%$ and $5\%$, respectively, for the sample sizes of 100, when the data are drawn from the von Mises distribution with the fixed concentration parameter $\kappa=3$. Under the alternative hypothesis, before the changepoint,  the mean direction  $\mu=0$  whereas after the true changepoint at $k^*=\dfrac{n}{2}$, the mean direction is $\mu+\delta_{\mu}$, where $\delta_{\mu} \in [-\frac{\pi}{2}, \frac{\pi}{2}].$
For each of the specifications, $5 \times 10^3$ iterations have been performed. The power curves in Figure-\ref{power_500ss_vm}(a) \& (b) are also obtained for a similar specification with the sample size of $500.$

\begin{figure}[h!]
	\centering
	\subfloat[]{%
		{\includegraphics[width=0.32\textwidth, height=0.35\textwidth]{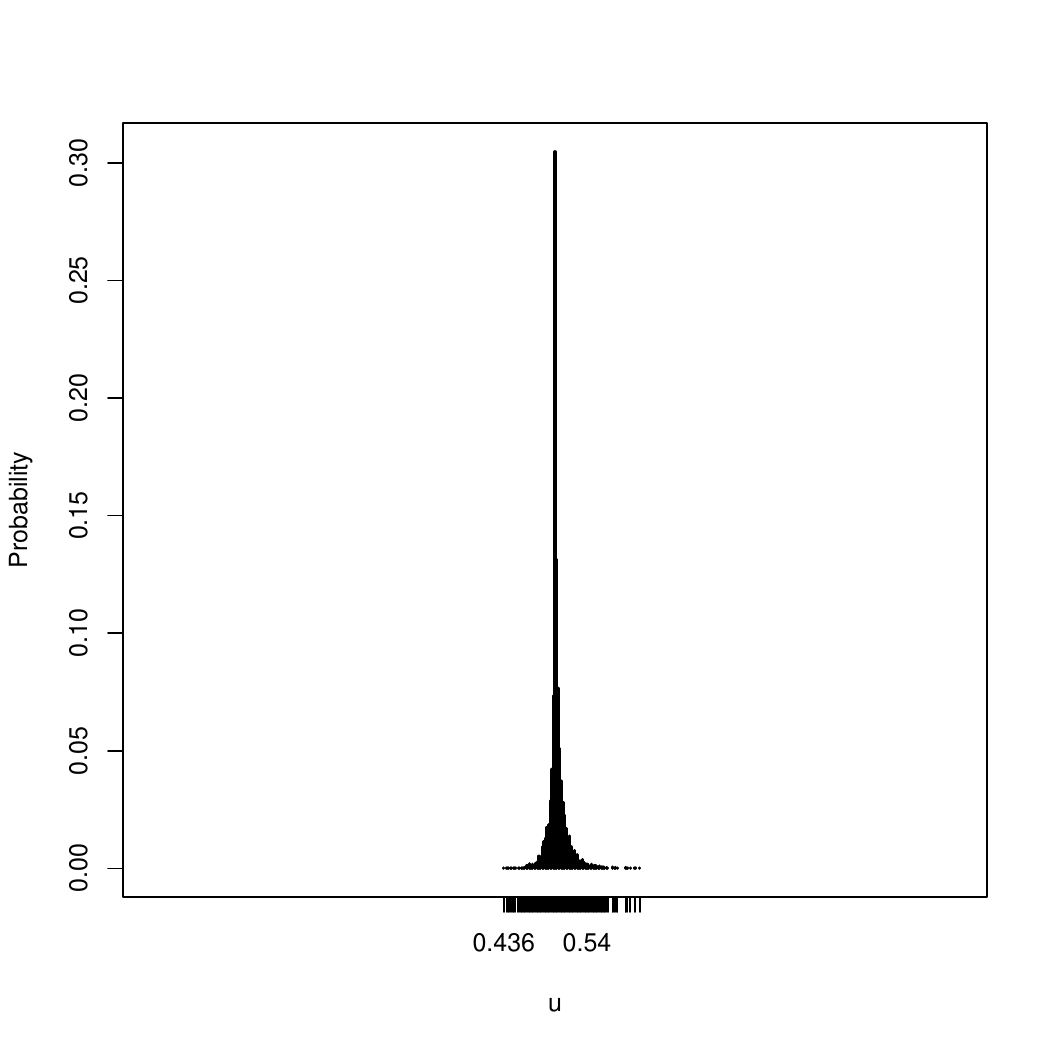}}}
	\subfloat[ ]{%
		{\includegraphics[width=0.32\textwidth, height=0.35\textwidth]{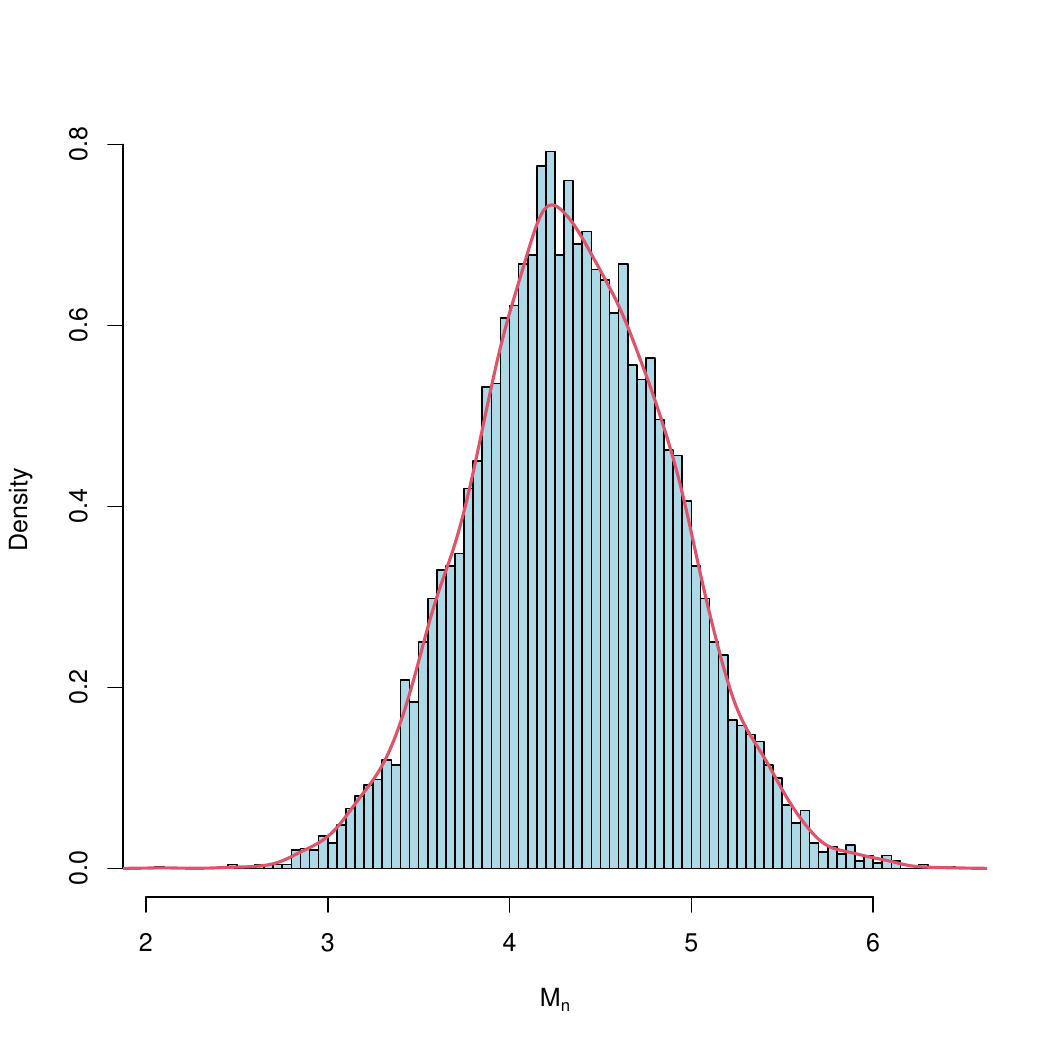}}}\hspace{5pt}
	
	\caption{ (a) The plot of the distribution of the estimated location of changepoint and (b) the histogram of the SAMC test statistic, $\mathbb{M}_n$ when the true changepoint is at $k^*=\dfrac{n}{2}$,  $\mu_1=\frac{\pi}{6}$ (mean direction before changepoint), $ \mu_2=\frac{2\pi}{6} $ (mean direction after changepoint) under the alternative hypothesis, $H_{1}$ when the data are from von Mises distribution (sample size $n=500$,  and fixed concentration parameter $\kappa=2$).}
	\label{sb_alt_kappa_plot}
\end{figure}

\begin{figure}[h!]
	\centering
	\subfloat[ level $1\%$  .]{%
		{\includegraphics[ trim= 0 0 0 40, clip, width=0.32\textwidth, height=0.32\textwidth]{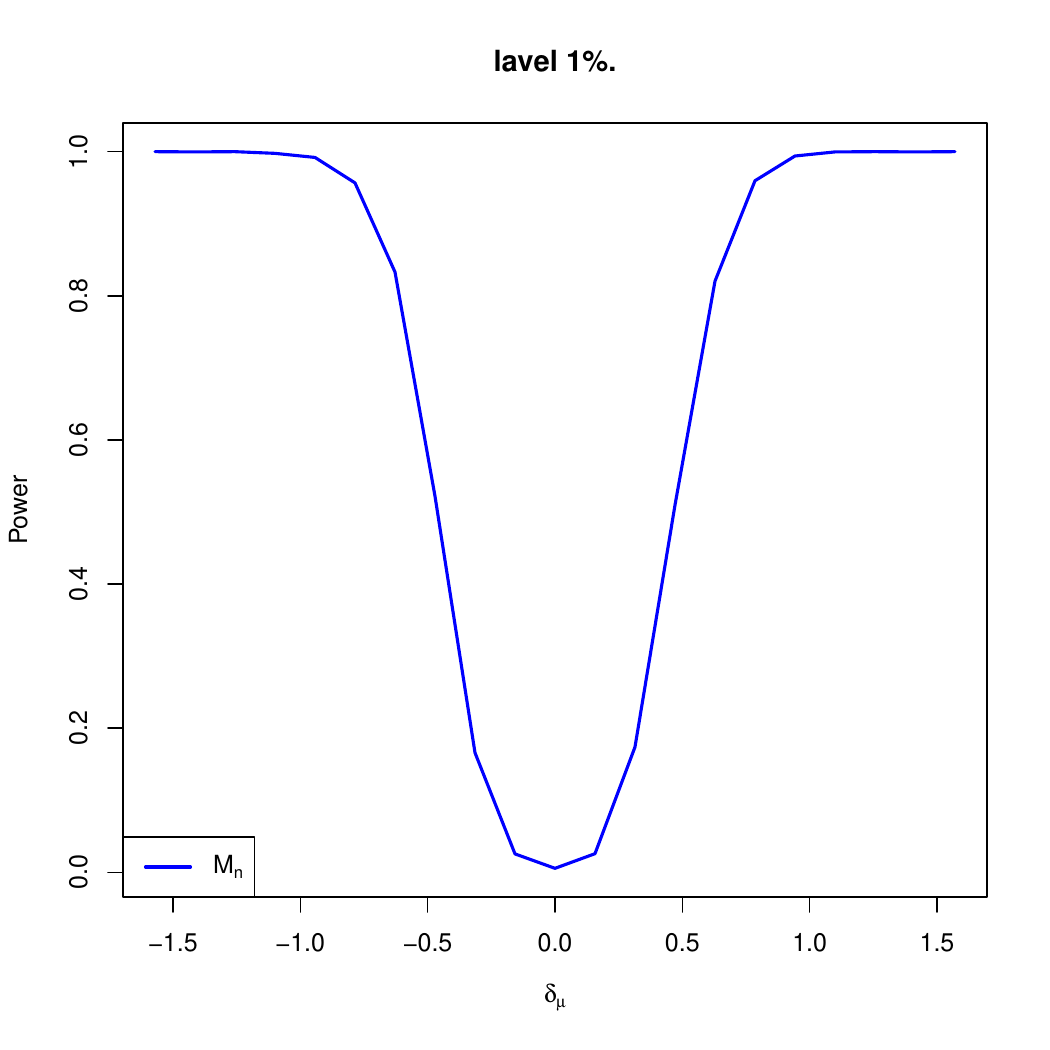}}}\hspace{5pt}
	\subfloat[ level $5\%$.]{%
		{\includegraphics[trim= 0 0 0 40, clip, width=0.32\textwidth, height=0.32\textwidth]{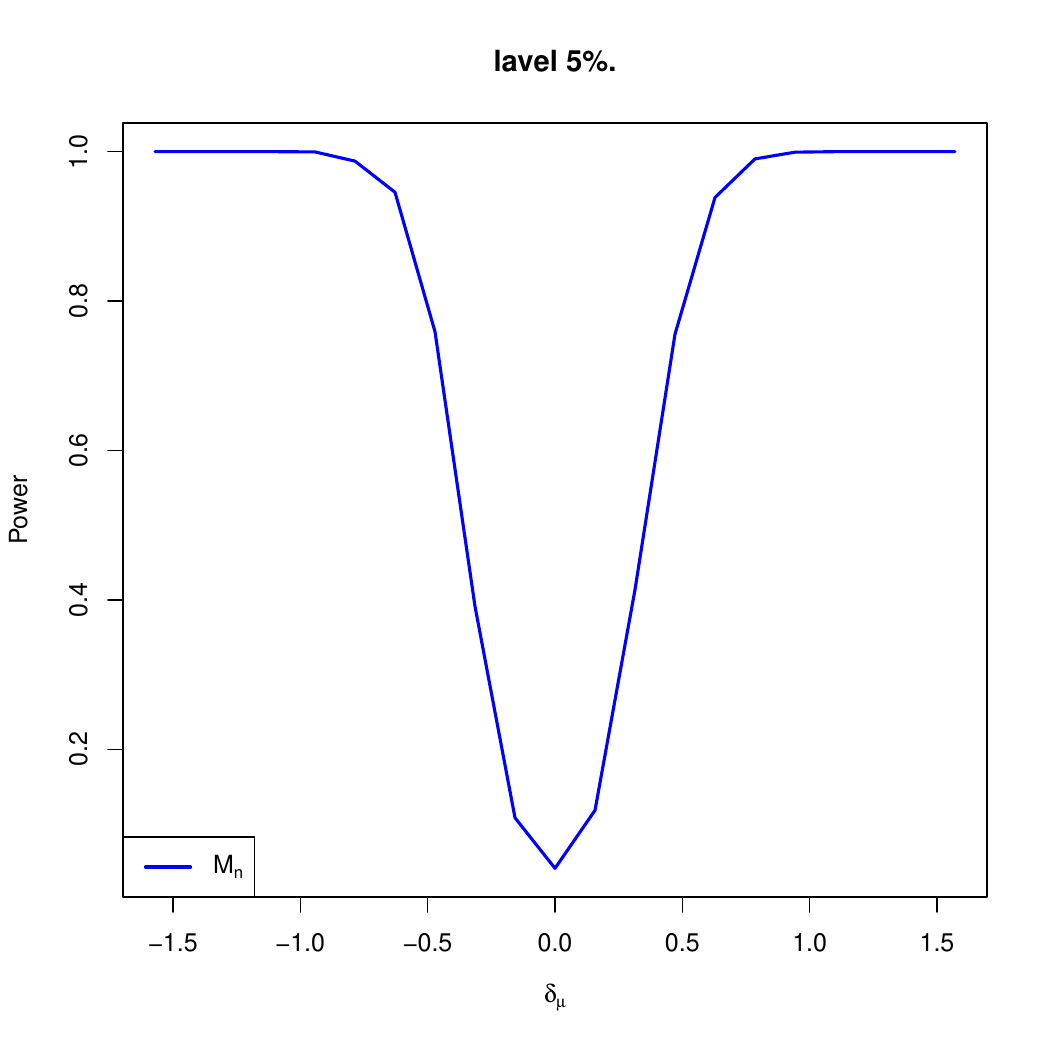}}}
	\vspace{5pt} 
	\caption{Power curves of the SAMC test statistic, $\mathbb{M}_n$  with respect to $\delta_\mu$ for the sample size of $n=100,$ drawn from von Mises distribution  with fixed concentration parameter, $\kappa=3$. The true changepoint at $k^*=\dfrac{n}{2}$ under $H_{1},$ and the pre-changed mean direction is $\mu=0$.}
	\label{power_100ss_vm}
\end{figure}

\begin{figure}[h!]
	\centering
	\subfloat[ level $1\%$.]{%
		{\includegraphics[trim= 0 0 0 40, clip, width=0.32\textwidth, height=0.32\textwidth]{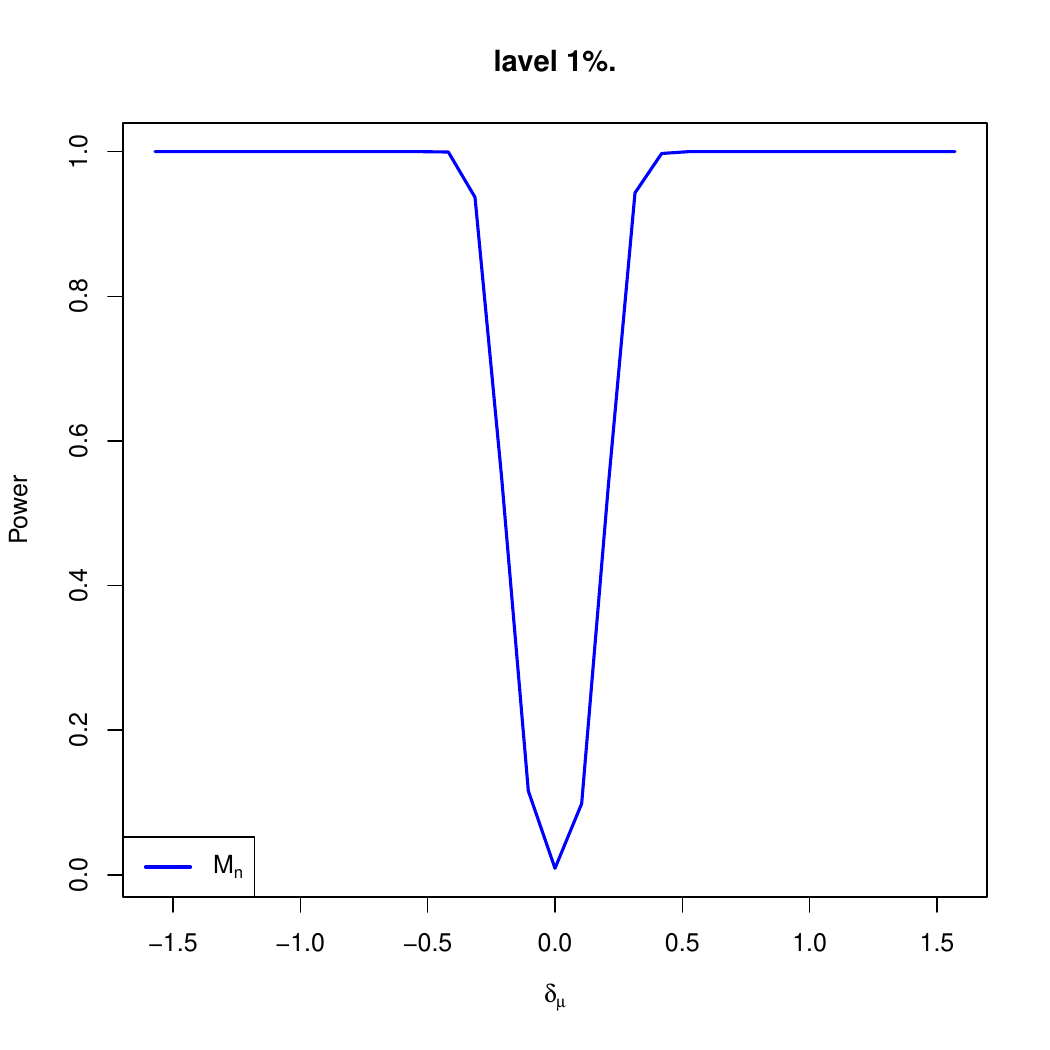}}}\hspace{5pt}
	\subfloat[ level $5\%$.]{%
		{\includegraphics[trim= 0 0 0 40, clip, width=0.32\textwidth, height=0.32\textwidth]{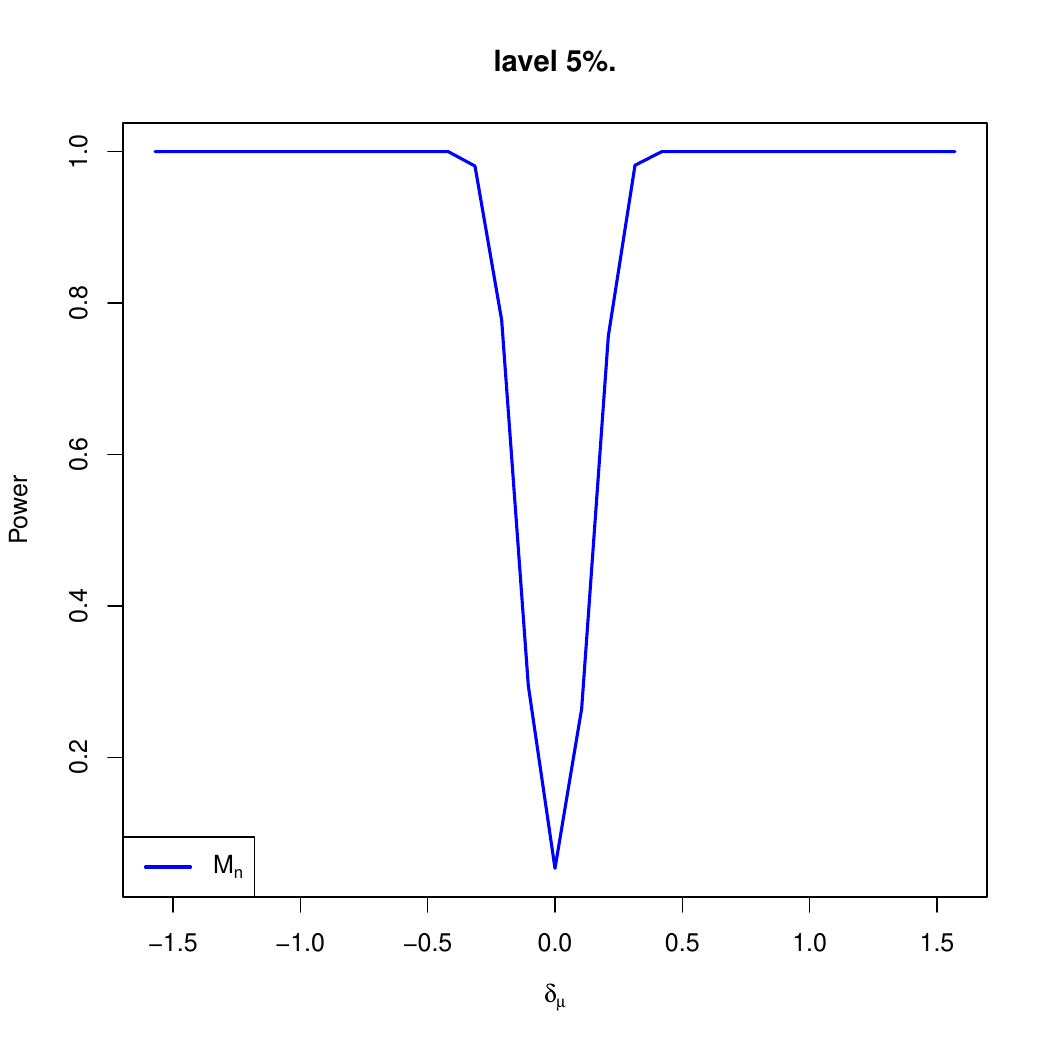}}}
	\vspace{5pt} 
	\caption{Power curves of the SAMC test statistic, $\mathbb{M}_n$  with respect to $\delta_\mu$ for the sample size of $n=500,$ drawn from von Mises distribution  with fixed concentration parameter, $\kappa=3$. The true changepoint at $k^*=\dfrac{n}{2}$ under $H_{1},$ and the pre-changed mean direction is $\mu=0$.}
	\label{power_500ss_vm}
\end{figure}

To show that our approach does not depend on the distribution, we have considered another popular circular distribution, the wrapped-Cauchy distribution, which has the following probability density function: 

\begin{equation}
	f_{\text{wc}}(\theta;\mu,\rho)=\frac{1}{2\pi}\dfrac{1-\rho^2}{1+\rho^2-2\rho \cos({\theta-\mu})},
	\label{wrap cauchy}
\end{equation}

where $0\leq \theta<2\pi$, $0\leq \mu<2\pi$, and $0\leq \rho <1$.
Figure-\ref{power_100ss_wc}(a) \& (b) depict the power curves of the SAMC test statistic, $\mathbb{M}_n$, at the level of $1\%$ and $5\%$, respectively, for the sample sizes of $100$,  when the data are drawn from the wrapped-Cauchy distribution with the fixed concentration parameter $\rho=0.7$  with the mean direction  $\mu=0$. Under the alternative hypothesis, before the changepoint,  the mean direction  $\mu=0$  whereas after the true changepoint at $k^*=\dfrac{n}{2}$, the mean direction is $\mu+\delta_{\mu}$, where $\delta_{\mu} \in [-\frac{\pi}{2}, \frac{\pi}{2}].$
For each of the specifications, $5 \times 10^3$ iterations have been performed. The power curves in Figure-\ref{power_500ss_wc}(a) \& (b) are also obtained for a similar specification with the sample size of $500.$



\begin{figure}[h!]
	\centering
	\subfloat[ level $1\%$  .]{%
		{\includegraphics[trim= 0 0 0 40, clip, width=0.32\textwidth, height=0.32\textwidth]{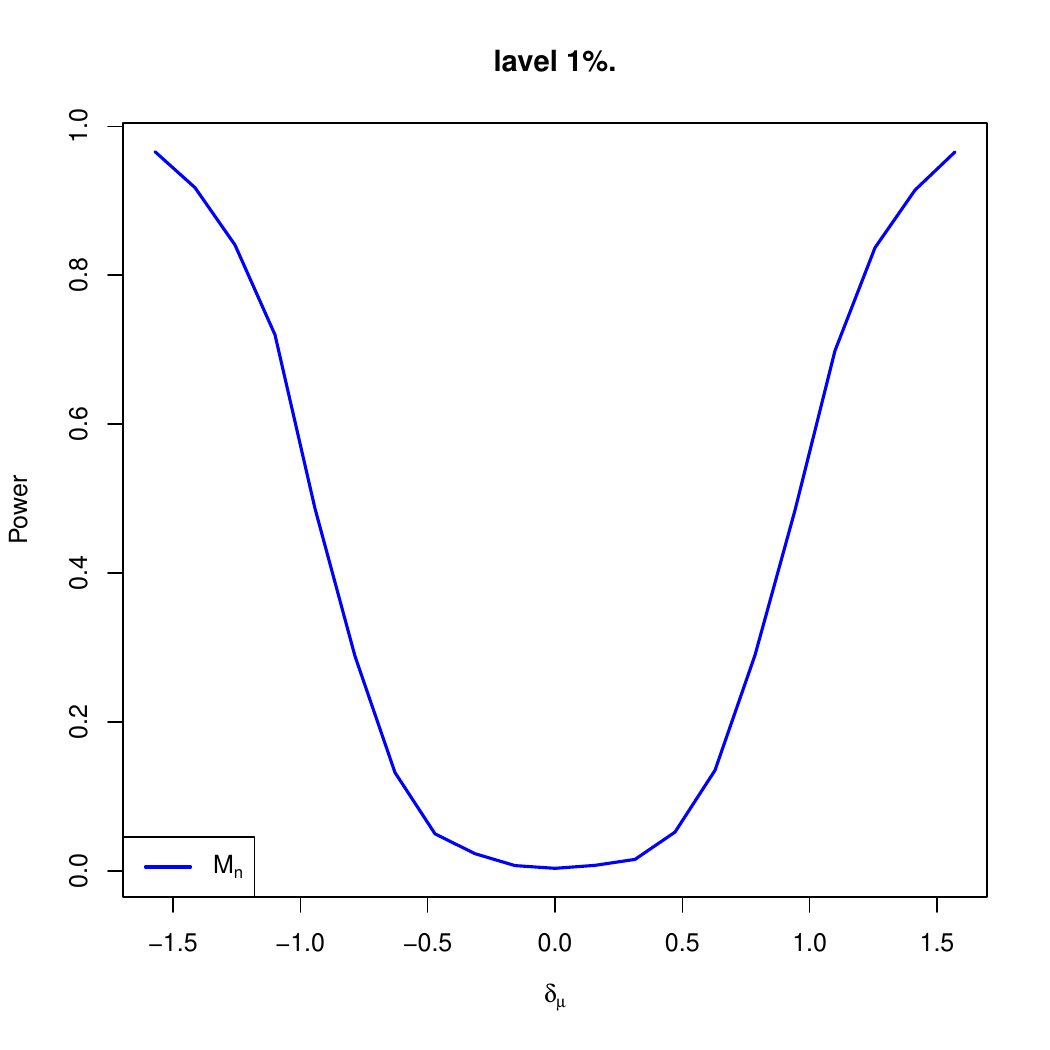}}}\hspace{5pt}
	\subfloat[ level $5\%$.]{%
		{\includegraphics[trim= 0 0 0 40, clip, width=0.32\textwidth, height=0.32\textwidth]{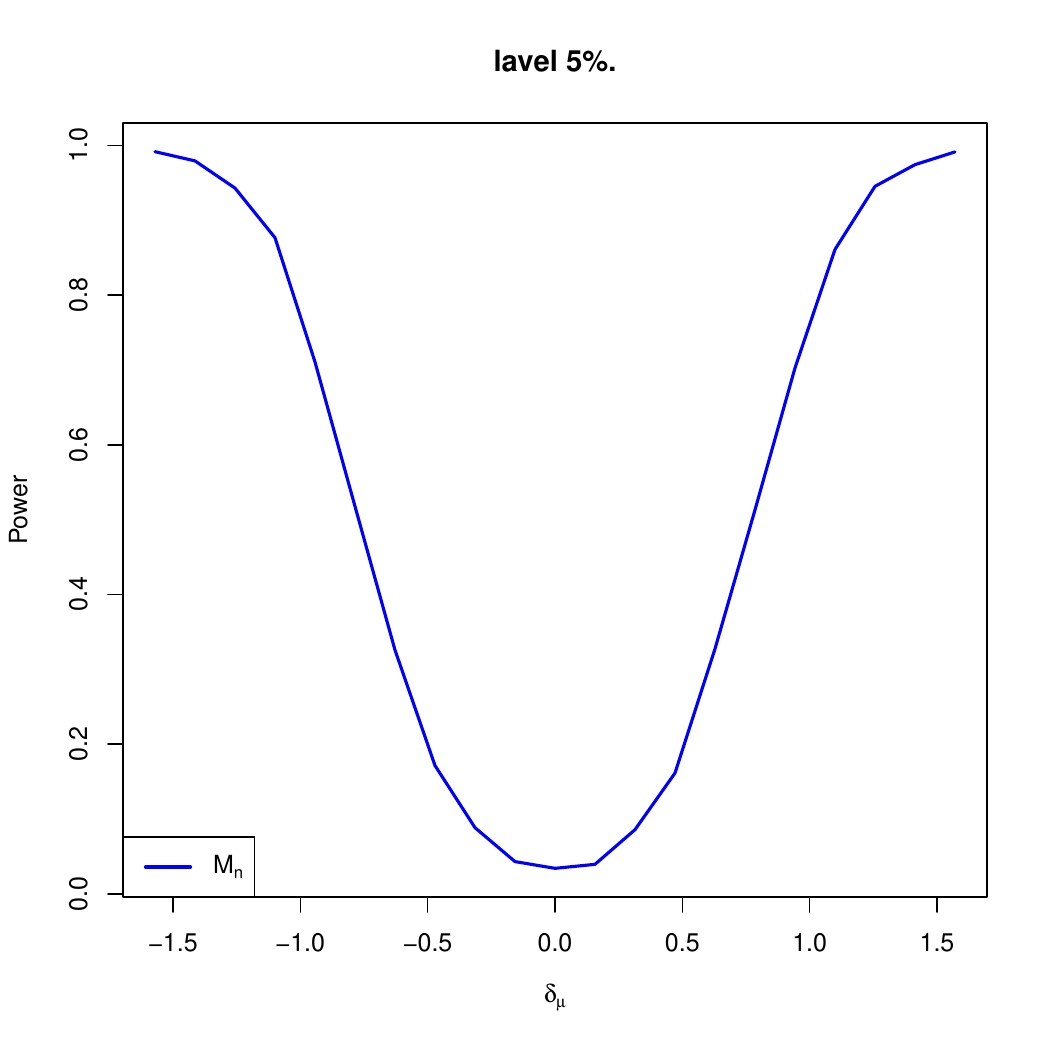}}}
	\vspace{5pt} 
	\caption{Power curves of the SAMC test statistic, $\mathbb{M}_n$  with respect to $\delta_\mu$ for the sample size of $n=100,$ drawn from wrapped Cauchy distribution  with fixed concentration parameter, $\rho=0.7$. The true changepoint at $k^*=\dfrac{n}{2}$ under $H_{1},$ and the pre-changed mean direction is $\mu=0$.}
	\label{power_100ss_wc}
\end{figure}

\begin{figure}[h!]
	\centering
	\subfloat[ level $1\%$  .]{%
		{\includegraphics[trim= 0 0 0 40, clip, width=0.32\textwidth, height=0.32\textwidth]{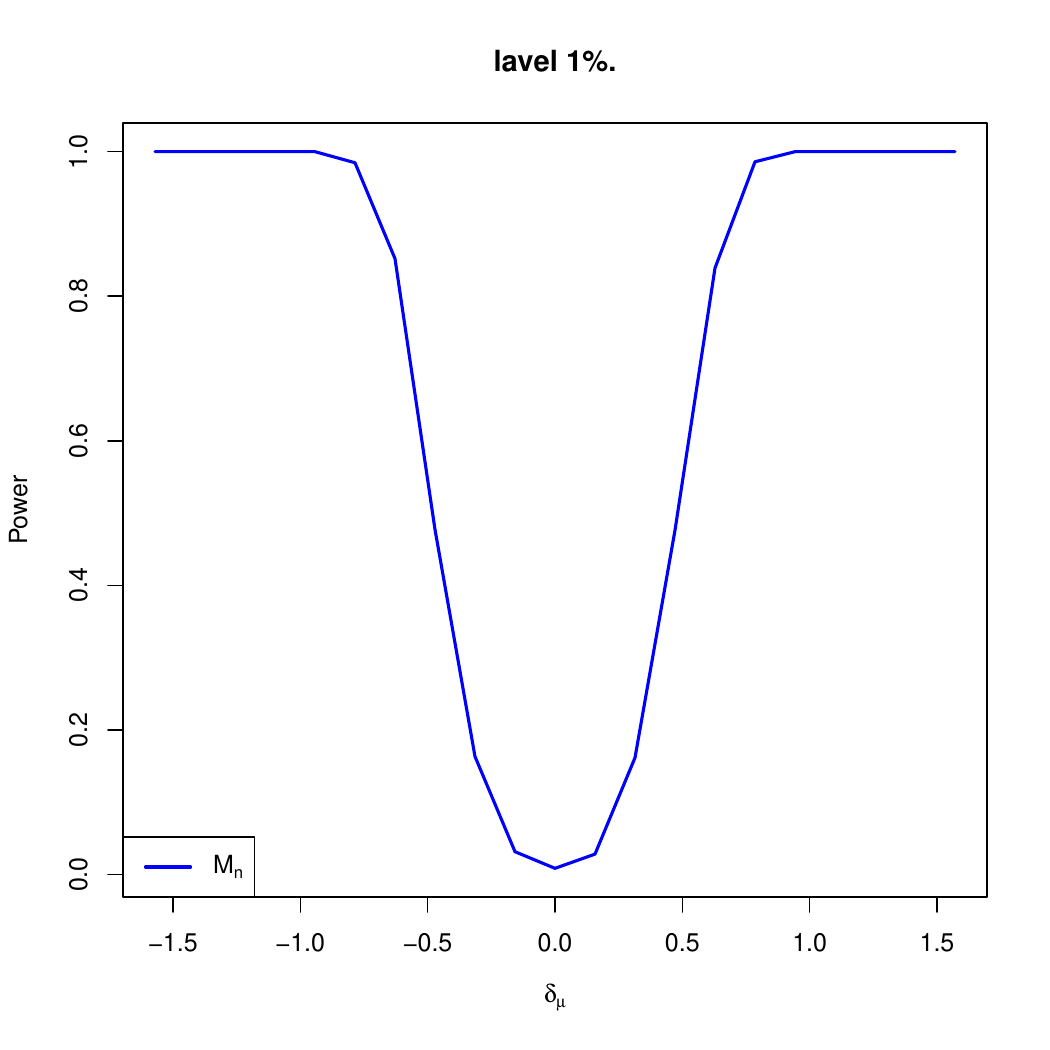}}}\hspace{5pt}
	\subfloat[ level $5\%$.]{%
		{\includegraphics[trim= 0 0 0 40, clip, width=0.32\textwidth, height=0.32\textwidth]{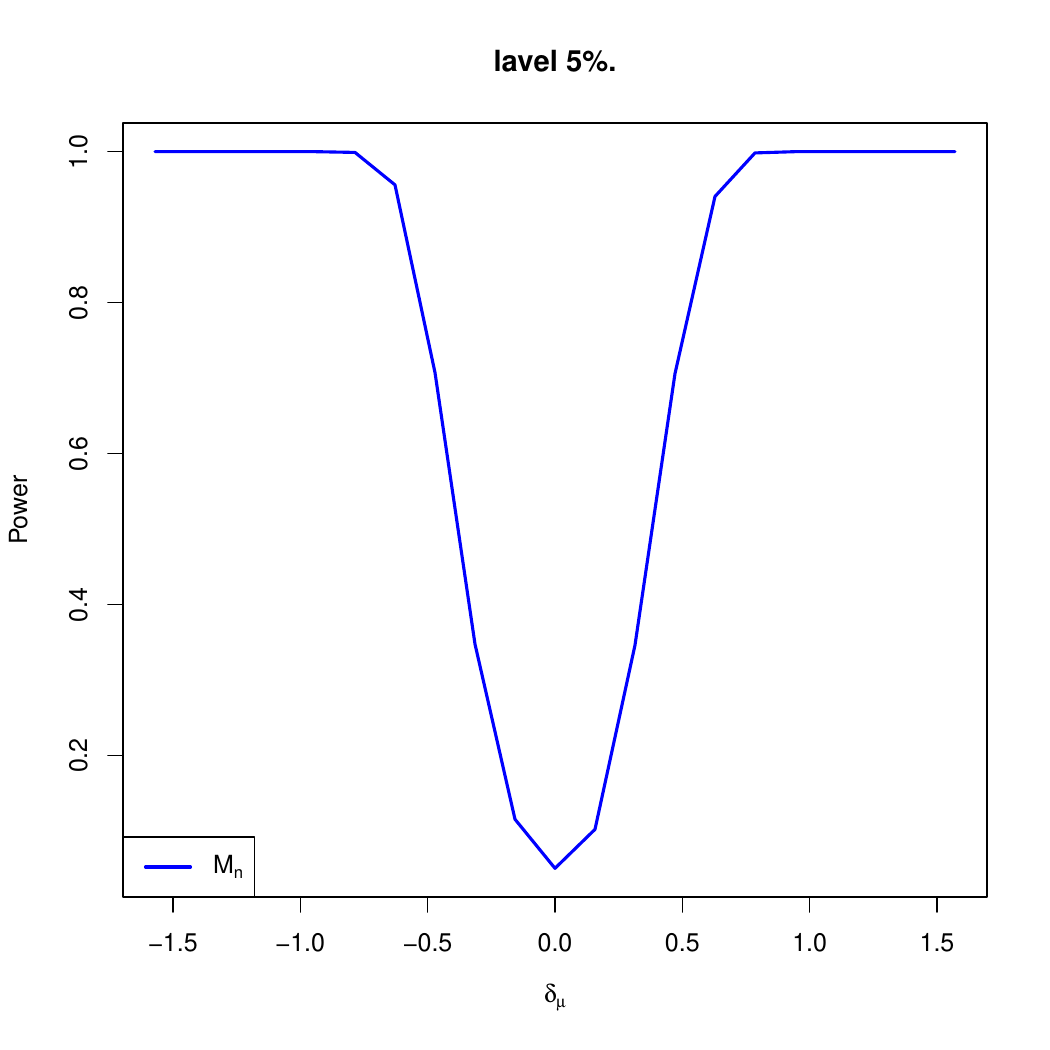}}}
	\vspace{5pt} 
	\caption{Power curves of the SAMC test statistic, $\mathbb{M}_n$  with respect to $\delta_\mu$ for the sample size of $n=500$ drawn from wrapped Cauchy distribution  with fixed concentration parameter, $\rho=0.7$. The true changepoint at $k^*=\dfrac{n}{2}$ under $H_{1},$ and the pre-changed mean direction is $\mu=0$.}
	\label{power_500ss_wc}
\end{figure}

\textcolor{black}{To investigate sensitivity to subtle directional changes, we conducted an
additional simulation study under the contiguous alternatives
$\mu_{2,n}
=
\left(
\mu_1+\frac{h}{\sqrt n}
\right)\bmod 2\pi,
\quad h\in\mathbb{R}.$
Here, \(h\) is a fixed local-alternative parameter rather than a
circular angle; negative and positive values represent changes in opposite
directions around the circle.
We considered
$n\in\{200,500,1000\},\quad
\mu_1\in
\left\{
0,\frac{\pi}{6},\frac{\pi}{3},\pi
\right\},
\quad
\kappa\in\{1,3\},$
with the changepoint located at
$k^*=\lfloor n/2\rfloor$
and
$h\in\{-6,-5.5,\ldots,5.5,6\}.$
The proposed SAMC test was compared with the known-concentration von Mises generalized likelihood ratio test (GLRT) due to \cite{ghosh1999change}. The latter is based on the segment resultant lengths and hence on the
von Mises sufficient statistics
$\left(
\sum_i\cos\theta_i,\sum_i\sin\theta_i
\right).$}

\textcolor{black}{%
Separate empirical null critical values were computed for each procedure
and each simulation configuration using \(10000\) null replications.
The empirical size at \(h=0\) was subsequently evaluated using the
\(2500\) Monte Carlo replications.
Across the \(24\) combinations of sample size, baseline direction, and
concentration, the mean empirical rejection probabilities were \(0.0503\)
for SAMC and \(0.0502\) for the GLRT. As summarized in
Table-\ref{tab:null_calibration}, both procedures were therefore
satisfactorily calibrated at the nominal \(5\%\) significance level.
For each simulation configuration and each value of \(h\), power was
approximated using \(2500\) Monte Carlo replications. The local-power curves
remained broadly stable as \(n\) increased because the magnitude of the
angular change decreased at the contiguous-alternative rate \(n^{-1/2}\).
}

\begin{table}[htbp]
\centering

\label{tab:null_calibration}
\begin{tabular}{lcc}
\hline
Procedure
& Mean empirical rejection probability
& Nominal level \\
\hline
SAMC & \(0.0503\) & \(0.05\) \\
GLRT & \(0.0502\) & \(0.05\) \\
\hline
\end{tabular}
\caption{Empirical null calibration of SAMC and the GLRT. Separate
empirical critical values were computed for each procedure and simulation
configuration using \(10000\) null replications. The reported rejection
probabilities are averages over the \(24\) simulation configurations and
were evaluated using \(2500\) replications at \(h=0\).}
\end{table}

\textcolor{black}{%
Table-\ref{tab:local_power_n1000} summarizes the local-power results for
\(n=1000\) and \(\lvert h\rvert=3\). Each reported value is the average
of the estimated powers obtained under \(h=-3\) and \(h=3\). In the
diffuse setting, \(\kappa=1\), SAMC and the GLRT exhibited broadly
comparable power. At \(\mu=0\), SAMC attained a power of \(0.1020\),
compared with \(0.0832\) for the GLRT, corresponding to a gain
of \(0.0188\). At \(\mu=\pi/6\), the two procedures produced almost
identical powers, namely \(0.0954\) for SAMC and \(0.0960\) for the GLRT.
The GLRT attained moderately higher power at \(\mu=\pi/3\), with an
absolute difference of \(0.0176\), while its advantage at \(\mu=\pi\)
was only \(0.0050\).
}

\begin{table}[htbp]
\centering

\label{tab:local_power_n1000}
\begin{tabular}{ccccc}
\hline
\(\kappa\)
& Baseline direction \(\mu\)
& SAMC
& GLRT
& Difference \\
\hline
1 & \(0\)       & \(0.1020\) & \(0.0832\) & \( 0.0188\) \\
1 & \(\pi/6\)   & \(0.0954\) & \(0.0960\) & \(-0.0006\) \\
1 & \(\pi/3\)   & \(0.0768\) & \(0.0944\) & \(-0.0176\) \\
1 & \(\pi\)     & \(0.0898\) & \(0.0948\) & \(-0.0050\) \\
\hline
3 & \(0\)       & \(0.3918\) & \(0.3538\) & \( 0.0380\) \\
3 & \(\pi/6\)   & \(0.4294\) & \(0.3620\) & \( 0.0674\) \\
3 & \(\pi/3\)   & \(0.4764\) & \(0.3616\) & \( 0.1148\) \\
3 & \(\pi\)     & \(0.3504\) & \(0.3464\) & \( 0.0040\) \\
\hline
\end{tabular}
\caption{Mean empirical power of SAMC and the GLRT for \(n=1000\) and
\(\lvert h\rvert=3\). The reported powers are averaged over \(h=-3\)
and \(h=3\). The final column is defined as SAMC minus GLRT.}
\end{table}

\textcolor{black}{%
The advantage of SAMC was more pronounced for concentrated angular data
with \(\kappa=3\). At \(\mu=0\), SAMC achieved an average power of
\(0.3918\), compared with \(0.3538\) for the GLRT, corresponding to an
absolute gain of \(0.0380\). More substantial improvements were observed
at \(\mu=\pi/6\) and \(\mu=\pi/3\), where SAMC exceeded the GLRT by
\(0.0674\) and \(0.1148\), respectively. The largest gain occurred at
\(\mu=\pi/3\), for which the average powers of SAMC and the GLRT were
\(0.4764\) and \(0.3616\), respectively. At \(\mu=\pi\), the corresponding
powers were \(0.3504\) and \(0.3464\), indicating that the two procedures
performed almost identically, with a small absolute advantage of \(0.0040\)
for SAMC.
}

\textcolor{black}{%
Overall, these results indicate that SAMC is highly competitive with the
sufficient-statistic likelihood benchmark and can provide appreciable power
gains for concentrated angular distributions. Nevertheless, its relative
performance varies with the baseline direction and sample size, and hence
no claim of uniform dominance is made. A further practical advantage of
SAMC is that its construction does not require the specification of a
parametric circular distribution, whereas the GLRT considered here relies
on the von Mises likelihood. For clarity,  Figure-\ref{local_power_1000ss_vm_comp} presents
the local-power comparison for \(n=1000\). The complete collection of
local-power curves for
\(n\in\{200,500,1000\}\), all retained baseline directions
$\mu\in\left\{0,\frac{\pi}{6},\frac{\pi}{3},\pi\right\},$
and both concentration levels
$\kappa\in\{1,3\},$
are reported in the appendix. These additional
Figures-\ref{fig:local_power_kappa1_all_n} \& \ref{fig:local_power_kappa3_all_n} in the appendix confirm that the relative behaviour of SAMC and the GLRT remains
broadly stable across the considered sample sizes.
}

 \begin{figure}[h!]
	\centering
{\includegraphics[trim= 0 0 0 40, clip, width=0.99\textwidth, height=0.5\textwidth]{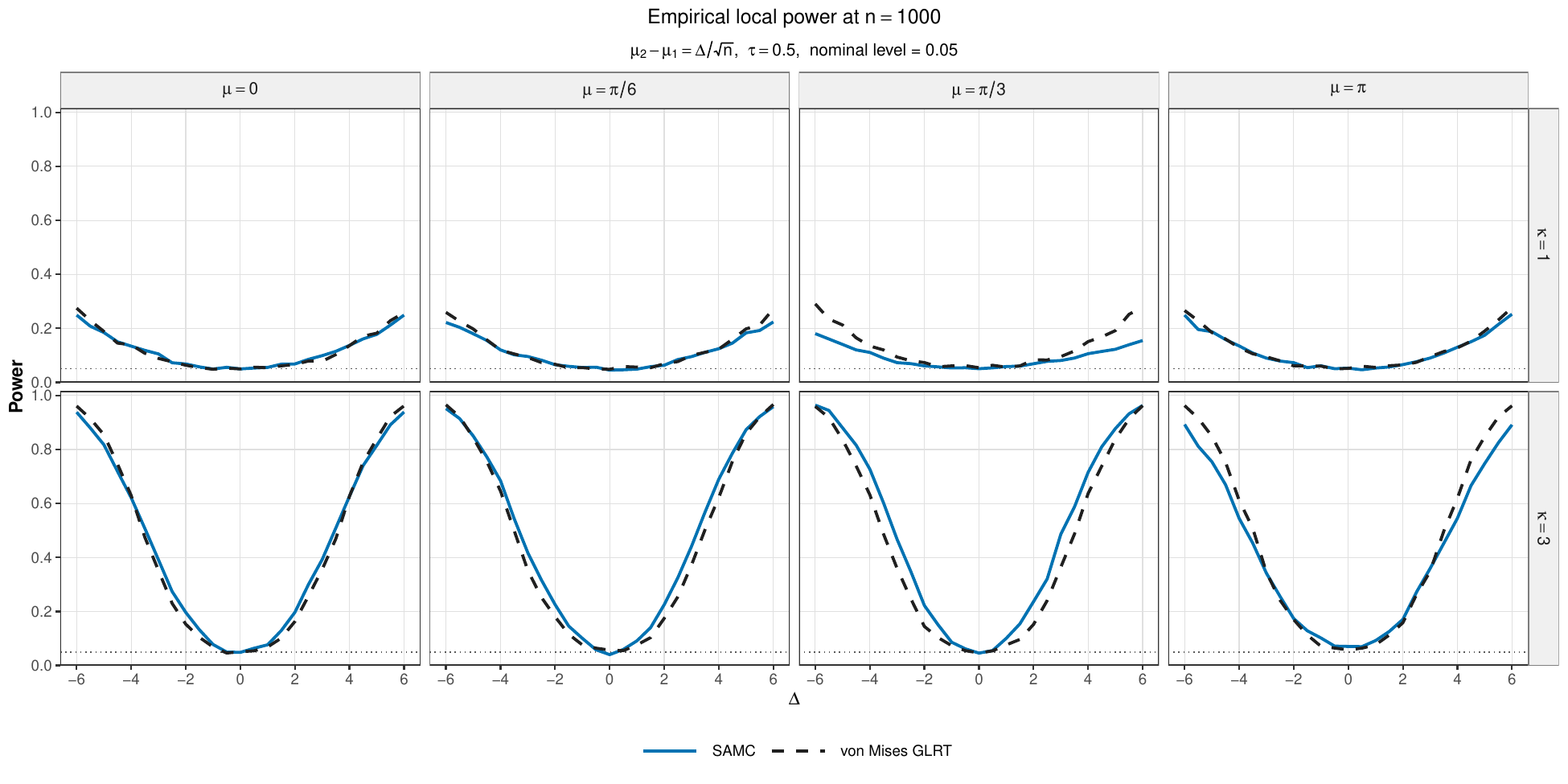}}
	 \textcolor{black}{\caption{
Empirical local-power curves at the \(5\%\) significance level under
the contiguous alternatives, with the changepoint located at
\(k^*=\lfloor n/2\rfloor\). The proposed SAMC test is compared with the
known-concentration von Mises GLRT based on the segment resultant lengths.
The columns correspond to
\(\mu=0,\pi/6,\pi/3,\pi\), and the rows correspond to
\(\kappa=1\) and \(\kappa=3\). The displayed results use \(n=1000\);
results for \(n=200\) and \(n=500\) are provided in the appendix.
}
\label{local_power_1000ss_vm_comp}}
\end{figure}

\subsection {Power comparison} 
\label{comarison_sec}
We compare the proposed changepoint detection method for mean direction shifts in circular data with two established approaches. The first one is a rank-based test (say, LRBT) for circular data proposed by \cite{lombard1986change}, and the other one is a test for non-Euclidean data using the graph-based binary segmentation method (gSeg) introduced by \cite{chen2015graph}.
 To illustrate that the proposed test, which utilizes the square of an angle, surpasses an alternative method based on circular arc length, we compare them with respect to their power curves.  The circular arc length distance between two angles, \(\theta_1\) and \(\theta_2\), is defined as follows:
 \cite[see][pp. 18-19]{Mardia_2000}:  
\[
\min[\theta_1-\theta_2, 2\pi - (\theta_1-\theta_2)] = \pi - \lvert \pi - \lvert \theta_1-\theta_2 \rvert \rvert.
\]
In this study, we measure the distance from \(0\) to each angle \(\theta_1, \dots, \theta_n\). Using the above equation, the distance can be expressed as:  
\begin{equation}
	c_i = \pi - \lvert \pi - \lvert \theta_i \rvert \rvert, \quad \text{for } i = 1, \dots, n.
	\label{arc_length_distance}
\end{equation}
 By substituting \(a_i\) with \(c_i\) in Equation-\ref{cusum process} and applying the relevant sample mean and variance, we derive a new test statistic, \(C_n\).  Given that \(c_i\) are independent and identically distributed sequences, the distribution under \(H_0\) remains pivotal. 
 Figures \ref{power_100ss_vm_comp}(a) \& (b) illustrate the comparison of power curves among the proposed SAMC test, the LRBT test, the gSeg test, and the arc length-based test at significance levels of \(1\%\) and \(5\%\).  For each level, \(5 \times 10^3\).  Monte Carlo iterations were conducted using a sample size of \(n = 500\), with observations drawn from a von Mises distribution with a concentration parameter \(\kappa = 3\).  The equispaced mean differences were constrained within the interval \(-\frac{\pi}{2} \leq \delta_\mu = (\mu_2 - \mu_1) \leq \frac{\pi}{2}\), while the true changepoint was considered at \(k^* = \frac{n}{2}\).  
 We also analysed the performance of these four tests: SAMC,  LRBT, gSeg, and the arc length-based, across various true changepoint locations \(k^* \in \{50, 100, 150, 200, 250, 300, 350, 400, 450\}\) in the alternative setting \(H_1\), maintaining a constant mean shift of \(\delta_\mu = \frac{\pi}{10}\).  For each \(k^*\), a total of \(5 \times 10^3\) iterations were performed using \(n = 500\) samples drawn from the von Mises distribution with \(\kappa = 3\).  Figures \ref{power_location_change}(a) and \ref{power_location_change}(b) demonstrate that, at both \(1\%\) and \(5\%\) significance levels, the proposed statistic \(\mathbb{M}_{n}\) consistently exhibits superior power compared to all competing methods. 
 The results indicate that the proposed test outperforms both existing tests and a baseline test based on the circular arc length metric. \textit{From the extensive simulation studies, it is observed that the proposed test is more powerful than the existing ones when the mean direction shift is between $[-\pi/2,~ \pi/2]$. Beyond that range, its performance may fall, and hence we consider the proposed test as a more powerful test.}

\begin{figure}[h!]
	\centering
	\subfloat[ level $1\%$  .]{%
		{\includegraphics[trim= 0 0 0 40, clip, width=0.32\textwidth, height=0.32\textwidth]{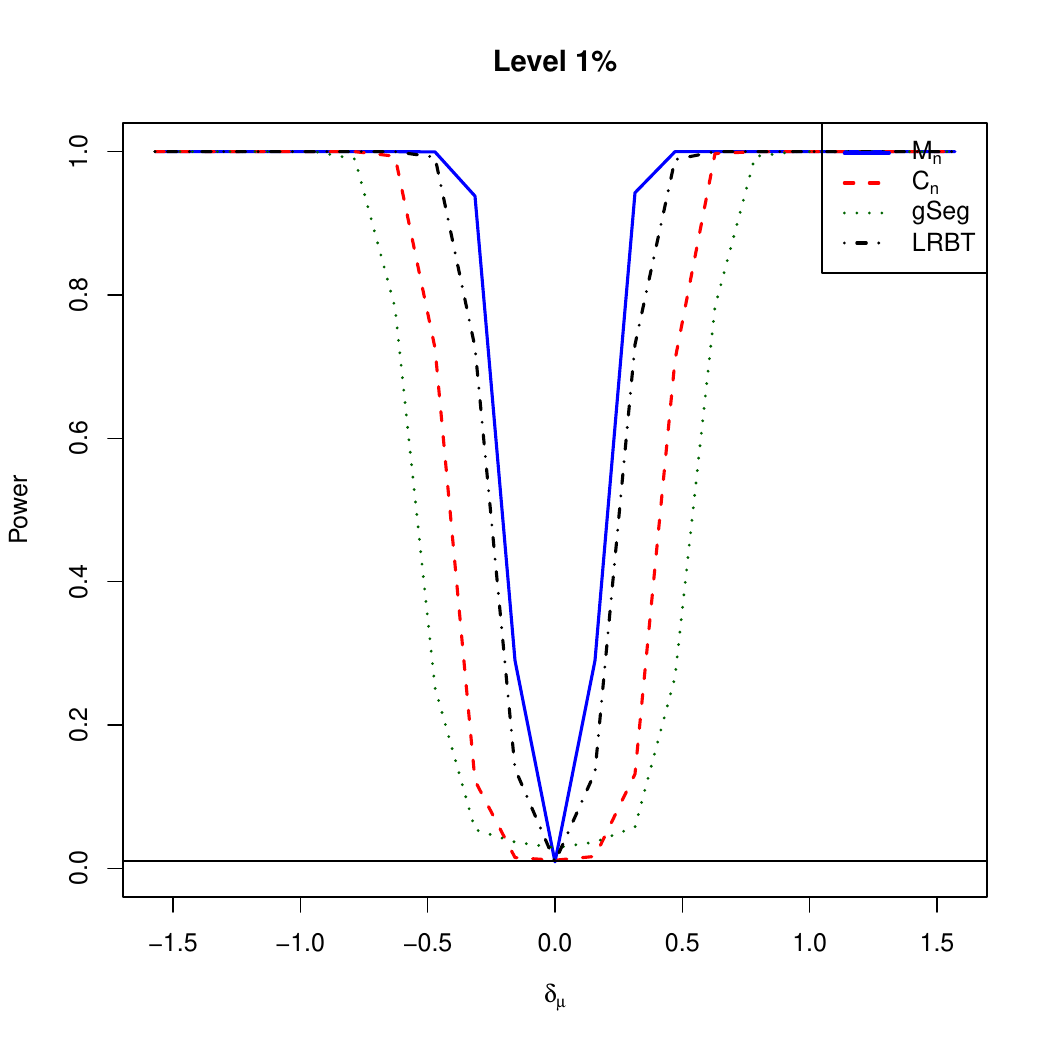}}}\hspace{5pt}
	\subfloat[ level $5\%$.]{%
		{\includegraphics[trim= 0 0 0 40, clip, width=0.32\textwidth, height=0.32\textwidth]{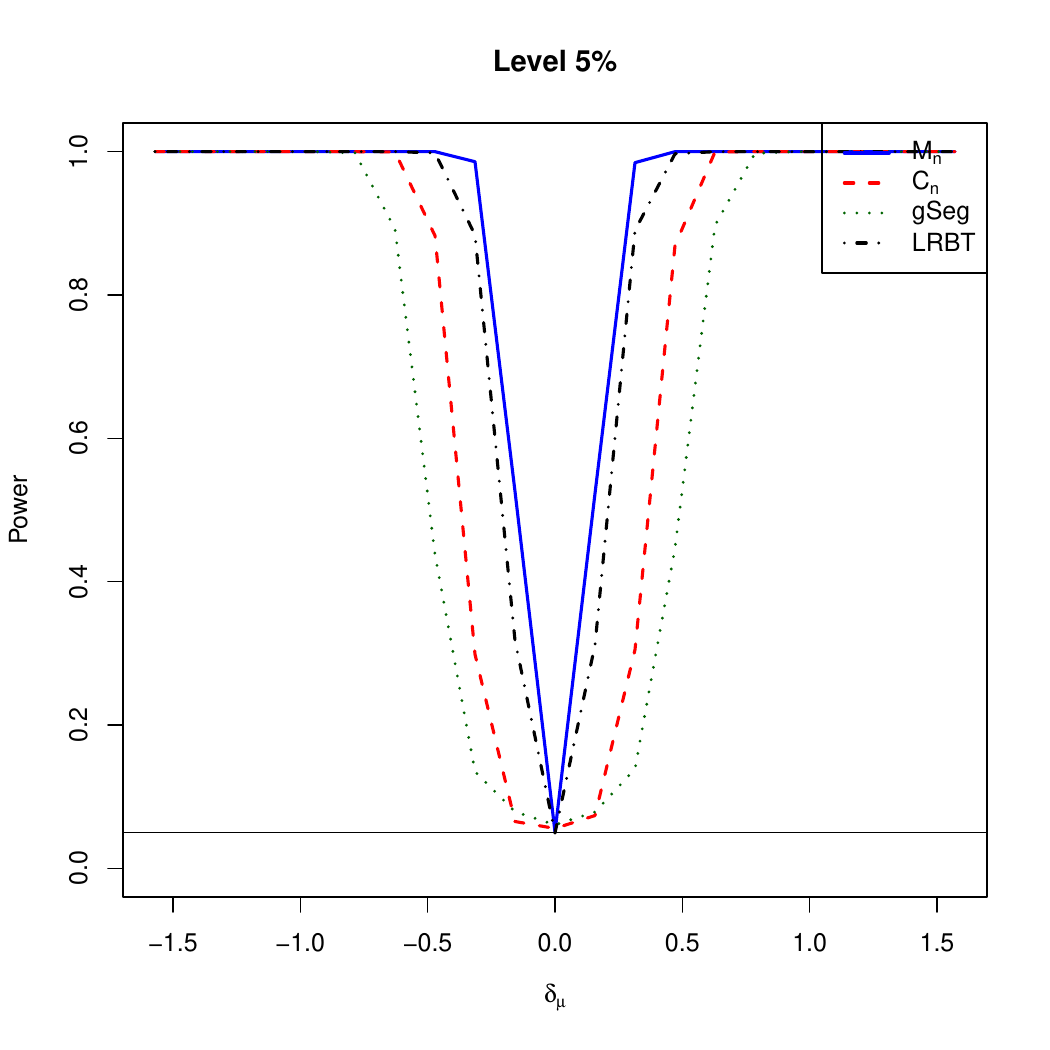}}}
	\vspace{5pt} 
	\caption{Plot of power functions of the proposed SAMC test, the LRBT test, the gSeg test, and the circular arc based test for sample of size $n=500$ from von Mises distribution with the concentration parameter  $\kappa=3$, and equispaced mean direction differences in  $-\frac{\pi}{2}\leq \delta_{\mu}=(\mu_2-\mu_1)\leq \frac{\pi}{2},$  and the true changepoint at $k^*=\dfrac{n}{2}$ under $H_{1}$.}
	\label{power_100ss_vm_comp}
\end{figure}


\begin{figure}[h!]
	\centering
	\subfloat[ level $1\%$  .]{%
		{\includegraphics[trim= 0 0 0 40, clip, width=0.322\textwidth, height=0.32\textwidth]{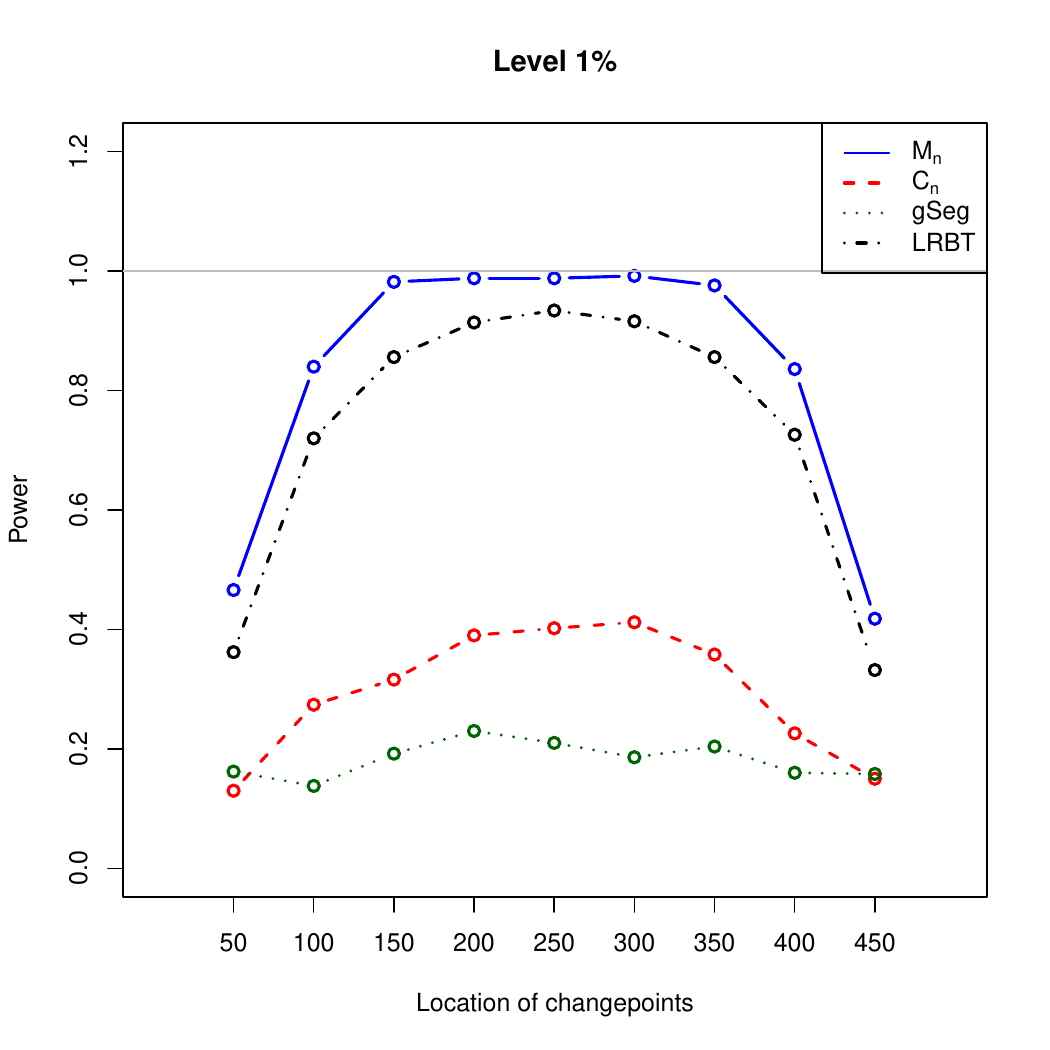}}}\hspace{5pt}
	\subfloat[ level $5\%$.]{%
		{\includegraphics[trim= 0 0 0 40, clip, width=0.322\textwidth, height=0.32\textwidth]{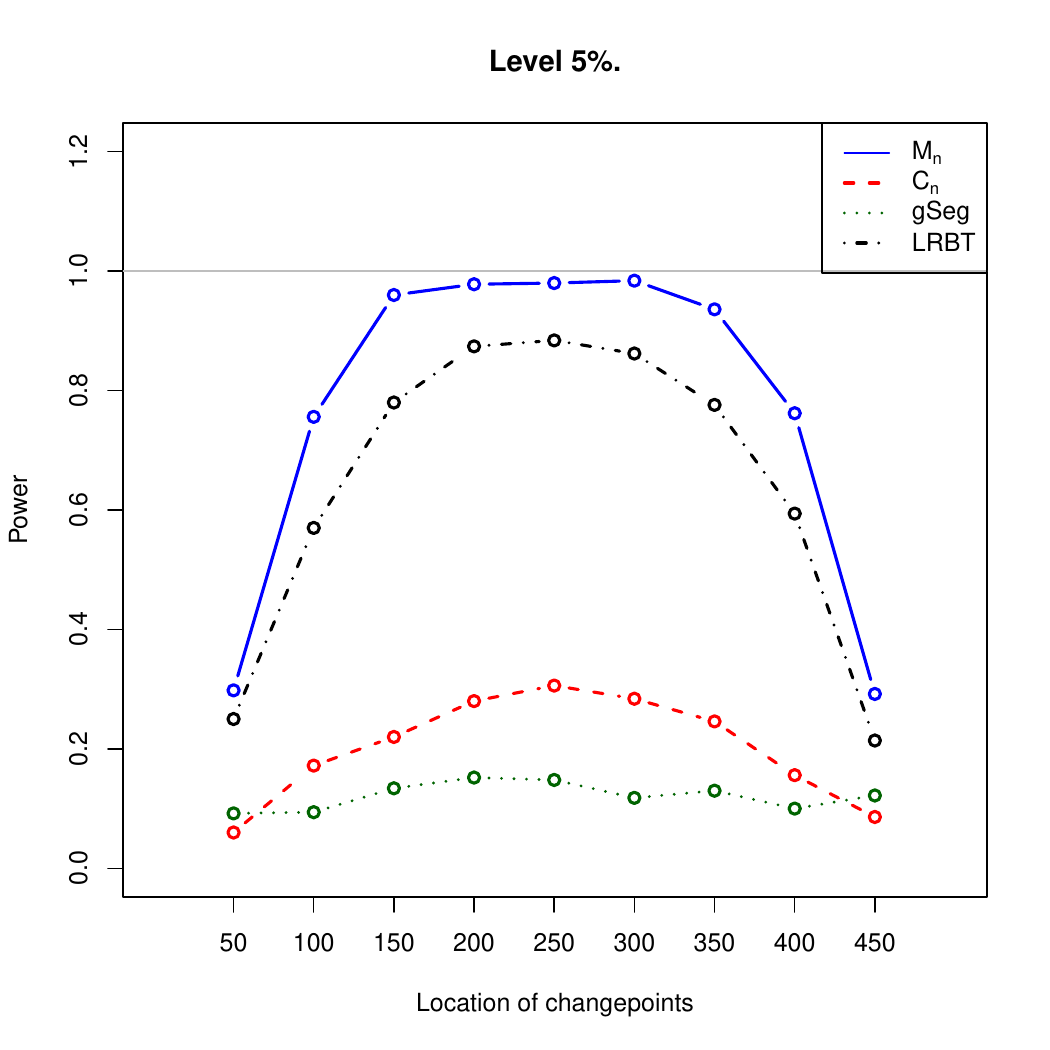}}}
	\vspace{5pt} 
	\caption{Power comparison between the proposed SAMC test, the LRBT test, the gSeg test, and the circular arc based test for sample of size $n=500$ with the different locations of the true changepoints at  $k^*=50,100,150,200,250,300,350,400$, or $450$ under $H_{1}$ with shift in mean direction $\delta_\mu=\frac{\pi}{10}$ for the data following von Mises distribution with concentration parameter $\kappa=3$.  }
	\label{power_location_change}
\end{figure}

\section{Data analysis}
\label{data_analysis}

Now, we apply our newly proposed distribution-free method to two popular cryptocurrencies, Bitcoin and Ethereum, and gold. To test our proposed changepoint detection technique, we collected a Bitcoin per-minute historical dataset from Kaggle (source: \url{https://www.kaggle.com/datasets/prasoonkottarathil/btcinusd}). 
This dataset includes one-minute historical data from January 1, 2017, to March 1, 2022, resulting in 1860 daily observations after processing. We collected a per-minute historical dataset for Ethereum from Kaggle: \url{https://www.kaggle.com/datasets/prasoonkottarathil/ethereum-historical-dataset?select=ETH_1min.csv}
which covers the period from May 9, 2016, to April 16, 2020, yielding 1409 daily observations after processing.  Similarly, the dataset used for gold in this study is available at the following link: \url{https://www.kaggle.com/datasets/novandraanugrah/xauusd-gold-price-historical-data-2004-2024?select=XAU_1m_data.csv}, which consists of multi-resolution historical price records (2004–2025). We have used per-minute data from January 1, 2019, to December 31, 2024. Here we get a total $1863$ observation after processing.

\textcolor{black}{
For each dataset, the proposed SAMC test was first applied to the complete
sequence of daily angular timestamps. The statistic was evaluated over all
candidate locations \(k=1,\ldots,n-1\), and a changepoint was reported only
when the full-sequence statistic exceeded the corresponding finite-grid
critical value at the chosen significance level. Whenever a significant
changepoint was detected, the same SAMC test was applied once to each of the
two resulting subsegments as a descriptive post-detection stability check.
In all reported analyses, neither subsegment showed evidence of an
additional significant changepoint, and no further recursive subdivision
was performed. We emphasize that this one-step diagnostic is not proposed
as a general multiple-changepoint procedure. The theoretical and
methodological contributions of the paper concern the single-changepoint
problem, while a complete multiple-changepoint framework, including
selection of the number of changepoints, minimum-segment-length conditions,
penalized criteria such as BIC or MDL, and sequential error control, is left
for future research.
}

\textcolor{black}{Before applying the proposed methodology, we examined each angular
timestamp sequence using the runs test of
Wald and Wolfowitz~\cite{wald1940test}. For the daily lowest-price
timestamps, the resulting \(p\)-values were \(0.0949\) for Bitcoin,
\(0.3241\) for Ethereum, and \(0.1714\) for Gold. Thus, the null
hypothesis of randomness was not rejected at the \(5\%\) significance
level for any of the three sequences.
We emphasize that failure to reject the runs test does not establish
independence, stationarity, or all probabilistic conditions required by
the asymptotic theory. The test provides only a limited diagnostic for
departures from randomness. Consequently, the i.i.d. assumption used in
the theoretical development should be regarded as a working approximation
for the present financial application. These tests were conducted using the \texttt{R} package \texttt{tseries} ~(\url{https://doi.org/10.32614/CRAN.package.tseries}).}
\\

\noindent
\textbf{Empirical Findings and Market Context:}
For the  Bitcoin dataset,  the findings associated with the timestamp of occurrence of the lowest price are reported in Table-\ref{table:data_bitcoin_cp_table}, where the last column indicates the mean direction of the significant segments. The p-value has been computed with respect to the limiting distributions $K_{\infty}^{(1860)}$. From the table we see that the significant changepoint at index 970 falls end of August 2019. This timing coincides with a major market-structure event: the launch of \emph{Binance Futures} (September 13, 2019), which rapidly gained market share in perpetual swaps 
(see: \url{https://www.binance.com/en/blog/all/390692819288195072}). 
The increased role of derivative trading plausibly altered the timing of intraday liquidity and funding cycles, thereby shifting the circular mean of the occurrence of the time of daily minimum. 
For the Ethereum dataset, the findings associated with the timestamp of occurrence of the lowest price are reported in Table-\ref{table:data_ethereum_cp_table}, with the last column indicating the mean direction of the significant segments. The p-values are computed with respect to the limiting distribution $K_{\infty}^{(1409)}$.
The significant changepoint at index 896 maps to the last week of October 2018. This period overlaps with the sharp cryptocurrency market drawdown of mid-November 2018, when major assets, including ETH reached multi-year lows 
(see: \url{https://www.coindesk.com/markets/2018/11/19/bitcoin-price-hits-13-month-low-as-crypto-market-slumps}).
For the Gold price dataset, the results corresponding to the timestamps of the lowest prices are summarized in Table-\ref{table:data_gold_cp_table_low}, where the last column reports the mean direction of each significant segment. The p-values are evaluated with respect to the limiting distribution $K_{\infty}^{(1863)}$. A significant changepoint is detected at index 1310 (p-value = 0.0305), which corresponds to the last week of March 2023. This period aligns with a phase of heightened volatility in global commodity markets, during which gold prices experienced substantial fluctuations amid inflation concerns. The finding closely matches the Federal Reserve's FOMC decision to raise rates (see: \url{https://www.federalreserve.gov/newsevents/pressreleases/monetary20230322a1.htm}). Also, this period aligns with the collapse of Silicon Valley Bank  (\url{ https://www.cnbc.com/2023/03/22/feds-powell-says-svb-collapse-may-slow-the-economy-through-tighter-credit.html}).

\begin{figure}[h!]
	\centering
	\subfloat[ ]{%
		{\includegraphics[width=0.4 \textwidth, height=0.3\textwidth]{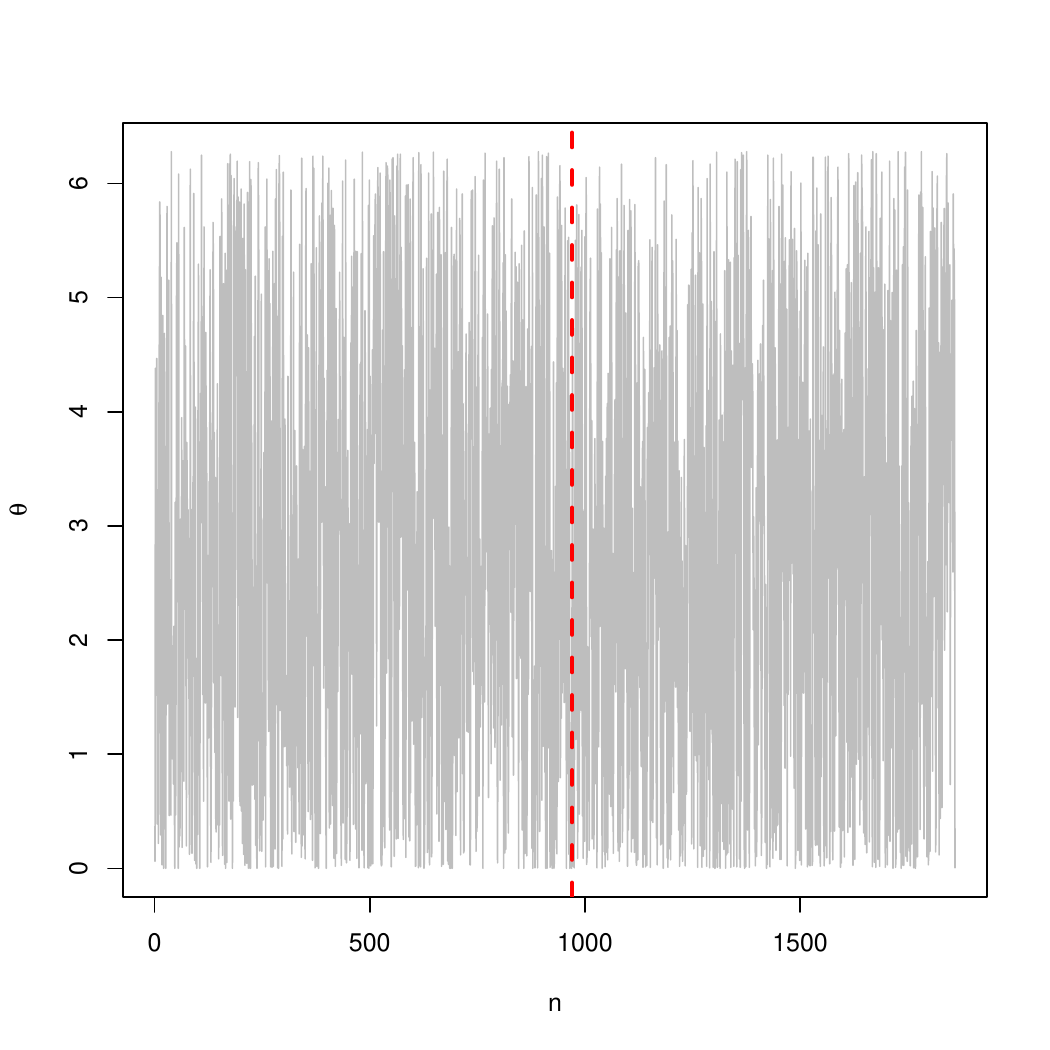}}}
	\subfloat[]{%
		{\includegraphics[width=0.4 \textwidth, height=0.3\textwidth]{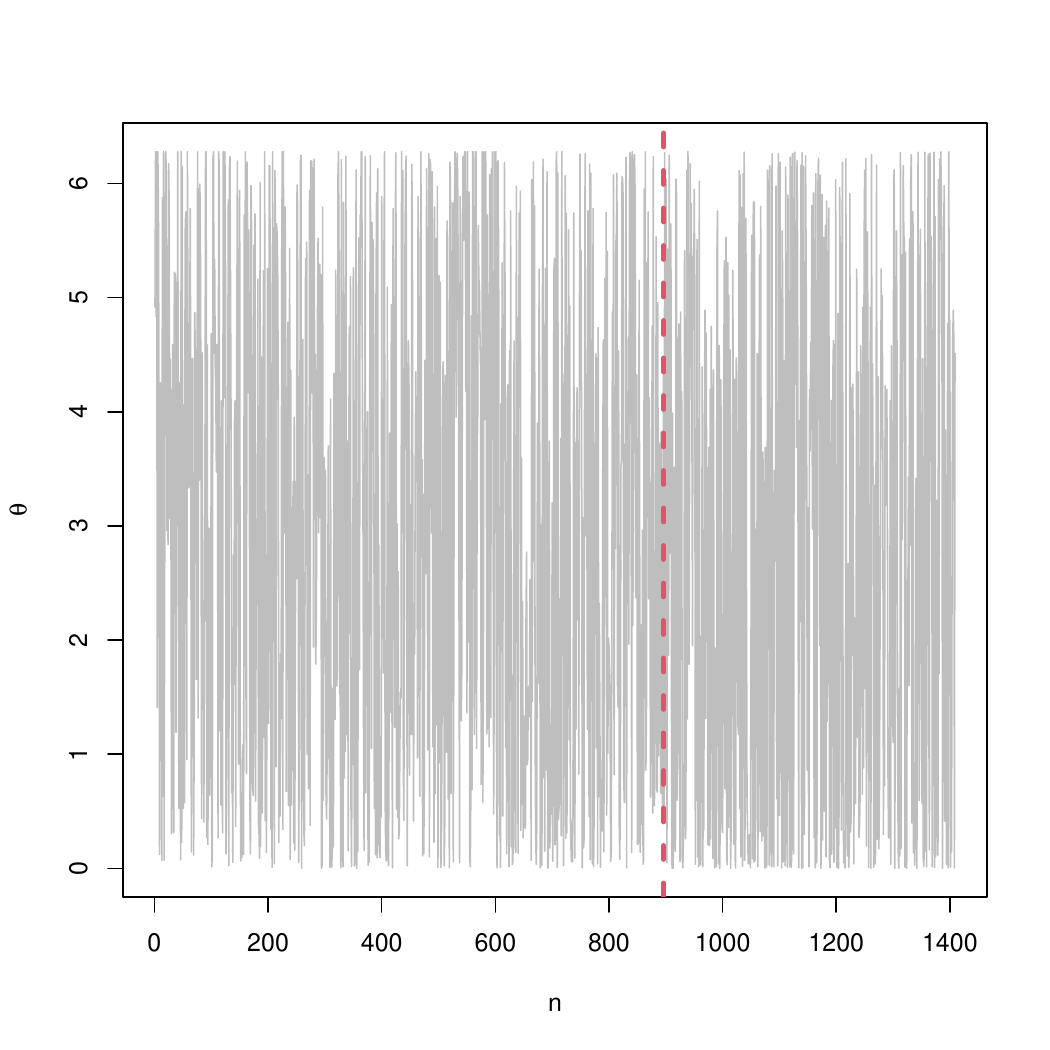}}}

	\subfloat[ ]{%
		{\includegraphics[width=0.4 \textwidth, height=0.3\textwidth]{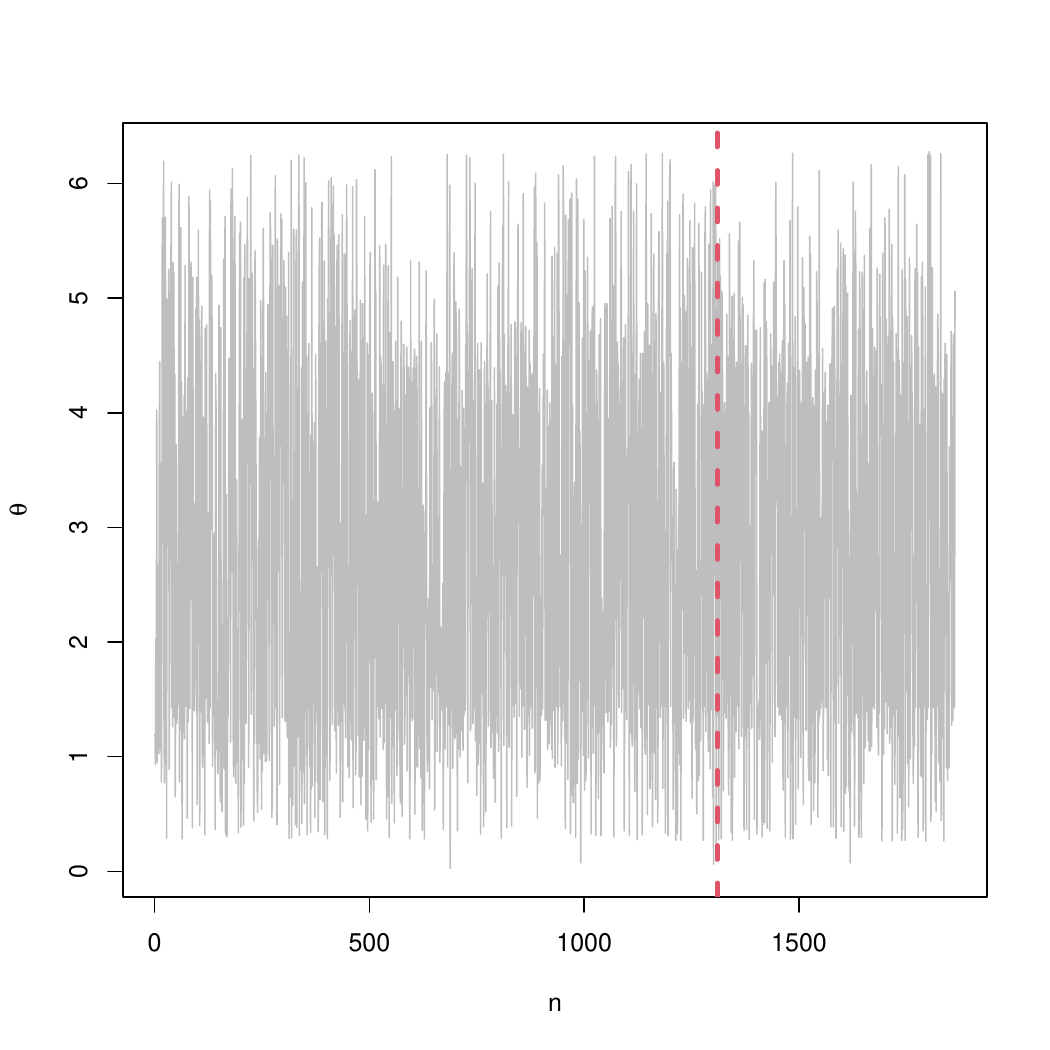}}}
	\subfloat[]{%
		{\includegraphics[width=0.4 \textwidth, height=0.3\textwidth]{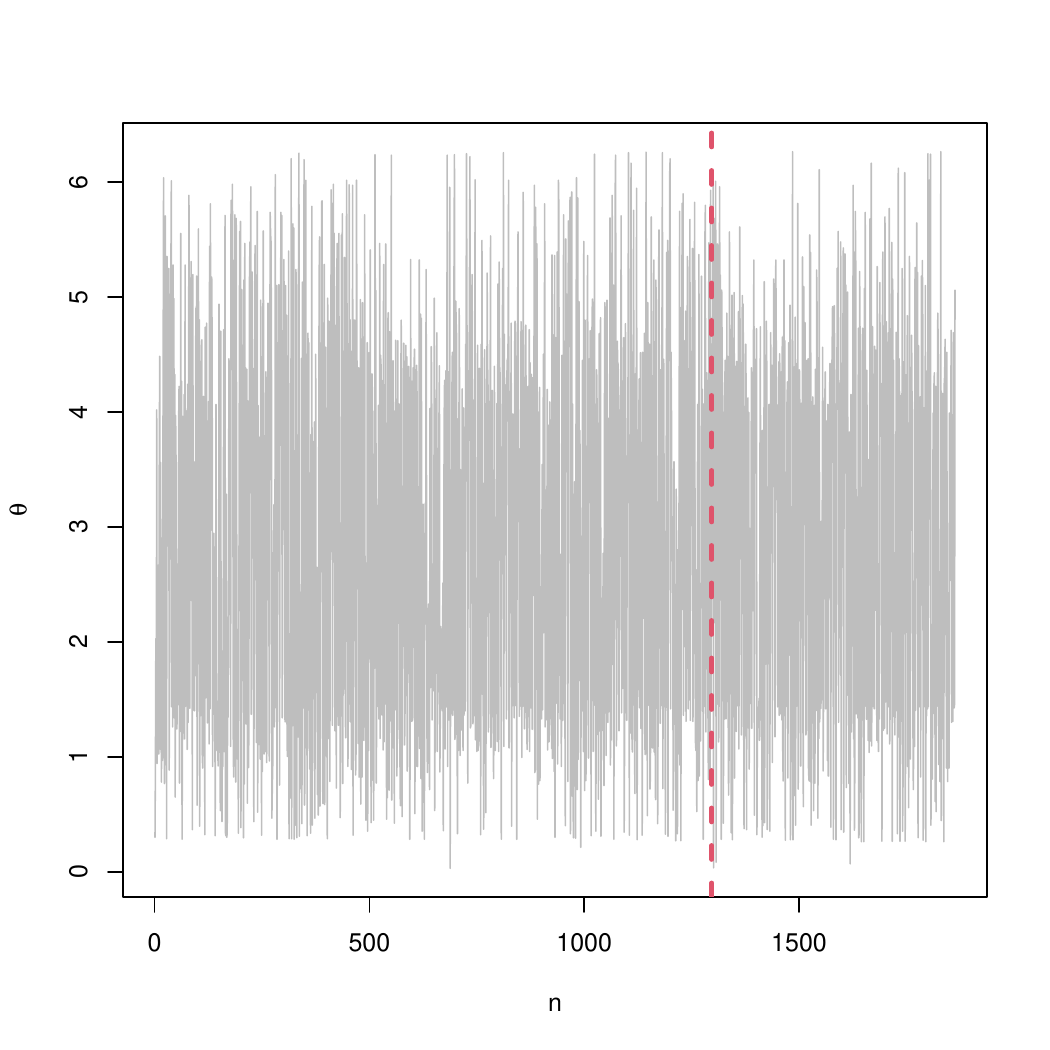}}}
\caption{
Temporal plots of the angular timestamps used in the empirical analysis:
(A) daily lowest-price timestamps for Bitcoin;
(B) daily lowest-price timestamps for Ethereum;
(C) daily lowest-price timestamps for Gold; and
(D) daily highest-price timestamps for Gold.
The red dashed vertical line in each panel denotes the changepoint estimated by the SAMC procedure.}

\label{data_analysis_bit_eth}
\end{figure}

\begin{figure}[h!]
	\centering
    	\subfloat[ ]{%
		{\includegraphics[width=6.15cm, height=6cm]{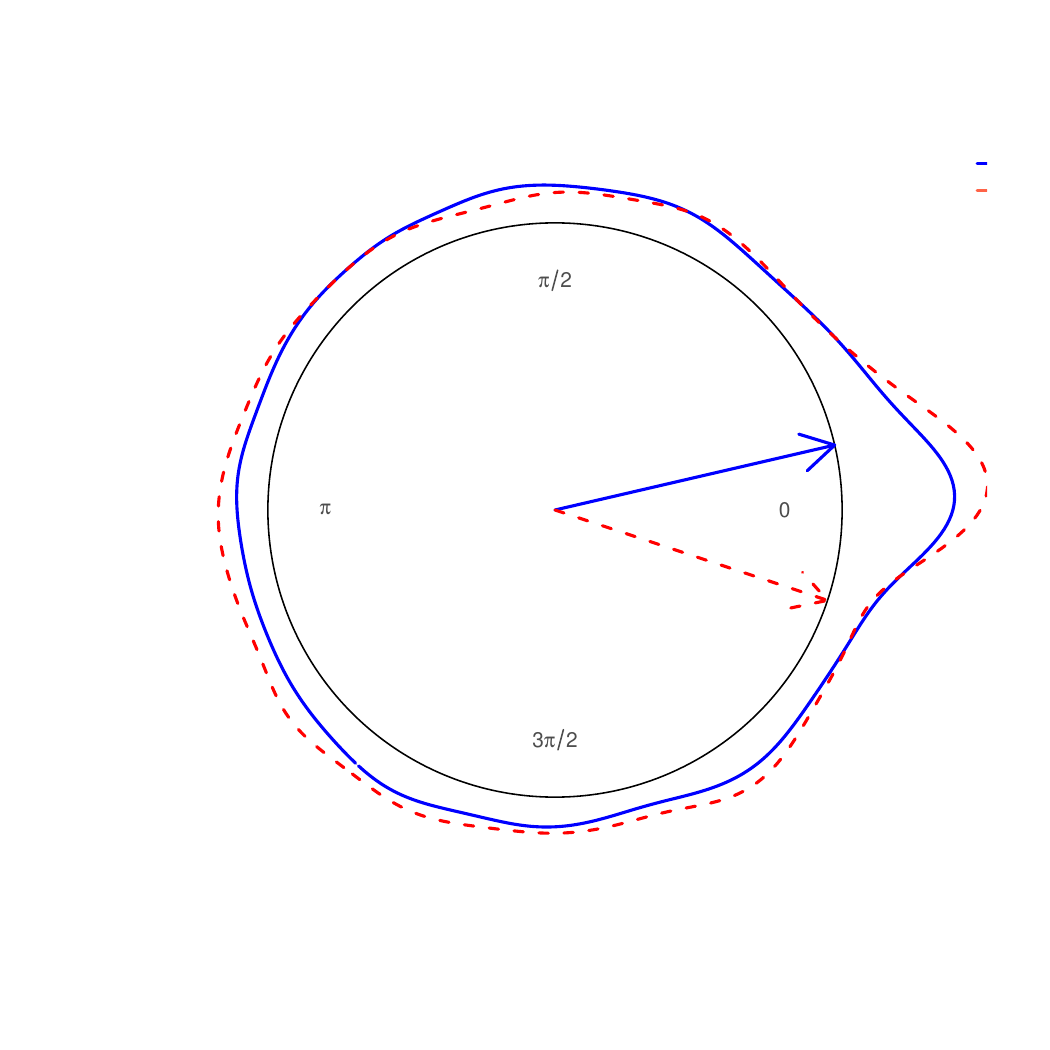}}}
	\subfloat[]{%
		{\includegraphics[width=6.15cm, height=6cm]{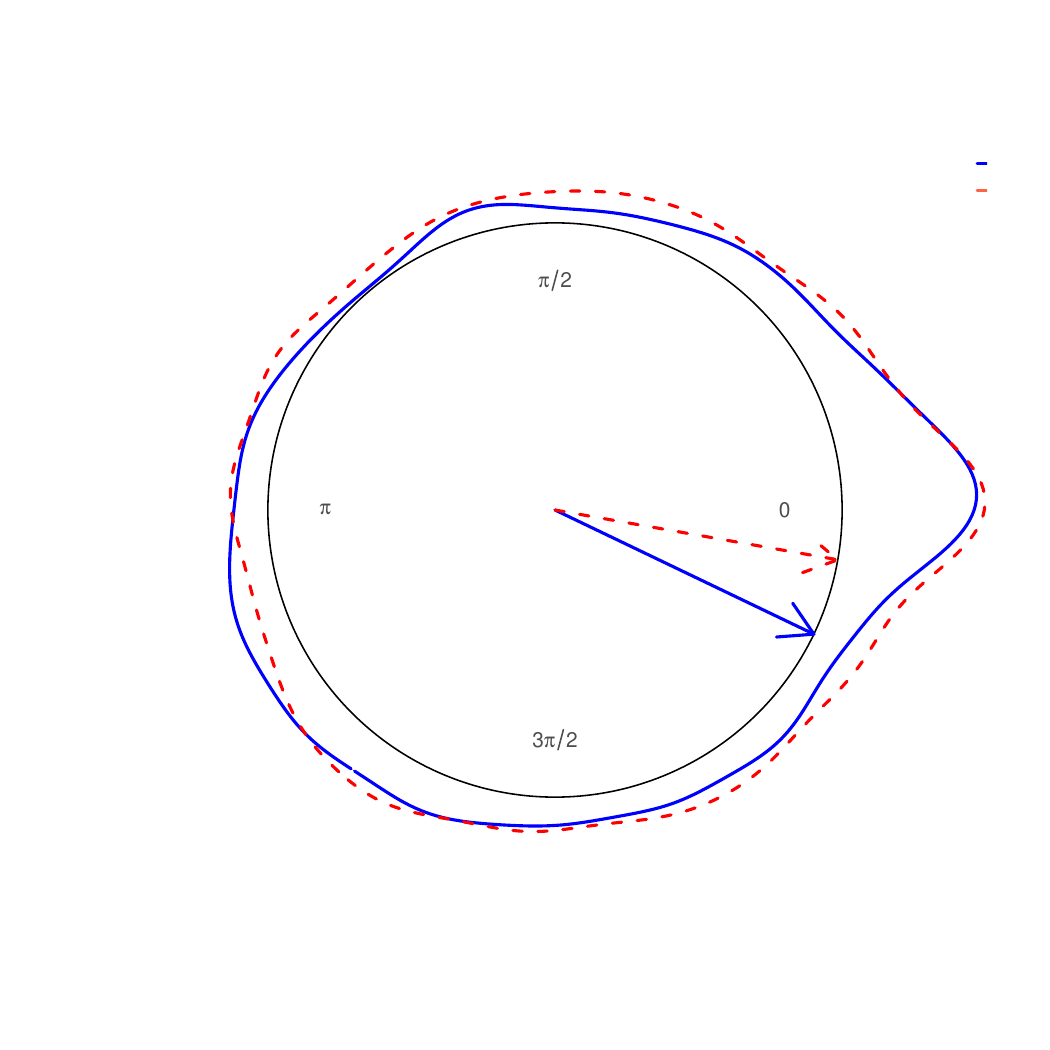}}}

	\subfloat[ ]{%
		{\includegraphics[width=6.15cm, height=6cm]{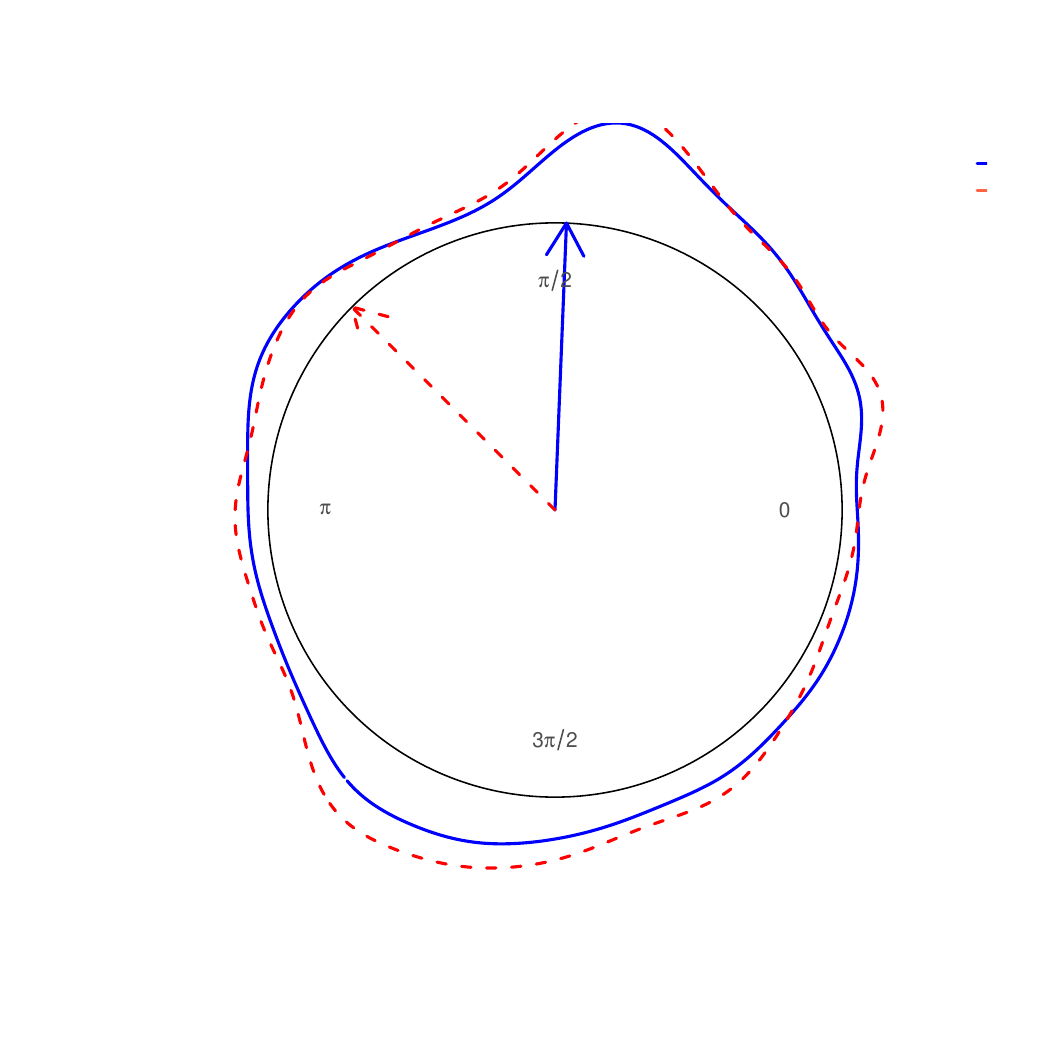}}}
	\subfloat[]{%
        {\includegraphics[width=6.15cm, height=6cm]{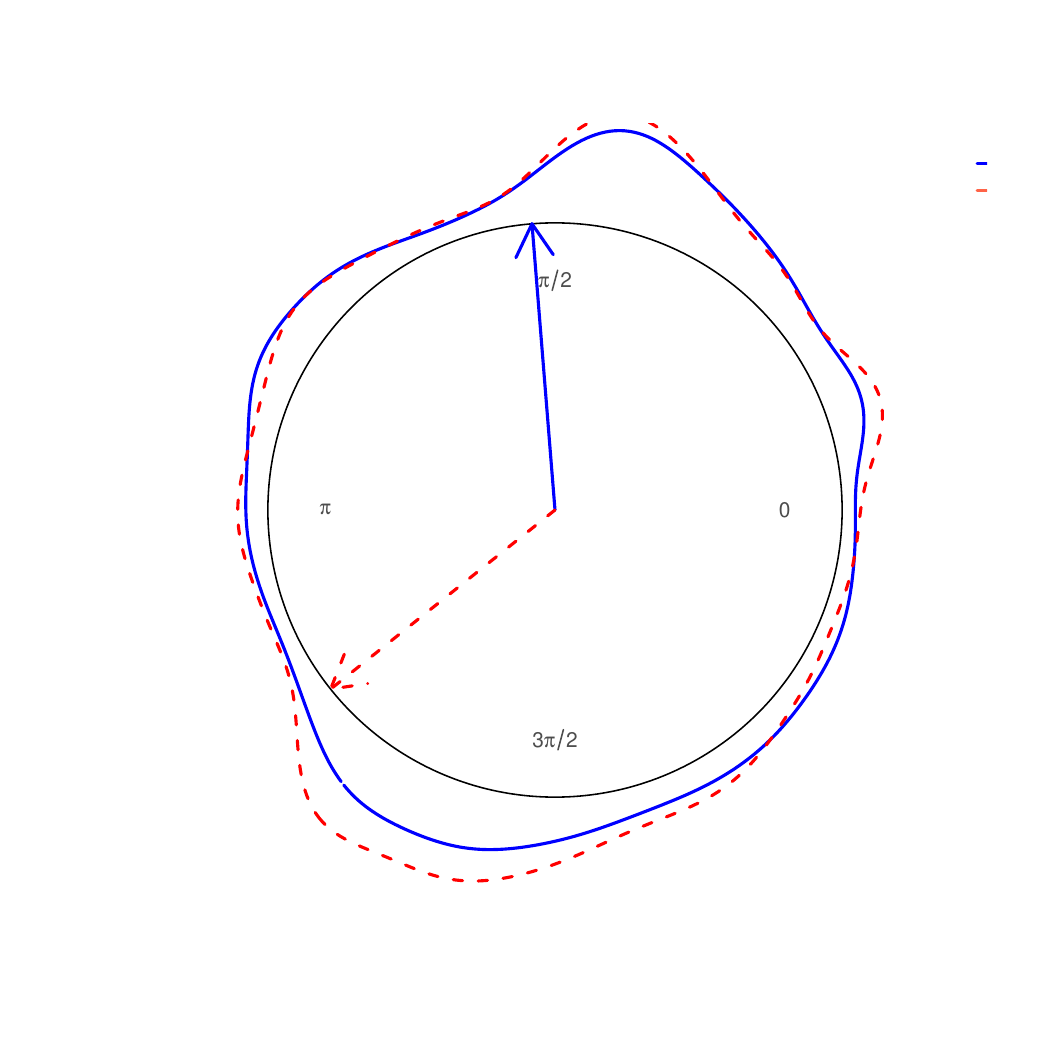}}}\hspace{5pt}
\textcolor{black}{\caption{
Circular density estimates before and after the changepoint detected by
SAMC: (A) Bitcoin daily lowest-price timestamps;
(B) Ethereum daily lowest-price timestamps;
(C) Gold daily lowest-price timestamps; and
(D) Gold daily highest-price timestamps.
The solid blue and dashed red curves represent the pre- and
post-changepoint  estimated densities, respectively. The corresponding mean
directions are shown by radial lines using the same line convention.
}
\label{fig:circular_density_cp}}
\end{figure}


\textcolor{black}{
Figure-\ref{data_analysis_bit_eth} presents the observed angular timestamps and
the estimated changepoint locations. Panels (A)--(C) correspond to the
timestamps of the daily lowest prices of Bitcoin, Ethereum, and Gold,
respectively, while panel (D) corresponds to the timestamps of the daily
highest Gold prices. In each panel, the estimated changepoint is indicated
by a red dashed vertical line.
Figure-\ref{fig:circular_density_cp} presents the corresponding circular
density estimates before and after the detected changepoint. Panels
(A)--(C) correspond to the daily lowest-price timestamps for Bitcoin,
Ethereum, and Gold, respectively, and panel (D) corresponds to the daily
highest-price timestamps for Gold. The solid blue curve represents the
pre-changepoint density, whereas the dashed red curve represents the
post-changepoint density. The corresponding mean directions are indicated
by radial lines using the same line convention.
}

When applying the propose method to the timestamps of the highest price occurrences in Bitcoin and Ethereum, there is no significant changepoint found in either of the datasets. While applying the proposed method to the timestamps of the highest price occurrences in the Gold price dataset, a significant changepoint is detected at index 1297 (p-value = 0.0181). This corresponds to the first week of March 2023. Figure-\ref{data_analysis_bit_eth}(D) and Table-\ref{table:data_gold_cp_table_high} show the result associated with the timestamps of the highest price of the Gold price dataset, while panel (H) shows the circular density plot for the daily highest prices in the Gold price dataset along with the mean direction before and after the changepoint.

We emphasize correlation rather than causation: the changepoints are identified purely from circular timestamp data, but their alignment with documented market events highlights the economic interpretability of the proposed methodology.

\section{Limitations and Directions for Future Work}
\label{future work}
While the present work introduces a distribution-free and geometry-driven framework for detecting changes in angular mean direction, several important avenues remain open for future research.  
\subsection*{A. \textit{Dependence structures in circular time series}}  

The asymptotic results developed in this paper assume that the angular
observations are independent and identically distributed. Under this
assumption, the studentized CUSUM process converges to a standard Brownian
bridge, yielding the Kolmogorov limiting distribution.

In financial applications, however, angular timestamps obtained from
time-indexed price processes may exhibit serial dependence. Although the
runs tests reported in Section-\ref{data_analysis} did not reject
randomness for the analyzed sequences, such tests cannot establish
independence, stationarity, or the complete dependence conditions required
by the functional central limit theorem. Accordingly, the i.i.d.
assumption is treated as a working approximation in the present
application.

Extending SAMC to dependent circular observations is therefore an
important direction for future research. Such an extension would require
an appropriate functional limit theorem and a dependence-adjusted
long-run variance estimator. Possible developments may consider
strong-mixing, weakly dependent, or circular time-series processes.



\subsection*{B. \textit{Changepoints in the concentration parameter and general distributional shifts}}  
The present formulation particularly addresses changes in the mean direction, predicated on the assumption of a constant concentration parameter.  In numerous real-world situations, both the dispersion of angular measurements and other characteristics of the underlying distribution may fluctuate, indicating changes in volatility, uncertainty, or the producing mechanism itself.  In financial time series, one may observe sudden shifts from a unimodal distribution to a multimodal distribution, or from a highly concentrated regime to a more dispersed structure.  Expanding the geometric framework to develop companion tests for concentration shifts and, more broadly, for distributional changes beyond the mean direction is a crucial further step.  A promising strategy involves developing joint testing algorithms capable of differentiating changes in mean, concentration, and higher-order distributional characteristics, thereby offering a more thorough instrument for detecting structural fractures in angular data.


\subsection*{C. \textit{\textcolor{black}{Extension to multiple changepoints}}}

The theoretical and methodological framework developed in this paper
concerns the detection and estimation of a single changepoint in the mean
direction of angular observations. In the empirical analysis, whenever the
full-sequence SAMC test detected a significant changepoint, the same test was
applied once to each of the two resulting subsegments as a descriptive
post-detection stability check. Neither subsegment showed evidence of an
additional significant changepoint in the reported applications, and no
further recursive subdivision was performed. These one-step checks should
not be interpreted as a general multiple-changepoint procedure or as a
method for estimating the optimal number of changepoints.

A rigorous extension of SAMC to multiple changepoints remains an important
direction for future research. Such an extension would require a
statistically principled mechanism for determining the number and locations
of changepoints, rather than an ad hoc stopping rule. Possible approaches
include penalized segmentation based on information criteria such as BIC or
MDL, or sequential testing procedures with explicit familywise error-rate
or false-discovery-rate control. A complete formulation would also require
minimum-segment-length conditions, appropriate calibration of the stopping
criterion, and theoretical investigation of consistency and localization
accuracy.

From a computational perspective, the transformed sequence
$a_i=\operatorname{sgn}(\theta_i)A_C^{(0)}(\theta_i)$
may be incorporated into scalable segmentation procedures such as wild
binary segmentation or pruned dynamic-programming methods. However, these
extensions require a separate methodological and theoretical development
and are therefore beyond the scope of the present single-changepoint study.
\section{Conclusion}
\label{conclusion}
This study has emphasized the transformative characteristics of cryptocurrencies as a digital asset category that operates via computer networks. The inherent unpredictability and speculative characteristics of cryptocurrency markets present unique difficulties and possibilities for financial modeling, risk evaluation, and regulatory supervision.  Admitting the possibilities of latent changes in the extreme values of cryptocurrencies, we need precise modeling of the data, which is streaming $24\times 7$. 
We represented the timestamps as circular data and evaluated their volatility using the intrinsic geometry of a torus. In addition, we have proposed a distribution-free test to detect the plausible existence of changepoint in the dataset. It is observed that the distribution of the test statistic follows the Kolmogorov distribution when the data is free from any change. Under the alternative hypothesis, it is proved that the proposed test is consistent.
This method was effectively utilized to analyze the timestamps of extreme prices in Bitcoin and Ethereum, showcasing its practical usability and efficacy. This study not only provides a new distribution-free approach to deal with circular data but also enhances the modeling of the timestamps of events that are important in the financial market. 

Although we illustrate the method using cryptocurrency and Gold data, the scope of the proposed framework extends well beyond financial applications. 
Any setting where observations are recorded on a circular domain can benefit from this approach. 
Examples include environmental sciences (e.g., shifts in prevailing wind or wave directions during climate events), medicine (e.g., changes in circadian rhythm patterns), and transportation studies (e.g., daily traffic flow cycles). 
In each case, detecting structural changes in angular patterns provides insights that conventional linear methods cannot capture. 
Thus, while the cryptocurrency example demonstrates applied relevance in one emerging domain, the broader contribution of this work lies in offering a principled, geometry-based changepoint methodology for circular data.

\section*{Author Contributions} \textbf{Surojit Biswas}:  Conceptualization, formulation, structural and theoretical derivations, implementation, simulations, and manuscript preparation. \textbf{Buddhananda Banerjee}: Conceptualization, formulation, structural and theoretical derivations, and manuscript preparation.

\section*{Acknowledgments}

S. Biswas gratefully acknowledges the support received through the Junior and Senior Research Fellowships provided by the Ministry of Human Resource Development (MHRD), Government of India, and IIT Kharagpur. Both authors sincerely thank the anonymous referees for their valuable comments and constructive suggestions, which substantially improved the manuscript.

\section*{Funding}

 Author B. Banerjee would like to thank the Science and Engineering Research Board (SERB), Department of Science \& Technology, Government of India, for the MATRICS grant (File number MTR/2021/000397)  for the funding of the project.

\section*{Conflict of Interest} The authors declare no conflicts of interest.

\section*{Data and Code Availability Statement.}
The Bitcoin, Ethereum, and Gold datasets analysed in this study are
publicly available from the sources cited in the empirical-analysis
section. The R code used to implement the transformation,
calculate the SAMC statistic, generate finite-grid critical values, conduct
the simulation studies, reproduce the local-power comparisons, and perform
the empirical analyses is available at
\texttt{}{\url{INSERT-PERMANENT-REPOSITORY-LINK}}. The processed data can be obtained from the corresponding author.





\bibliographystyle{apalike}
\bibliography{buddha_bib}

\section{Appendix}

\subsection{Proof of Lemma-\ref{lemma1}} \label{lemma proof}
From \cite{biswas2025semi}, we can see that 
\begin{eqnarray}
    	A_1&= rR~\phi\left[\theta +\frac{r}{R}  \sin{\theta} \right]. \label{area first segment}  \\
        A_2&=rR~[2\pi-\phi]\left[\theta +\frac{r}{R} \sin{\theta} \right]. \label{area second segment} \\
A_3&= rR~\phi\left[(2\pi-\theta)-\frac{r}{R} \sin{\theta} \right].\label{area third segment} \\
A_4&=rR \left[2\pi-\phi\right]\left[(2\pi-\theta)-\frac{r}{R} \sin{\theta} \right].
			\label{area fourth segment}
\end{eqnarray}

Using the Equation-\ref{area first segment},\ref{area second segment},\ref{area third segment} \& \ref{area fourth segment}, for $\phi=\theta$ and $\frac{r}{R}=1$ we get 
$$A_1= rR~\theta\left[\theta +  \sin{\theta} \right], A_2=rR~[2\pi-\theta]\left[\theta + \sin{\theta} \right],$$
$$A_3= rR~\theta\left[(2\pi-\theta)- \sin{\theta} \right], \mbox{~~and~~} A_4 =rR \left[2\pi-\theta\right]\left[(2\pi-\theta)- \sin{\theta} \right], \mbox{~~respectively.} $$

Now we prove that $A_1<A_2, A_1<A_3, A_1<A_4$.
Let us consider two separate cases.\\

\noindent\textbf{Claim:}
For $0 < \theta \le \pi$, $A_1$ is minimal; for $\pi < \theta \le 2\pi$, $A_4$ is minimal.

\medskip
\noindent\textbf{Case 1: $0 < \theta \le \pi$}

\begin{enumerate}
\item \emph{Compare $A_1$ and $A_2$:}
\[
\frac{1}{rR}(A_1 - A_2) = \theta(\theta + \sin\theta) - (2\pi - \theta)(\theta + \sin\theta)
= (2\theta - 2\pi)(\theta + \sin\theta).
\]
For $0 < \theta < \pi$, we have $2\theta - 2\pi < 0$ and $\theta + \sin\theta > 0$, hence $A_1 - A_2 < 0$.

\item \emph{Compare $A_1$ and $A_3$:}
\[
\begin{aligned}
\frac{1}{rR}(A_1 - A_3) &= \theta(\theta + \sin\theta) 
- \theta\left[(2\pi - \theta) - \sin\theta\right] \\
&= \theta\left[ (\theta + \sin\theta) - (2\pi - \theta - \sin\theta) \right] \\
&= \theta\left[ 2\theta + 2\sin\theta - 2\pi \right] 
= 2\theta \left[ \theta + \sin\theta - \pi \right].
\end{aligned}
\]
Since $\theta + \sin\theta < \pi$ in this range, it follows that $A_1 - A_3 < 0$.

\item \emph{Compare $A_1$ and $A_4$:}
\[
\begin{aligned}
\frac{1}{rR}(A_1 - A_4 )
&= \theta(\theta + \sin\theta) - (2\pi - \theta)\left[(2\pi - \theta) - \sin\theta\right] \\
&= \theta^2 + \theta\sin\theta - (2\pi - \theta)^2 + (2\pi - \theta)\sin\theta \\
&= \theta^2 - (2\pi - \theta)^2 + 2\pi\sin\theta \\
&= \left[ \theta^2 - \left(4\pi^2 - 4\pi\theta + \theta^2\right) \right] + 2\pi\sin\theta \\
&= 4\pi(\theta - \pi) + 2\pi\sin\theta \\
&= 2\pi\left[ 2(\theta - \pi) + \sin\theta \right].
\end{aligned}
\]
For $0 < \theta \le \pi$, the bracketed term is non-positive, hence $A_1 - A_4 \le 0$.
\end{enumerate}
Thus, $A_1$ is the minimum for $0 < \theta \le \pi$.

\noindent
\textbf{Case-2:} When $\pi<\theta\leq 2\pi.$
Similar to  Case-1, we can prove that $A_4<A_1, A_4<A_2, A_4<A_3.$ Hence, it is omitted for the sake of brevity.\\

So, we see that when $0<\theta\leq \pi,$ $A_1=\theta\left[\theta +  \sin{\theta} \right]$ is minimum, and when $\pi<\theta\leq 2\pi,$ $A_4=(2\pi-\theta)\left[2\pi-\theta- \sin{\theta}\right]=(2\pi-\theta)\left[(2\pi-\theta)+\sin{(2\pi-\theta)} \right]$ is minimum.
This completes the proof.

\subsection{\textcolor{black}{Mathematical Derivation of Remark-\ref{rmk: info_preserv}}}\label{info_preserv}
Let
$A=g(\Theta)
=
\operatorname{sgn}(\Theta)A_C^{(0)}(\Theta).$ We quantify the information-loss ratio
$1-\frac{I_A(\mu)}{I_{\mathrm{full}}(\mu)}.$ Here we show that this ratio is zero when \(I_A(\mu)\) denotes the Fisher
information contained in the complete transformed observation \(A=g(\Theta)\).
Now the transformation can be written as
\[
g(\theta)
=
\begin{cases}
\dfrac{\theta(\theta+\sin\theta)}{(2\pi)^2},
&0\leq\theta\leq\pi,\\[3mm]
-\dfrac{(2\pi-\theta)\{2\pi-(\theta+\sin\theta)\}}
{(2\pi)^2},
&\pi<\theta<2\pi.
\end{cases}
\]
On \([0,\pi]\), \(g(\theta)\) is strictly increasing from \(0\) to
\(1/4\). On \((\pi,2\pi)\), it is strictly increasing from values
arbitrarily close to \(-1/4\) to values arbitrarily close to \(0\).
The two branch ranges are disjoint. Consequently, \(g\) is one-to-one
on \([0,2\pi)\), apart from immaterial endpoint issues of probability
zero, and possesses a branch-wise inverse on $(-1/4,0)\cup[0,1/4].$
Since \(g\) does not depend on the unknown parameter \(\mu\), the density
of \(A=g(\Theta)\) is, on each branch,
\[
f_A(a;\mu)
=
f_\Theta(g^{-1}(a);\mu)
\left|
\frac{d}{da}g^{-1}(a)
\right|.
\]
The Jacobian factor is independent of \(\mu\). Therefore,
\[
\frac{\partial}{\partial\mu}
\log f_A(a;\mu)
=
\frac{\partial}{\partial\mu}
\log f_\Theta(g^{-1}(a);\mu).
\]
It follows that
$I_A(\mu)
=
E\left[
\left\{
\frac{\partial}{\partial\mu}
\log f_A(A;\mu)
\right\}^{2}
\right]
=
E\left[
\left\{
\frac{\partial}{\partial\mu}
\log f_\Theta(\Theta;\mu)
\right\}^{2}
\right]
=
I_{\mathrm{full}}(\mu).$ For the von Mises model with known concentration parameter \(\kappa\),
\[
f_\Theta(\theta;\mu,\kappa)
=
\frac{\exp\{\kappa\cos(\theta-\mu)\}}
{2\pi I_0(\kappa)},~~~~ \mbox{and}~~\frac{\partial}{\partial\mu}
\log f_\Theta(\theta;\mu,\kappa)
=
\kappa\sin(\theta-\mu).
\]

Hence,
$I_{\mathrm{full}}(\mu)
=
\kappa\frac{I_1(\kappa)}{I_0(\kappa)}.$ For a sample of size \(n\), the total Fisher information is
\[
I_{\mathrm{full},n}(\mu)
=
n\kappa\frac{I_1(\kappa)}{I_0(\kappa)}.
\]
The von Mises sufficient statistic
$\left(
\sum_{i=1}^n\cos\Theta_i,\,
\sum_{i=1}^n\sin\Theta_i
\right)$
contains this full information. Since \(g\) is invertible and independent
of \(\mu\), the complete transformed sample
$\bigl(g(\Theta_1),\ldots,g(\Theta_n)\bigr)$
contains the same information
$I_{g(\Theta_1),\ldots,g(\Theta_n)}(\mu)
=
I_{\mathrm{full},n}(\mu).$
Thus, the Fisher-information loss caused solely by the transformation is
$1-
\frac{
I_{g(\Theta_1),\ldots,g(\Theta_n)}(\mu)
}{
I_{\mathrm{full},n}(\mu)
}
=0.$

This information-preservation result concerns the complete transformed
sample and should not be interpreted as showing that the SAMC statistic is
fully efficient. SAMC does not use the complete transformed likelihood;
instead, it uses standardized partial sums of \(g(\Theta_i)\). Therefore,
any efficiency loss arises from the moment-based CUSUM aggregation rather
than from the geometric transformation.



\subsection{Proof of Lemma-\ref{bbridge_lemma}}
\label{null_proof}
\begin{proof}
Let $b_i=a_i-E(a_i)$, then clearly, $E(b_i)=0$, and $\sigma_{b}^2=\text{Var~}(b_i)<\infty,$ for $i=1,\cdots,n.$ The estimated variance for $b_i$'s is 
$$\sigma_{b}^2=Var(b_1)\widehat{=}\frac{1}{n-1} \sum_{i=1}^{n}\left(b_i-\bar{b}\right)^2=\widehat\sigma_b^2.$$ Now construct the CUSUM process as

\begin{eqnarray}
       T_{b}(k)&=&\frac{1}{\sqrt{n}~\widehat\sigma_b }\left[ \sum_{i=1}^{k} b_i-k\bar{b}   \right] \mbox{~~for all~~} k=1,\ldots,n \nonumber\\ 
       &=&\frac{\sigma_b}{\widehat\sigma_b} \frac{1}{\sqrt{n}~\sigma_b }\left[ \sum_{i=1}^{k} b_i-k\bar{b}   \right]\\ \nonumber 
       \label{cusum process zero}
\end{eqnarray}
Let us consider $u \in (0,1)$, and denote $k= \lfloor{nu}\rfloor$. Hence, from Equation-\ref{cusum process zero} we can write

\begin{equation}
  T_{b}(k)=  T_{b}(\lfloor{nu}\rfloor)= \frac{1}{\sqrt{n~} \sigma_{b}} \left[ \sum_{i=1}^{\lfloor{nu}\rfloor} b_i-u \sum_{i=1}^n b_i   \right].
    \label{concentration_asym_zero}
\end{equation}

Now using Donsker's theorem \cite[see][Ch. 16]{billingsley2013convergence}  and Slutsky's theorem \cite[see][Ch. 9]{athreya2006measure} we can write 
\begin{eqnarray}
    \frac{1}{\sqrt{n~} \sigma_{b}} \sum_{i=1}^{\lfloor{nu}\rfloor} b_i &\implies& W(u)   \text{~~~~~for all~~~~} 0<u\leq1,
\end{eqnarray}
where `$\implies$' represents weak convergence, and $W(u)$ is the Wiener processes on $[0,1]$. Hence, Equation-\ref{concentration_asym_zero} becomes
\begin{equation}
  T_{b}(k)\implies  W(u)-u~W(1)=B_0(u),
    \label{bbridge_zero}
\end{equation}
where, $B_0(u)$ is standard Brownian bridge. Noting that $\widehat\sigma_{b}^2=\widehat\sigma_{a}^2$ and $T_{b}(k)=T(k).$
It is immediate that
\begin{equation}
  T(k)\implies B_0(u) \text{~~~~~~under~~} H_0.
\end{equation}
Hence, the lemma follows.
\end{proof}

\subsection{Proof of Theorem-\ref{loc_consistancy}}
\label{loc_consistancy_proof}

\begin{proof}
Let us consider $k\leq k^{*},$ where $k^*$ is the true but unknown location of the changepoint. Also, denote $m_1 =E(a_i)$ when $i\leq k^*$, $m_2 =E(a_i)$ when $i> k^*$, and $\Delta=(m_1-m_2)$, using Equation-\ref{square_area}. Denote $b_i=a_i-E(a_i)$ for $i= 1, \cdots,n,$ such that $E(b_i) = 0$ and $\text{Var}(b_i) = \sigma_b^2 < \infty$. Then we can write the partial sum processes as
\begin{eqnarray}
\frac{1}{\sqrt{n}}\left[ \sum_{i=1}^{k} a_i - k\bar{a} \right] 
&=& \frac{1}{\sqrt{n} }\left[ \sum_{i=1}^{k} b_i-k\bar{b} \right]+ \frac{k(n-k^*)}{n^{3/2}} \Delta
\label{before_cp}
\end{eqnarray}

Similarly, for $k >k^*$ one can write the following:
\begin{eqnarray}
 \frac{1}{\sqrt{n} }\left[ \sum_{i=1}^{k} a_i-k\bar{a} \right]
&=& \frac{1}{\sqrt{n}}\left[ \sum_{i=1}^{k} b_i - k\bar{b} \right]+ \frac{k^*(n-k)}{n^{3/2}} \Delta
\label{after_cp}
\end{eqnarray}

Let $k = \lfloor n u \rfloor$ with $u \in (0,1)$ and using Equation-\ref{before_cp} and \eqref{after_cp}, consider:
\begin{align}
T_n(u) 
&= \frac{\sigma_b}{\widehat \sigma_a} \left\{ \frac{1}{\sqrt{n} \sigma_b} \left[ \sum_{i=1}^{\lfloor n u \rfloor} b_i - u \sum_{i=1}^n b_i \right] 
 + \frac{\sqrt{n} \, \min\{u,u^*\} (1 - \max\{u,u^*\}) \Delta}{\sigma_b} \right\}.
\label{eq:Un_decomp}
\end{align}

\noindent
It can be easily shown that under $H_1$, 
$$\widehat \sigma_a^2 \xrightarrow{p} \sigma_b^2 + [u^*(1-u^*)]\Delta^2 \implies \frac{\sigma_b}{\widehat \sigma_a} = O_p(1) .$$

\noindent
By the functional CLT, the process $\displaystyle\frac{1}{\sqrt{n} \sigma_b} \left[ \sum_{i=1}^{\lfloor n u \rfloor} b_i - u \sum_{i=1}^n b_i \right]$ converges weakly to a standard Brownian bridge $B_0(u)$ on $[0,1]$, which is of $O_p(1)$ uniformly in $u$. From \eqref{eq:Un_decomp} and the above orders, uniformly in $u$:
\begin{eqnarray}
    T_n(u) = \sqrt{n} \, c_*(u) + O_p(1),
    \label{consistency_order}
\end{eqnarray}
\textcolor{black}{where $c_*(u) = \frac{\min\{u,u^*\} (1 - \max\{u,u^*\}) \Delta}{\sigma_b}$. Since $\Delta \neq 0$, $c_*(u)$ is uniquely maximized at $u = u^*$. The first term of Equation-\ref{consistency_order}, that is $\sqrt{n}c_*(u)$, dominates the location of the maxima.  
By the argmax continuous mapping theorem \cite[see][]{ferger2004continuous},
\[
\hat{u}^* := \frac{\hat{k}^*}{n} \xrightarrow{p} u^*,
\]
assuring consistency in proportion. }
\end{proof}

\subsection{Proof of Corollary-\ref{consistency_corr}}
\label{alt_consistency_proof}

\begin{proof}
Let $k^*$ be the true location of the changepoint. So, we evaluate  $T_n$ from Equation-\ref{consistency_order} at $u=u^*$ as $k^* = \lfloor n u^* \rfloor$, we have
$$T_n(u^*) \;=\; \sqrt{n}\,c_*(u^*) \;+\; O_p(1) \;\xrightarrow{p}\; +\infty.$$
because
as $\sqrt{n}\to\infty$, $\sqrt{n}\,c_*(u^*)\to\infty$ (since $c_*(u^*)>0$ and ).
Consequently, for any fixed finite threshold $k_\alpha$,
\[
\mathbb{P}_{H_1}\big(U_n(u^*) \le k_\alpha\big) \longrightarrow 0.
\]
Since $\mathbb{M}_n=\max_{u}|U_n(u)| \ge |U_n(u^*)|$, it follows that
\[
\mathbb{P}_{H_1}\big(\mathbb{M}_n \le k_\alpha\big)\le \mathbb{P}_{H_1}\big(|U_n(u^*)|\le k_\alpha\big)\longrightarrow 0.
\]
Thus, the Type-II error tends to $0$ and the power tends to $1$. 
\end{proof}

\RestyleAlgo{ruled}
\SetKwComment{Comment}{/* }{ */}
\begin{algorithm}[H]
\caption{Computation of finite-sample cut-off from Kolmogorov distribution, $K^{(n)}_\infty$.}
\label{alg:algo_cutoff}

\KwIn{
$n, N \in \mathbb{N}$  \Comment*[r]{Segment size, Number of simulation } \\
\hspace{0.8cm}$\alpha \in (0,1)$  \Comment*[r]{ Significance level} 
Initialize $\mathcal{L} \gets $ an empty list 

}

\For{$i = 1$ \textbf{to} $N$}{
     Simulate null data: $z_1, \dots, z_n \sim \mathcal{N}(\mu, \sigma^2)$ \\
     Center the data: $z_t' \gets z_t - \bar{z}$\\
     Compute cumulative sum statistic:
          \[
          B_k \gets \frac{\left| \sum_{t=1}^k z_t' \right|}{\sqrt{n} \cdot \hat{\sigma}}, \quad k = 1, \dots, n
          \]
     Record test value: $\mathcal{L}_i \gets \displaystyle \max_k B_k$
}
Compute threshold: $k_\alpha$ $\gets$ empirical $(1-\alpha)$-quantile of $\mathcal{L}$

\KwResult{$k_\alpha:$ The cutoff value at $(1-\alpha)$ quantile.}
\end{algorithm}

\RestyleAlgo{ruled}
\SetKwComment{Comment}{/* }{ */}
\begin{algorithm}[H]
\caption{Power calculation algorithm for detecting a changepoint in the mean direction}
\label{alg:algo_power_kappa}

\KwIn{
$n \in \mathbb{N}$ \Comment*[r]{Segment size} 
$I \in \mathbb{N}$ \Comment*[r]{Number of iterations (large for accuracy)} 
$R = r = 1$  \Comment*[r]{Radii of the unit circle} 
$cp \in \{1,2,\ldots, n-1\}$ \Comment*[r]{Location of true changepoint} 
$\kappa \in (0, \infty)$ \Comment*[r]{Concentration parameter for von Mises dist.} 
$\mu_p, \mu_0 \in [0, 2\pi)$ \Comment*[r]{Mean shift and initial mean}
}

\For{$j=1$ \KwTo $I$}{
    Generate $\theta_{1},\ldots,\theta_{cp} \stackrel{i.i.d.}{\sim} f(\theta;\mu_0,\kappa)$\; 
    Generate $\theta_{cp+1},\ldots,\theta_{n} \stackrel{i.i.d.}{\sim} f(\theta;(\mu_0+\mu_p) \bmod 2\pi,\kappa)$\;

    Compute $sgn(\theta) \gets \begin{cases}
			1& \text{~if~~} 0\leq \theta \leq \pi\\
			-1 & \text{~if~} \theta>\pi,
		\end{cases}$\\
        
    Compute $a_i \gets sgn(\theta) \cdot A_C^{(0)}(\theta_i)$ for $i=1,\ldots,n$\;

    Estimate variance: $\widehat{\mathrm{Var}}(a) \gets \frac{1}{n-1} \sum_{i=1}^n (a_i - \bar{a})^2$, where $n\bar{a} = \sum_{i=1}^n a_i$\;

    \For{$k=1$ \KwTo $n$}{
        $T(k) \gets \left( n \, \widehat{\mathrm{Var}}(a) \right)^{-1/2} \left( \sum_{i=1}^k a_i - k \bar{a} \right)$\;
    }

    Compute $\mathbb{M}_n \gets \max_{1 \le k < n} |T(k)|$\;
    Record $Re[j] \gets \begin{cases}
			1& \text{~if~} \mathbb{M}_n > k_{\alpha}\\
			0 & \text{~if~} \mathbb{M}_n \leq k_{\alpha},
		\end{cases}$   \Comment*[r]{$k_{\alpha}$ from Algorithm-\ref{alg:algo_cutoff}} 
}
Compute empirical power: $power \gets \frac{1}{I} \sum_{j=1}^I Re[j]$\;

\KwResult{Vary $\mu_p$ while keeping $\mu_0$ fixed to obtain the power curve.}
\end{algorithm}
\begin{rmk}[Computational complexity]
\textcolor{black}{In the empirical analysis, SAMC is first applied to the complete sequence.
If a significant changepoint is detected, the two resulting subsegments are
examined once as a post-detection stability check. In all reported cases,
these second-stage tests were nonsignificant, so no further recursion was
performed.}

\textcolor{black}{For a segment of length \(m\), computation of the transformed observations,
sample mean, variance, and all CUSUM values requires \(\mathcal{O}(m)\)
operations. Hence, a single SAMC scan has complexity \(\mathcal{O}(m)\).
For the reported empirical procedure, the total cost is
$n+k+(n-k)=2n,$
and therefore remains
$\mathcal{O}(n).$}

\textcolor{black}{If the method were extended to a fully recursive binary-segmentation
procedure, the computational cost would typically be
\(\mathcal{O}(n\log n)\) under approximately balanced splits, while a
highly unbalanced worst case could reach \(\mathcal{O}(n^2)\).}

\textcolor{black}{We also clarify that the changepoint analysis is not performed on the raw
one-minute financial records. These data are first reduced to one angular
timestamp per day, resulting in sequence lengths of \(1860\) for Bitcoin,
\(1409\) for Ethereum, and \(1863\) for Gold.}

\textcolor{black}{For a general multiple-changepoint extension, scalable procedures such as
PELT, FPOP, or wild binary segmentation could be considered. However, such
extensions require additional penalization and theoretical development and
are beyond the single-changepoint scope of the present paper. }
\end{rmk}

\RestyleAlgo{ruled}
\SetKwComment{Comment}{/* }{ */}

\begin{algorithm}[H]
\caption{\textcolor{black}{SAMC test for a single changepoint in an observed angular sequence}}
\label{alg:samc_observed}

\KwIn{
$\theta_1,\ldots,\theta_n\in[0,2\pi)$
\Comment*[r]{Observed angular sequence}

$\alpha\in(0,1)$
\Comment*[r]{Significance level}

$K_{\infty}^{(n)}$
\Comment*[r]{Finite-grid critical value from Algorithm-1}
}

\If{$n<2$}{
    Stop and report that at least two observations are required;
}

Compute
\[
\operatorname{sgn}(\theta_i)
\gets
\begin{cases}
 1,  & 0\leq\theta_i\leq\pi,\\
-1,  & \pi<\theta_i<2\pi,
\end{cases}
\quad i=1,\ldots,n;
\]

Compute
\[
a_i
\gets
\operatorname{sgn}(\theta_i)
A_C^{(0)}(\theta_i),
\quad i=1,\ldots,n;
\]

Compute
\[
\bar a
\gets
\frac{1}{n}\sum_{i=1}^{n}a_i;
\]

Compute
\[
\widehat{\sigma}_a^2
\gets
\frac{1}{n-1}
\sum_{i=1}^{n}(a_i-\bar a)^2;
\]

\If{$\widehat{\sigma}_a^2=0$}{
    Stop and report that the statistic is undefined for a degenerate
    transformed sequence;
}

\For{$k=1$ \KwTo $n-1$}{
    Compute
    \[
    T(k)
    \gets
    \frac{1}{\sqrt{n}\widehat{\sigma}_a}
    \left(
    \sum_{i=1}^{k}a_i-k\bar a
    \right);
    \]
}

Compute
\[
\mathbb{M}_n
\gets
\max_{1\leq k<n}|T(k)|;
\]

Compute
\[
\widehat{k}
\gets
\min\left\{
k\in\{1,\ldots,n-1\}:
|T(k)|=\mathbb{M}_n
\right\};
\]
\Comment*[r]{The smallest maximizer is used in the event of a tie}

\eIf{$\mathbb{M}_n>K_{\infty}^{(n)}$}{
    Reject $H_0$ and report $\widehat{k}$ as the estimated changepoint;
}{
    Do not reject $H_0$ and report no significant changepoint;
}

\KwResult{
The test decision and, upon rejection of \(H_0\), the estimated
single-changepoint location \(\widehat{k}\). The algorithm then terminates.
}

\end{algorithm}
\textcolor{black}{Algorithm-\ref{alg:samc_observed} describes the operational application of
SAMC to an observed angular sequence. The procedure is formulated for one
unknown changepoint and terminates after reporting either one estimated
changepoint or no significant changepoint. Therefore, no recursive stopping
criterion or minimum-segment-length tuning parameter is involved. Algorithms-\ref{alg:algo_cutoff} and \ref{alg:algo_power_kappa}, respectively, concern the calculation of the finite-grid critical
value and the Monte Carlo evaluation of power, whereas
Algorithm-\ref{alg:samc_observed} provides the observed-data implementation.}

\begin{figure}[htbp]
    \centering
    \includegraphics [trim= 0 0 0 40, clip, width=0.9\textwidth]
    {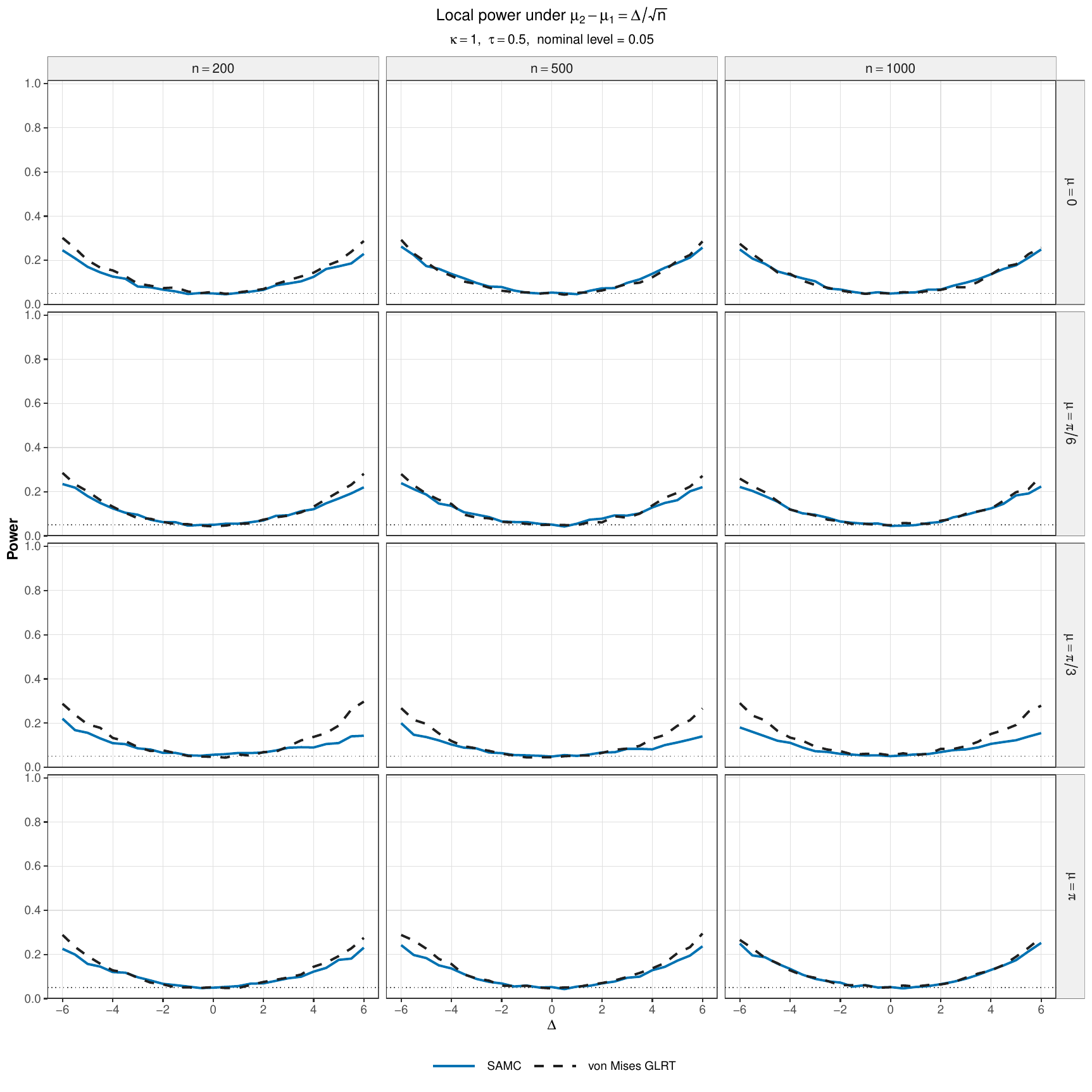}
   \textcolor{black}{ \caption{
    Empirical local-power curves at the \(5\%\) significance level  for \(\kappa=1\).
    The columns correspond to
    \(n=200,500,1000\), and the rows correspond to
    \(\mu=0,\pi/6,\pi/3,\pi\).
    The solid line represents the power curve of SAMC, and the dashed line represents the same for
    known-concentration von Mises GLRT. The horizontal dotted line marks
    the nominal significance level \(0.05\).
    }
    \label{fig:local_power_kappa1_all_n}}
\end{figure}

\begin{figure}[h!]
    \centering
    \includegraphics[trim= 0 0 0 40, clip,width=0.9\textwidth]
    {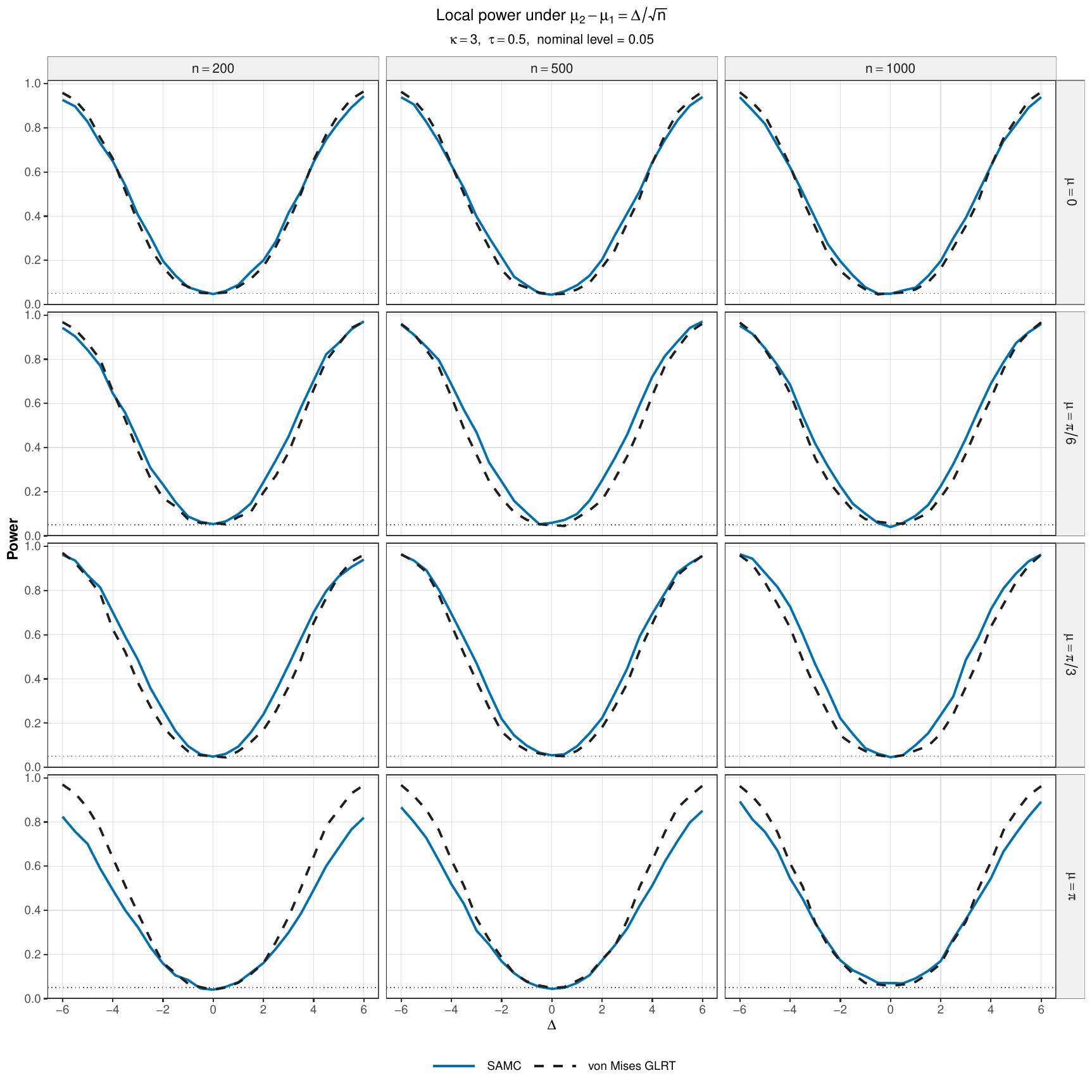}
   \textcolor{black}{ \caption{
  Empirical local-power curves at the \(5\%\) significance level  for \(\kappa=3\).
    The columns correspond to
    \(n=200,500,1000\), and the rows correspond to
    \(\mu=0,\pi/6,\pi/3,\pi\).
    The solid line represents the power curve SAMC, and the dashed line represents the same for
    known-concentration von Mises GLRT. The horizontal dotted line marks
    the nominal significance level \(0.05\).
    }
    \label{fig:local_power_kappa3_all_n}}
\end{figure}

\begin{table}[h!]
\centering
\resizebox{0.55\linewidth}{!}{%
  \begin{tabular}{|cc|c|c|c|c|l|}
\hline
& & \multicolumn{3}{ |c| }{Quantiles} \\ \hline
& & 0.90th & 0.95th & 0.99th  \\ \hline
\multicolumn{1}{ |c  }{\multirow{2}{*}{$n=50$} } &
\multicolumn{1}{ |c| }{$\kappa=0.5$} & 1.1131& 1.2425& 1.4328   \\ 
\multicolumn{1}{ |c  }{}                        &
\multicolumn{1}{ |c| }{$\kappa=1$}  & 1.1224& 1.2465& 1.4772    \\ 
\multicolumn{1}{ |c  }{}                        &
\multicolumn{1}{ |c| }{$\kappa=1.5$} & 1.1103& 1.2294& 1.4557   \\ 
\multicolumn{1}{ |c  }{}                        &
\multicolumn{1}{ |c| }{$\kappa=2$}  & 1.1002& 1.2101& 1.4089    \\ 
\multicolumn{1}{ |c  }{}                        &
\multicolumn{1}{ |c| }{$\kappa=4$} & 1.1289 & 1.2449& 1.4727   \\ 
\multicolumn{1}{ |c  }{}                        &
\multicolumn{1}{ |c| }{$\kappa=10$} & 1.1313& 1.2550& 1.4574    \\ 
\multicolumn{1}{ |c  }{}                        &
\multicolumn{1}{ |c| }{$K^{(50)}_\infty$} & 1.1456& 1.2657& 1.4815     \\ \hline

\multicolumn{1}{ |c  }{\multirow{2}{*}{$n=100$} } &
\multicolumn{1}{ |c| }{$\kappa=0.5$} & 1.1558& 1.2890& 1.5648   \\ 
\multicolumn{1}{ |c  }{}                        &
\multicolumn{1}{ |c| }{$\kappa=1$}  & 1.1625& 1.2845& 1.5710    \\ 
\multicolumn{1}{ |c  }{}                        &
\multicolumn{1}{ |c| }{$\kappa=1.5$} &1.1469& 1.2598& 1.5308    \\ 
\multicolumn{1}{ |c  }{}                        &
\multicolumn{1}{ |c| }{$\kappa=2$}  & 1.1416& 1.2614& 1.5107   \\ 
\multicolumn{1}{ |c  }{}                        &
\multicolumn{1}{ |c| }{$\kappa=4$} & 1.1423& 1.2623& 1.5206   \\ 
\multicolumn{1}{ |c  }{}                        &
\multicolumn{1}{ |c| }{$\kappa=10$} & 1.1674& 1.2928& 1.5510    \\ 
\multicolumn{1}{ |c  }{}                        &
\multicolumn{1}{ |c| }{$K^{(100)}_\infty$} & 1.1711& 1.3161& 1.5662     \\ \hline

\multicolumn{1}{ |c  }{\multirow{2}{*}{$n=200$} } &
\multicolumn{1}{ |c| }{$\kappa=0.5$} & 1.1769& 1.3136& 1.5861     \\ 
\multicolumn{1}{ |c  }{}                        &
\multicolumn{1}{ |c| }{$\kappa=1$}  & 1.1902& 1.3176& 1.5896   \\ 
\multicolumn{1}{ |c  }{}                        &
\multicolumn{1}{ |c| }{$\kappa=1.5$} &1.1557& 1.2958& 1.5344  \\ 
\multicolumn{1}{ |c  }{}                        &
\multicolumn{1}{ |c| }{$\kappa=2$}  & 1.1771& 1.3136& 1.5707   \\ 
\multicolumn{1}{ |c  }{}                        &
\multicolumn{1}{ |c| }{$\kappa=4$} & 1.1663& 1.3037& 1.5961   \\ 
\multicolumn{1}{ |c  }{}                        &
\multicolumn{1}{ |c| }{$\kappa=10$} & 1.1940& 1.3289& 1.5638   \\ 
\multicolumn{1}{ |c  }{}                        &
\multicolumn{1}{ |c| }{$K^{(200)}_\infty$} & 1.1739& 1.2966& 1.5604  \\ \hline

\multicolumn{1}{ |c  }{\multirow{2}{*}{$n=500$} } &
\multicolumn{1}{ |c| }{$\kappa=0.5$} & 1.2015& 1.3381& 1.6110     \\ 
\multicolumn{1}{ |c  }{}                        &
\multicolumn{1}{ |c| }{$\kappa=1$}  & 1.2036& 1.3292& 1.5783      \\ 
\multicolumn{1}{ |c  }{}                        &
\multicolumn{1}{ |c| }{$\kappa=1.5$} & 1.1901& 1.3031& 1.5503  \\ 
\multicolumn{1}{ |c  }{}                        &
\multicolumn{1}{ |c| }{$\kappa=2$}  & 1.1907& 1.3244& 1.5603   \\ 
\multicolumn{1}{ |c  }{}                        &
\multicolumn{1}{ |c| }{$\kappa=4$} & 1.1785& 1.3095& 1.5480   \\ 
\multicolumn{1}{ |c  }{}                        &
\multicolumn{1}{ |c| }{$\kappa=10$} & 1.1834& 1.3233& 1.5843   \\ 
\multicolumn{1}{ |c  }{}                        &
\multicolumn{1}{ |c| }{$K^{(500)}_\infty$} &1.1999& 1.3323& 1.5575  \\ \hline

\multicolumn{1}{ |c  }{\multirow{2}{*}{$n=1000$} } &
\multicolumn{1}{ |c| }{$\kappa=0.5$} & 1.2115& 1.3449& 1.6170   \\ 
\multicolumn{1}{ |c  }{}                        &
\multicolumn{1}{ |c| }{$\kappa=1$}  & 1.2135& 1.3494& 1.5827   \\ 
\multicolumn{1}{ |c  }{}                        &
\multicolumn{1}{ |c| }{$\kappa=1.5$} & 1.1992& 1.3191& 1.5551  \\ 
\multicolumn{1}{ |c  }{}                        &
\multicolumn{1}{ |c| }{$\kappa=2$}  & 1.2011& 1.3227& 1.5630   \\ 
\multicolumn{1}{ |c  }{}                        &
\multicolumn{1}{ |c| }{$\kappa=4$} & 1.1832& 1.3117& 1.5499   \\ 
\multicolumn{1}{ |c  }{}                        &
\multicolumn{1}{ |c| }{$\kappa=10$} & 1.1990& 1.3365& 1.5710   \\ 
\multicolumn{1}{ |c  }{}                        &
\multicolumn{1}{ |c| }{$K^{(1000)}_\infty$} &1.2029& 1.3471& 1.5750   \\ \hline

\end{tabular}}
\vspace{0.3cm}
\caption{Table of the quantile values   of the SAMC test statistic $\mathbb{M}_n$ under the null hypothesis, $H_{0}$ when the sample of sizes of $50, 100,200, 500,$ and $1000$ are drawn from von Mises distribution with mean direction $\mu=0,$ and different concentration parameters $\kappa=0.5,1, 1.5,2,4$ and $10$. The table also contains the cut-off values from the limiting distribution of $K_\infty^{(n)}$ (Equation-\ref{bbridge}) using Algorithm-\ref{alg:algo_cutoff} with the grid size of $n=50,100,200, 500,$ and $1000$, respectively.}
\label{table: null cut-off table concentration change}
\end{table}

\begin{table}[h!]
\centering
\resizebox{0.55\linewidth}{!}{%
  \begin{tabular}{|cc|c|c|c|c|l|}
\hline
& & \multicolumn{3}{ c| }{Quantiles} \\ \hline
& & 0.90th & 0.95th & 0.99th  \\ \hline
\multicolumn{1}{ |c  }{\multirow{2}{*}{$n=50$} } &

\multicolumn{1}{ |c| }{$\mu=\frac{\pi}{4}$} & 1.1181& 1.2388& 1.4665    \\ 
\multicolumn{1}{ |c  }{}                        &
\multicolumn{1}{ |c| }{$\mu=\frac{2\pi}{3}$}  & 1.1056 & 1.2148& 1.4161    \\ 
\multicolumn{1}{ |c  }{}                        &
\multicolumn{1}{ |c| }{$\mu=\frac{4\pi}{3}$} & 1.0889& 1.2072& 1.4354   \\ 
\multicolumn{1}{ |c  }{}                        &
\multicolumn{1}{ |c| }{$\mu=\frac{8\pi}{5}$}  & 1.1199& 1.2529 & 1.4631    \\ 
\multicolumn{1}{ |c  }{}                        &
\multicolumn{1}{ |c| }{$K^{(50)}_\infty$} & 1.1229 &1.2384& 1.4658     \\ 
\hline
\multicolumn{1}{ |c  }{\multirow{2}{*}{$n=100$} } &
\multicolumn{1}{ |c| }{$\mu=\frac{\pi}{4}$} & 1.1778& 1.3007& 1.5515    \\ \multicolumn{1}{ |c  }{}                        &
\multicolumn{1}{ |c| }{$\mu=\frac{2\pi}{3}$}  & 1.1344 &1.2493& 1.505    \\ 
\multicolumn{1}{ |c  }{}                        &
\multicolumn{1}{ |c| }{$\mu=\frac{4\pi}{3}$} & 1.1489& 1.2778& 1.5445   \\ 
\multicolumn{1}{ |c  }{}                        &
\multicolumn{1}{ |c| }{$\mu=\frac{8\pi}{5}$}  & 1.1668& 1.2913 &1.5468   \\ 
\multicolumn{1}{ |c  }{}                        &
\multicolumn{1}{ |c| }{$K^{(100)}_\infty$} & 1.1576& 1.2909& 1.5534     \\ 
\hline
\multicolumn{1}{ |c  }{\multirow{2}{*}{$n=200$} } &
\multicolumn{1}{ |c| }{$\mu=\frac{\pi}{4}$} & 1.1913& 1.3075& 1.5421  \\
\multicolumn{1}{ |c  }{}                        &
\multicolumn{1}{ |c| }{$\mu=\frac{2\pi}{3}$}  & 1.1857& 1.3051 &1.5633   \\
\multicolumn{1}{ |c  }{}                        &
\multicolumn{1}{ |c| }{$\mu=\frac{4\pi}{3}$} & 1.1582& 1.2854& 1.5279  \\ 
\multicolumn{1}{ |c  }{}                        &
\multicolumn{1}{ |c| }{$\mu=\frac{8\pi}{5}$}  & 1.1757& 1.2972& 1.5783    \\ 
\multicolumn{1}{ |c  }{}                        &
\multicolumn{1}{ |c| }{$K^{(200)}_\infty$} & 1.1817& 1.3135& 1.5625    \\
\hline
\multicolumn{1}{ |c  }{\multirow{2}{*}{$n=500$} } &
\multicolumn{1}{ |c| }{$\mu=\frac{\pi}{4}$} & 1.2045& 1.3504& 1.5993   \\ 
\multicolumn{1}{ |c  }{}                        &
\multicolumn{1}{ |c| }{$\mu=\frac{2\pi}{3}$}  & 1.1967& 1.3221& 1.5923    \\ 
\multicolumn{1}{ |c  }{}                        &
\multicolumn{1}{ |c| }{$\mu=\frac{4\pi}{3}$} & 1.1837& 1.3341& 1.6169   \\ 
\multicolumn{1}{ |c  }{}                        &
\multicolumn{1}{ |c| }{$\mu=\frac{8\pi}{5}$}  & 1.1967& 1.3242 &1.5762    \\ 
\multicolumn{1}{ |c  }{}                        &
\multicolumn{1}{ |c| }{$K^{(500)}_\infty$} & 1.1945 &1.3378 &1.5819     \\ 
\hline
\multicolumn{1}{ |c  }{\multirow{2}{*}{$n=1000$} } &
\multicolumn{1}{ |c| }{$\mu=\frac{\pi}{4}$} & 1.1999& 1.3437& 1.5953   \\ 
\multicolumn{1}{ |c  }{}                        &
\multicolumn{1}{ |c| }{$\mu=\frac{2\pi}{3}$}  & 1.2100& 1.3330 &1.6292    \\ 
\multicolumn{1}{ |c  }{}                        &
\multicolumn{1}{ |c| }{$\mu=\frac{4\pi}{3}$} & 1.2187& 1.3627& 1.6199  \\ 
\multicolumn{1}{ |c  }{}                        &
\multicolumn{1}{ |c| }{$\mu=\frac{8\pi}{5}$}  & 1.1891& 1.3212& 1.5845    \\ 
\multicolumn{1}{ |c  }{}                        &
\multicolumn{1}{ |c| }{$K^{(1000)}_\infty$} & 1.2042& 1.3507& 1.6105     \\ \hline

\end{tabular}}
\vspace{0.3cm}
\caption{Table of the quantile values  of the SAMC test statistic $\mathbb{M}_n$ under the null hypothesis, $H_{0}$ when the sample of sizes of $50, 100,200, 500,$ and $1000$ are drawn from von Mises distribution with fixed concentration parameter $\kappa=3,$ and different mean direction parameters $\mu=\frac{\pi}{4},\frac{2\pi}{3},\frac{4\pi}{3}$ and $\frac{8\pi}{5} $. The table also contains the cut-off values from the limiting distribution of $K_\infty^{(n)}$ (Equation-\ref{bbridge})  using Algorithm-\ref{alg:algo_cutoff} with the grid size of $n=50,100,200, 500,$ and $1000$, respectively.}
\label{table: null cut-off table mean change}
\end{table}

\begin{table}[h!]
	\centering
	\renewcommand{\arraystretch}{1} 
	\begin{tabular}{ |>{\centering\arraybackslash}p{2.5cm}|>{\centering\arraybackslash}p{2.5cm}|>{\centering\arraybackslash}p{2.5cm}|>{\centering\arraybackslash}p{2.5cm}| }
		\hline
		Data segment & Estimated location of changepoint & P-value &  Mean direction of significant segments in radians (degrees)\\
		\hline\hline
		1-1860 & 970 & 0.0004 &  \\
		1-970 & 840 & 0.2872 & 0.2280 (13.07) \\
		971-1860 &1490 & 0.6885 & 5.9629 (341.65)\\
		\hline
	\end{tabular}
	\vspace{0.3cm}
	\caption{ Detected changepoint in the Bitcoin data set.}
	\label{table:data_bitcoin_cp_table}
\end{table}

\begin{table}[h!]
	\centering
	\renewcommand{\arraystretch}{1} 
	\begin{tabular}{ |>{\centering\arraybackslash}p{2.5cm}|>{\centering\arraybackslash}p{2.5cm}|>{\centering\arraybackslash}p{2.5cm}|>{\centering\arraybackslash}p{2.5cm}| }
		\hline
		Data segment & Estimated location of changepoint & P-value &  Mean direction of significant segments in radians (degrees)\\
		\hline\hline
		1-1409 & 896 & 0.0004 &  \\
		1-896 & 314 & 0.0583 & 6.0168 (344.74) \\
		897-1409 & 976 & 0.4925 & 6.1040 (349.73) \\
		\hline
	\end{tabular}
	\vspace{0.3cm}
	\caption{Detected changepoint in the Ethereum data set.}
	\label{table:data_ethereum_cp_table}
\end{table}

\begin{table}[h!]
	\centering
	\renewcommand{\arraystretch}{1} 
	\begin{tabular}{ |>{\centering\arraybackslash}p{2.5cm}|>{\centering\arraybackslash}p{2.5cm}|>{\centering\arraybackslash}p{2.5cm}|>{\centering\arraybackslash}p{2.5cm}| }
		\hline
		Data segment & Estimated location of changepoint & P-value &  Mean direction of significant segments in radians (degrees)\\
		\hline\hline
		1-1863 & 1310 & 0.0004 &  \\
		1-1310 & 689 & 0.5070 & 1.5308 (87.70) \\
		1311-1863 & 1451 & 0.4642 & 87.7083(135.00) \\
		\hline
	\end{tabular}
	\vspace{0.3cm}
	\caption{Detected changepoint in the Gold price dataset for the timestamps associated with the lowest price.}
	\label{table:data_gold_cp_table_low}
\end{table}

\begin{table}[h!]
	\centering
	\renewcommand{\arraystretch}{1} 
	\begin{tabular}{ |>{\centering\arraybackslash}p{2.5cm}|>{\centering\arraybackslash}p{2.5cm}|>{\centering\arraybackslash}p{2.5cm}|>{\centering\arraybackslash}p{2.5cm}| }
		\hline
		Data segment & Estimated location of changepoint & P-value &  Mean direction of significant segments in radians (degrees)\\
		\hline\hline
		1-1863 & 1297 & 0.0181 &  \\
		1-1297 & 689 & 0.8524 & 1.6510 (94.58) \\
		1298-1867 & 1750 & 0.5044& 3.8513 (218.63) \\
		\hline
	\end{tabular}
	\vspace{0.3cm}
	\caption{Detected changepoint in the Gold price dataset for the timestamps associated with the highest price.}
	\label{table:data_gold_cp_table_high}
\end{table}

\newpage

\end{document}